\documentclass[preprint,12pt]{elsarticle}

\usepackage{amsmath}
\usepackage{amssymb}
\usepackage{amsthm}
\usepackage{amsfonts}
\usepackage{bm}
\usepackage{xcolor}
\usepackage{graphicx}
\graphicspath{{/home/jianfeng/Dropbox/HaizhaoYang/GrainBoundary/results/}}
\usepackage{enumerate}

\usepackage[caption=false,font=footnotesize]{subfig}

\usepackage{hyperref}

\theoremstyle{plain}
\newtheorem{theorem}{Theorem}

\theoremstyle{definition}

\theoremstyle{remark}

\DeclareMathOperator{\tr}{tr}

\DeclareMathOperator{\argmin}{arg min}

\DeclareMathOperator{\curl}{curl}
\DeclareMathOperator{\mass}{mass}
\newcommand{\curlDisc}{\curl^d}
\newcommand{\DeltaDisc}{\Delta^d}
\newcommand{\partialDisc}{\partial^d}

\newcommand{\ud}{\,\mathrm{d}}

\newcommand{\RR}{\mathbb{R}}

\newcommand{\TT}{\mathrm{T}}

\newcommand{\CC}{\mathbb{C}}

\IfFileExists{mathabx.sty}%
  {\DeclareFontFamily{U}{mathx}{\hyphenchar\font45}%
   \DeclareFontShape{U}{mathx}{m}{n}{<->mathx10}{}%
   \DeclareSymbolFont{mathx}{U}{mathx}{m}{n}%
   \DeclareFontSubstitution{U}{mathx}{m}{n}%
   \DeclareMathAccent{\widebar}{0}{mathx}{"73}%
}{%
  \PackageWarning{mathabx}{%
    Package mathabx not available, therefore\MessageBreak substituting
    widebar with overline\MessageBreak }%
  \newcommand{\widebar}[1]{\overline{#1}}%
}

\newcommand{\mc}[1]{\mathcal{#1}}

\newcommand{\abs}[1]{\lvert#1\rvert}

\newcommand{\norm}[1]{\left\lVert#1\right\rVert}

\renewcommand{\Re}{\mathfrak{Re}}

\newcommand{\jumpset}{{S}}                  
\newcommand{\pointgroup}{{P}}               
\newcommand{\proj}[1]{{\mathrm{proj}_{#1}}} 
\newcommand{\range}{{\mathrm{range}}}
\newcommand{\rngJ}{{L}}                     

\newcommand{\notinclude}[1]{}
\graphicspath{{./pic/},{../results/}}

\begin{document}
\begin{frontmatter}

\title{Combining $2D$ synchrosqueezed wave packet transform with optimization for crystal image analysis}

\author[1]{Jianfeng Lu}
\address[1]{Departments of Mathematics, Physics,
  and Chemistry, Duke University, Durham, NC 27708 USA}
  \ead{jianfeng@math.duke.edu}
  \author[2]{Benedikt Wirth}
\address[2]{Institute for Computational and Applied Mathematics
  and Cells-in-Motion Cluster of Excellence (EXC 1003 - CiM),
  University of M\"unster, D-48149 M\"unster, Germany}
  \ead{Benedikt.Wirth@uni-muenster.de}
  \author[3]{Haizhao Yang}
\address[3]{Department of Mathematics, Stanford University, Stanford,
  CA 94305 USA}
  \ead{haizhao@stanford.edu}
  

\begin{abstract}
  We develop a variational optimization method for crystal analysis
  in atomic resolution images, which uses information from a
  2D synchrosqueezed transform (SST) as input. The synchrosqueezed transform 
  is applied to extract initial information from atomic crystal images: 
  crystal defects, rotations and the gradient of elastic deformation.
  The deformation gradient estimate is then improved outside the identified defect 
  region via a variational approach, to obtain more robust
  results agreeing better with the physical constraints. The
  variational model is optimized by a nonlinear projected
  conjugate gradient method. Both examples of images from computer
  simulations and imaging experiments are analyzed, with results
  demonstrating the effectiveness of the proposed method.
\end{abstract}


\end{frontmatter}


\section{Introduction}

Defects, like dislocations, grain boundaries, vacancies, etc., play a
fundamental role in polycrystalline materials. They greatly change the
material behavior from a perfect crystal and affect the macroscopic
properties of the materials. Analysis of images arising from atomistic
simulations or imaging of polycrystalline materials hence becomes very
important to characterize and help understand the defects and their
effects on crystalline materials. While the defect analysis is
traditionally done by visual inspection, the large amount of data
made available due to advances in imaging and simulation techniques
creates a need of efficient computer-assisted or automated analysis.

Crystal deformation at the atomic scale is another 
important quantity that characterizes polycrystalline materials. 
When the deformation, denoted by $\psi$, is well-defined, the 
tensor field $F=\nabla \psi$ describes the local crystal strain;  
the polar decomposition of $F$ at each point gives grain rotations; 
the curl of $F$ provides information about defects and the well-known 
Burgers vector that represents the magnitude and
direction of the lattice distortion resulting from a dislocation. 
Since it is almost impossible to estimate the deformation manually, 
the development of computer-aided analysis becomes important. 

For crystalline materials, defects are physical domains of the
materials such that it is not possible to identify a smooth crystal 
deformation $\phi=\psi^{-1}$ that maps the atomic configuration back to a perfect 
lattice. In other words, the deformation gradient $G=\nabla\phi=F^{-1}$ is irregular and 
has nonzero curl at the defect location. In 
the opposite case, when a smooth deformation map does exist, the affine transform given by
the gradient of the map, $G = \nabla\phi$, transforms the image locally to an
undistorted lattice of atoms. Therefore, for a defect-free region of the material,
$G$ is a gradient field and thus is curl-free: $\curl
G = 0$. Crystal image analysis hence requires the detection 
of the defect regions and preferably also the estimation of the local 
elastic deformation $G$ away from the defects. 

In recent years several variational image processing methods for crystal analysis have emerged.
The most basic versions just segment the crystal image into several crystal grains and identify their orientation,
using clever formulations as convex optimization in $2$D \cite{Berkels:10,Strekalovskiy:11}
or efficient, sophisticated optimization algorithms for $3$D images \cite{ElseyWirth:14}.
Berkels et al.\ additionally determine a full deformation field $\psi$ \cite{Berkels:08} via nonlinear optimization,
while defects and local crystal distortion $G$ are identified in \cite{ElseyWirth:MMS}.

If accurate atom positions can easily be extracted from the
image (or if the data stems from atomistic simulations) such that the
input data consists of a discrete list of atom positions, then
  the local lattice orientation and deformation as well as defects can
  be obtained efficiently by identifying the nearest neighbors of each
  atom \cite{StukowskiAlbe1:10,StukowskiAlbe2:10,BeHuHa12}.

Another efficient approach is the crystal image analysis via $2$D
synchrosqueezed transforms (SSTs) in \cite{SSCrystal}. The SST was
originally proposed in \cite{DaubechiesMaes:96}, rigorously analyzed
in \cite{DaubechiesLuWu:11,SSSTFT,Yang:2014}, and extended to $2$D in
\cite{YangYing:2013,YangYing:2014}. It is proved that the $2$D SST can
accurately estimate the local wave vector of a nonlinear wave-like
image. Inspired by the fact that a deformed atomic crystal grain can
be considered as a superposition of several nonlinear wave-like
components, \cite{SSCrystal} proposes an efficient crystal image
analysis method based on a $2$D band-limited SST. In particular,
tracking the irregularity of the synchrosqueezed energy can identify
defects, and the deformation gradient $G$ can be obtained by a linear
system generated with local wave vector estimations.  A recent paper
\cite{Robustness} on the robustness of SSTs supports the application
of this analysis method to noisy crystal images. The
idea of exploiting the local periodic structure in the Fourier
domain to extract a deformation gradient has also been considered in
\cite{Hytch:98} in the crystal imaging literature.

Our work here is a variational approach based on the information
obtained from a band-limited $2$D synchrosqueezed wave 
packet transform following \cite{SSCrystal}. In this manuscript, 
we use this information
as input to a variational optimization in order to improve the robustness
of the analysis and, importantly, to make the results better agree
with the physical nature of defects.

\section{Variational method to retrieve deformation gradients}

Let us fix a perfect, unstrained crystal with a fixed orientation as our reference lattice.
Let $\Omega\subset\RR^2$ be the domain of the image.
Our objective is to find at each $x\in\Omega$ the local strain or deformation gradient $G(x)\in\RR^{2\times2}$
of a (locally defined) deformation $\phi$ which deforms a defect-free neighborhood of $x$ into the reference lattice.
(Note that since the image corresponds to a deformed crystal state, $x$ should be understood as an Eulerian coordinate.)
In other words, we seek the affine
map defined by $G(x)$ which maps the local atom arrangement to the
reference lattice. This is however impossible around defects. In
particular, while $\curl G = 0$ in the elastic region, in the defect
region, $\curl G$ is not zero, and the integral of $\curl G$ over a neighborhood of the defect such as a dislocation should match the defect's Burgers vector \cite{ElseyWirth:14}.

Assume the defect region is given by $\Omega_d\subset\Omega$ and the curl of G is given by $b$, consistent with the Burgers vectors and with $\curl G=0$ on $\Omega\setminus\Omega_d$.
We expect the displacement field to minimize the
elastic energy of the system outside the defect region, since the
system under imaging is in a quasistatic state.  Given $G_0$ a rough
guess of the deformation gradient, this motivates the energy
minimization
\begin{equation}\label{eq:energyelastic}
  \begin{aligned}
    & \min_G \int_{\Omega \backslash \Omega_{d}} \abs{G - G_0}^2 + W(G) \ud x \\
    & \text{s.\,t.} \curl G = b\,,
  \end{aligned}
\end{equation}
where $\abs{\cdot}$ denotes the Frobenius norm of a matrix, $\abs{A} =
(\tr (A^{\TT} A))^{1/2}$, and $W$ is the elastic stored energy density.

Since our reference lattice represents the undeformed equilibrium state of the crystal
and the atom configuration in the image is produced by the (local) deformation $\phi^{-1}$,
the stored elastic energy can be expressed in the standard Lagrangian form
as the integral over the reference domain $\phi(\Omega\setminus\Omega_d)$ of an elastic energy density $w$
that depends on $\nabla(\phi^{-1})=G^{-1}\circ\phi^{-1}$,
$$\int_{\phi(\Omega\setminus\Omega_d)}w(G^{-1}\circ\phi^{-1}(y))\,\ud y\,.$$
Here, $w$ satisfies the standard conditions coming from first principles, i.\,e.\ $w$ is frame indifferent, $w(A)=0$ for $A\in SO(2)$, $w(A)>0$ else, and $w(A)=\infty$ if $\det A\leq 0$.
After a change of variables the elastic energy turns into
$$\int_{\Omega\setminus\Omega_d}W(G)\,\ud x$$
for $W(G)=w(G^{-1})\det G$, where it is easy to see that $W$ has the same above properties as $w$.
For $w$ (or equivalently $W$) one can use a material-specific, possibly anisotropic energy density.
To be specific, since our numerical examples are all concerned with a triangular lattice exhibiting isotropic elastic behavior,
we here simply restrict ourselves to the following neo-Hookean-type elastic energy density
\begin{equation}
\label{eq:elasticEnergy}
  \begin{aligned}
    W(G) 
    & = \frac{\mu}{2} (\abs{G}^2 - 2) + \bigl(\frac{\mu}{2} + \frac{\lambda}{2}\bigr) (\det G - 1)^2 \\
    &- \mu (\det G - 1).
  \end{aligned}
\end{equation}
Note that in \eqref{eq:energyelastic}, the fidelity and elastic energy
terms are both evaluated outside the defect region. Within the defect
region, since it is not possible to map the local configuration of
atoms back to the reference state, the estimate $G_0$ is not
trustworthy. It is also well known that the elastic energy blows up
logarithmically approaching the dislocation core, and hence it only
makes sense to penalize the elastic energy away from the defects. 

The questions are then how to estimate $G_0$ and to determine the
defect region $\Omega_d$ as well as the vector field $b$ consistent with the defects' Burgers vectors. In this
work, the desired information is all obtained using synchrosqueezed
wave packet transforms. We first review the synchrosqueezed wave packet
transforms and the estimate of $G_0$ in \S\ref{sec:synsquez}. The
defect region and Burgers vectors are estimated using the
synchrosqueezed wave packet transforms as explained in
\S\ref{sec:defect}. The variational problem \eqref{eq:energyelastic} is then solved by a constrained minimization algorithm as described in \S\ref{sec:minimize}.

\subsection{Synchrosqueezed wave packet transforms (SSTs)}\label{sec:synsquez}

Before explaining the variational optimization, we briefly introduce the crystal image model established in \cite{SSCrystal} and the synchrosqueezed wave packet transforms in \cite{YangYing:2013,YangYing:2014} for the purpose of a self-contained description. 
Consider an image of a polycrystalline material with atomic resolution.
Denote the perfect reference lattice as 
\begin{equation*}
\mathcal L=\{av_1+bv_2\,:\,a,b\text{ integers}\}\,,
\end{equation*}
where $v_1,v_2\in\RR^2$ represent two fixed lattice vectors.
Let $s(x)$ be the image corresponding to a single perfect unit cell, extended periodically in $x$ with respect to the reference crystal lattice.
We denote by $\Omega$ the domain occupied by the whole
image and by $\Omega_k$, $k = 1, \ldots, M$, the grains the system consists of.
Now we model a polycrystal image $f:\Omega\to\RR$ as
\begin{equation}
f(x)= \alpha_k(x) s(N \phi_k(x))+c_k(x) \qquad \text{if } x \in \Omega_k,
\label{eqn:crystal}
\end{equation}
where $N$ is the reciprocal lattice parameter (or rather the relative
reciprocal lattice parameter as we will normalize the dimension of the
image).  The $\phi_k:\Omega_k \to \RR^2$ is chosen to map the atoms of
grain $\Omega_k$ back to the configuration of a perfect
crystal, in other words, it can be thought of as the inverse of the
elastic displacement field. The local inverse deformation gradient is
then given by $G = \nabla \phi_k$ in each $\Omega_k$. 
For generality, \eqref{eqn:crystal}
also includes a smooth amplitude envelope $\alpha_k(x)$ and a smooth
trend function $c_k(x)$ to take into account possible variation of
intensity, illumination, etc.~during the imaging process.  Using the
$2D$ Fourier series $\hat s$ of $s$ and the indicator functions
$\chi_{\Omega_k}$, we can rewrite \eqref{eqn:crystal} as
\begin{equation}
\begin{aligned}
  f(x) =\sum_{k=1}^M \chi_{\Omega_k} (x)\left( \sum_{\xi \in
      \mc{L}^{\ast}}\alpha_k(x) \widehat{s}(\xi)e^{iN \xi \cdot
      \phi_k(x)}+c_k(x)\right),
\label{eqn:imagefunction}
\end{aligned}
\end{equation}
where $\mc{L}^{\ast}$ is the reciprocal lattice of $\mc{L}$ (recall that $s$ is
periodic with respect to the lattice $\mc{L}$). In each grain
$\Omega_k$, the image is a superposition of wave-like components
$\alpha_k(x) \widehat{s}(\xi) e^{i N \xi \cdot \phi_k(x)}$ with
local wave vectors $N \nabla(\xi \cdot \phi_k(x))$ and local amplitude
$\alpha_k(x)|\widehat{s}(\xi)|$.  Our goal here is to apply the
band-limited $2$D SST developed in \cite{SSCrystal} to estimate the
defect region and also $G_0=\sum_{k=1}^M\chi_{\Omega_k}\nabla \phi_k$ in the interior of each
grain $\Omega_k$ as required in the variational formulation.

The starting point of $2$D SST is a wave packet $w_{a\theta x}$,
which is, roughly speaking, obtained by translating, rotating, and rescaling or modulating a mother wave packet $w:\RR^2\to\CC$ according to the parameters $x\in\RR^2$, $\theta\in[0,2\pi)$, and $a\in\RR$
\cite{YangYing:2013,YangYing:2014,SSCrystal}.  Let $W_f(a,\theta,x) =
\int_\Omega \overline{w_{a\theta x}}(y)f(y)\,\ud y$ be the wave packet transform of $f$ at
scale $a$ and angle $\theta$ in the frequency domain, and at spatial
location $x$. The wave packet transform is a generalization of curvelet and
wavelet transforms with better flexibility in frequency scaling and
consequently is better suited to analyze crystal images with complex
geometry. As a convolution with smooth wave packets, $W_f$ is well-defined and smooth even under very low regularity requirements for $f$, e.\,g.\ $f\in L^\infty(\RR^2)$. The spectrum of the windowed Fourier transform of a given crystal image spreads out in the
phase space, as illustrated in Figure~\ref{fig:patch}(b). The $2$D
synchrosqueezed wave packet transform (SST) aims at sharpening this phase space 
representation. In the SST, for each $(a, \theta, x)$, we define the
corresponding local wave vector
\begin{equation*}
  v_f(a, \theta,x) = \Re \frac{\nabla_x
    W_f(a,\theta,x)}{2\pi i W_f(a,\theta,x)}.
\end{equation*}
Here, $\nabla_x W_f$ denotes the gradient of $W_f$ with respect to its third argument $x$.
The synchrosqueezed (SS) energy distribution of $f$ is then constructed as
\begin{equation}
T_{\!f}(v,x) \!=\! \int_0^{2\pi}\!\!\!\!\!\int_{0}^{\infty} \!\!|W_{\!f}(a,\theta,x)|^2 \delta(v_{\!f}(a,\theta,x)\!-\!v) \, a\ud a\ud\theta\,, \label{eq:SED}
\end{equation}
where $\delta$ denotes the Dirac measure. By the stationary phase method, in a small neighborhood of a point $x$ such that $W_f(a,\theta,x)$ is nonzero for some fixed $(a,\theta)$, $W_f(a,\theta,x)$ is essentially a plane wave in $x$ with a local wave vector approximating a certain $N \nabla(\xi \cdot \phi_k(x))$. Inspired by this intuition, it is proved in \cite{SSCrystal} that in the interior of a grain, the local wave vector estimation $v_f(a,\theta,x)$ approximates 
 local wave vectors $N\nabla (\xi\cdot \phi_k(x))$ accurately, if the deformation $\phi_k$ and the amplitude function $\alpha_k$ are smooth. 
The SST squeezes the wave packet spectrum $|W_f(a, \theta,
x)|^2$ according to $v_f(a, \theta, x)$ to
obtain a sharpened and concentrated representation of the image in the
phase space. Hence, in
the interior of a grain, the SS energy distribution $T_f$ has a support concentrating
around local wave vectors $N\nabla (\xi\cdot \phi_k(x))$, $ \xi \in
\mc{L}^{\ast}$, and is given approximately by (see e.g.,
Figure~\ref{fig:patch}(e) in polar coordinates)
\begin{equation}
  T_f(v, x) \approx \sum_{\xi \in \mc{L}^{\ast}} \alpha_k(x)^2|\hat{s}(\xi)|^2 
  \delta\bigl( v - N \nabla ( \xi\cdot \phi_k(x))\bigr), \label{eq:Tfapprox}
\end{equation}
understood in the distributional sense. Therefore, by locating the
energy peaks of $T_f$, we can obtain estimates of local wave vectors
$N\nabla(\xi\cdot \phi_k(x))$ and also their associated spectral energy. In
practice, we will aim at high energy peaks corresponding to $\xi$ close to the
origin in the reciprocal lattice.
For simplicity and concreteness, we will in this article focus on the case when the lattice is hexagonal.
We then have six such reciprocal lattice vectors $\xi$, which can be further reduced to three due to the
symmetry $\xi \leftrightarrow -\xi$. We will henceforth denote these
as $\xi_n$, $n = 1, 2, 3$, and denote by $v^{\text{est}}_n(x)$ the
estimate of $N \nabla (\xi_n \cdot \phi_k(x)) = N (\nabla \phi_k(x))
\xi_n$.  The inverse deformation gradient $G_0(x) =\nabla \phi_k(x)$ is
determined by a least squares fitting of the deformed reference reciprocal lattice vectors $N \xi_n$
to the $v^{\text{est}}_n$:
\[
G_0(x) = \underset{G}{\argmin} \sum_{n=1}^3 \norm{v^{\text{est}}_n(x) - N G \xi_n}_2^2.
\]
In practice, for each physical point $x$ we represent $T_f(\cdot, x)$
in polar coordinates $(r,\vartheta)\in[0,\infty)\times[0,\pi)$ (where
$\vartheta\in[\pi,2\pi)$ is ignored due to symmetry). To identify the
peak locations $\{v^{\text{est}}_n\}$, we choose the grid point with highest
amplitude in each $60$ degree sector of $\vartheta$. 

As the representation \eqref{eq:Tfapprox} of $T_f$ is no longer valid around the crystal defects, we may characterize the defect region by using an indicator function for the deviation of $T_f$ from the representation \eqref{eq:Tfapprox}. To this end, for each $n \in \{1, 2, 3\}$ (corresponding to one of the sectors), we define
\[
w_n(x)=\frac{\displaystyle \int_{B_\delta(v^{\text{est}}_n)} T_f(v,x) \ud v}{\displaystyle \int_{\arg v \in[(n-1)\pi/3,n\pi/3)} T_f(v,b) \ud v}\,,
\]
where $B_{\delta}(v^{\text{est}}_n)$ denotes a small ball around the
estimated local wave vector $v^{\text{est}}_n$. Hence, $\mass(x) :=
\sum_n w_n(x)$ will be close to $3$ in the interior of a grain due to
\eqref{eq:Tfapprox}, while its value will be much smaller than $3$
near the defects. This is illustrated in Figure~\ref{fig:patch}(c),
where we show $\mass(x)$ for the crystal image in
Figure~\ref{fig:patch}(a). The estimate of defect regions can be
obtained by a thresholding $\mass(x)$ at some value $\eta\in(0,3)$ according to
\begin{equation*}
\Omega_d=\{x\in\Omega\ :\ \mass(x)<\eta\}\,,
\end{equation*}
an illustration of which is shown in Figure~\ref{fig:patch}(d). 
Figure~\ref{fig:patchG} shows the estimate of $G_0$ by the SST. It is
observed that the crystal is slightly deformed in the grain interior
and heavily deformed at the defect region. 

To better interpret the inverse deformation gradient $G_0$, we compute
its polar decomposition
$G_0(x)=U_0(x)P_0(x)$ for each point $x\in \Omega$, where $U_0(x)$ is
a rotation matrix and $P_0(x)$ is a positive-semidefinite symmetric
matrix. The rotation angle of $U_0(x)$ describes the crystal
orientation at $x$; $\det(G_0(x))-1$ indicates the volume distortion
of $G_0(x)$; the quantity $|\lambda_1(x)-\lambda_2(x)|$, where
$\lambda_1(x)$ and $\lambda_2(x)$ are the eigenvalues of $P_0(x)$,
characterizes the difference in the principal stretches of
$G_0(x)$ as a measure of shear strength. The bottom panel of Figure~\ref{fig:patchG} shows these
quantities corresponding to the estimate of $G_0$ in the top panel. In
the later numerical examples, instead of $G$ itself we will always present the
crystal orientation, the volume distortion, and the difference in
principal stretches.

\begin{figure}[ht!]
  \begin{center}
    \begin{tabular}{cc}
      \includegraphics[height=1.2in]{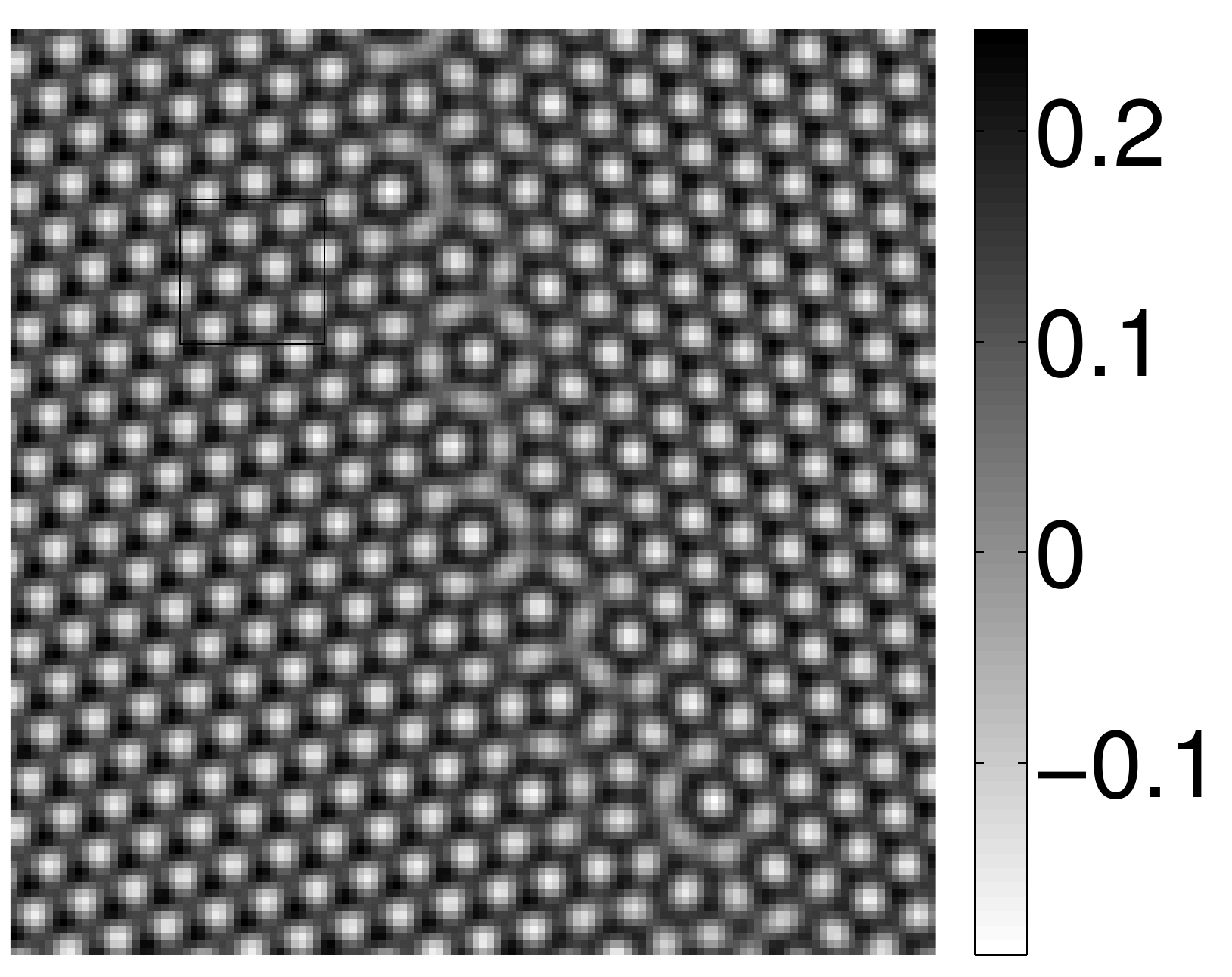} &  \includegraphics[height=1.2in]{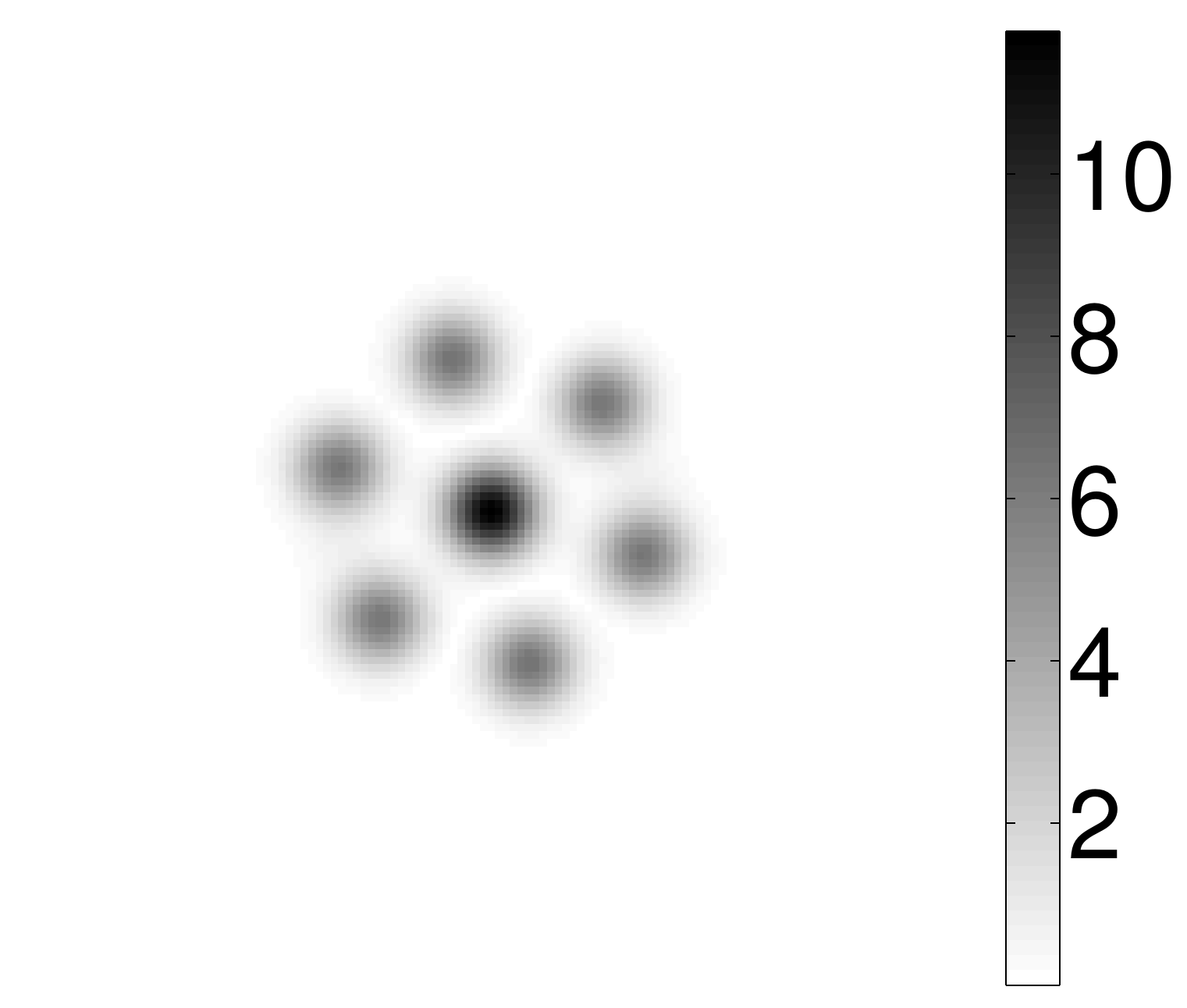} \\
       (a)& (b) \\
        \includegraphics[height=1.2in]{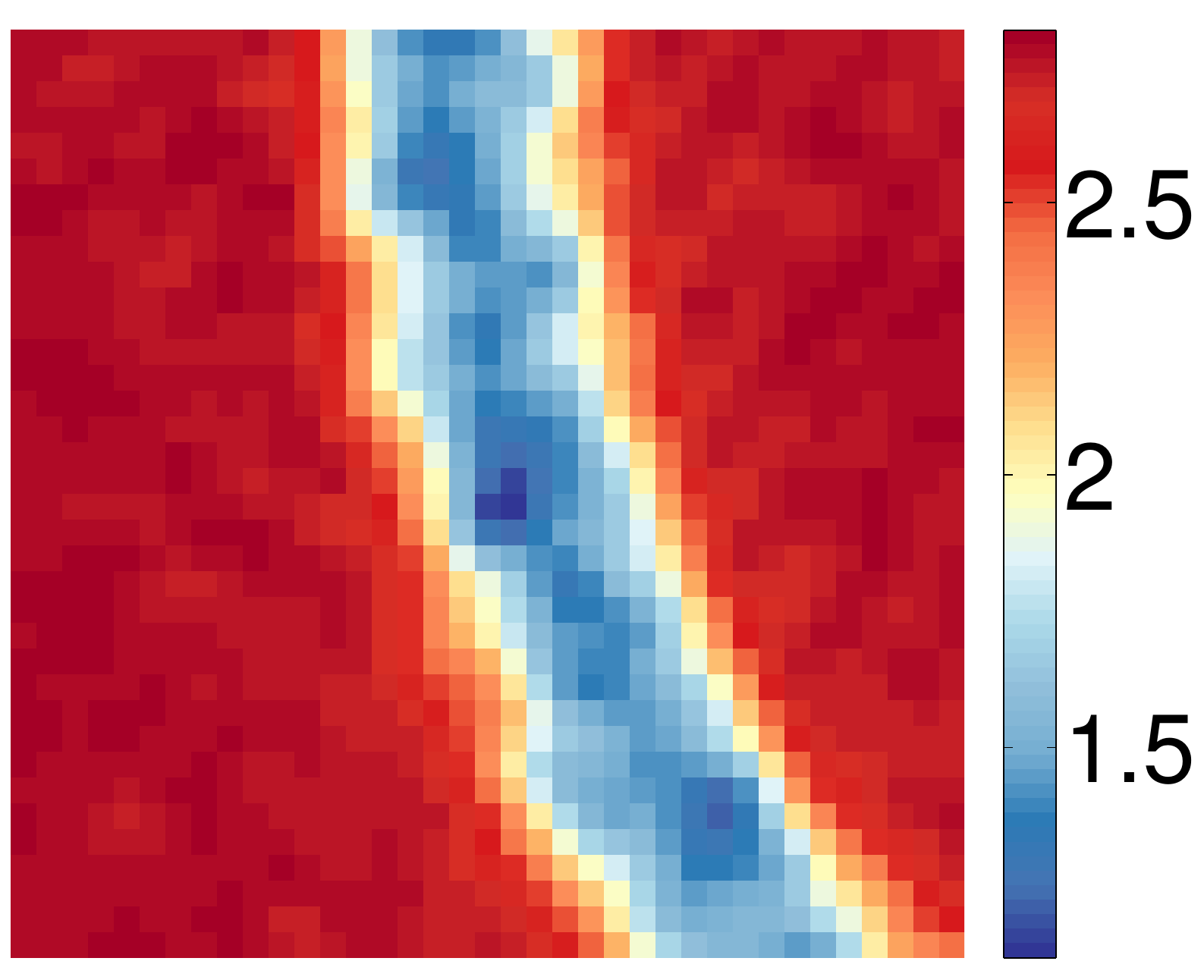} &  \includegraphics[height=1.2in]{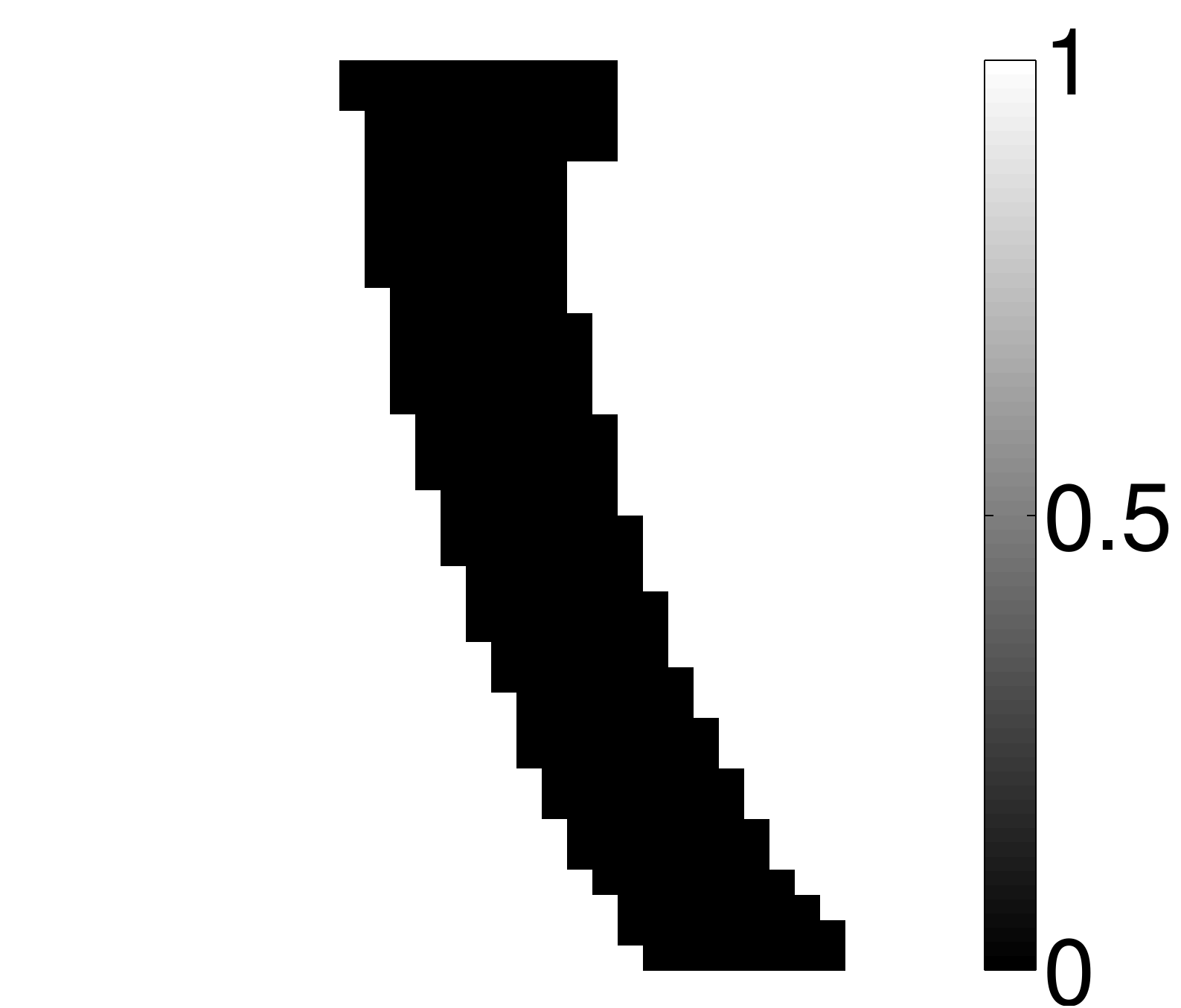}\\
     (c) & (d) \\
       \includegraphics[height=1.2in]{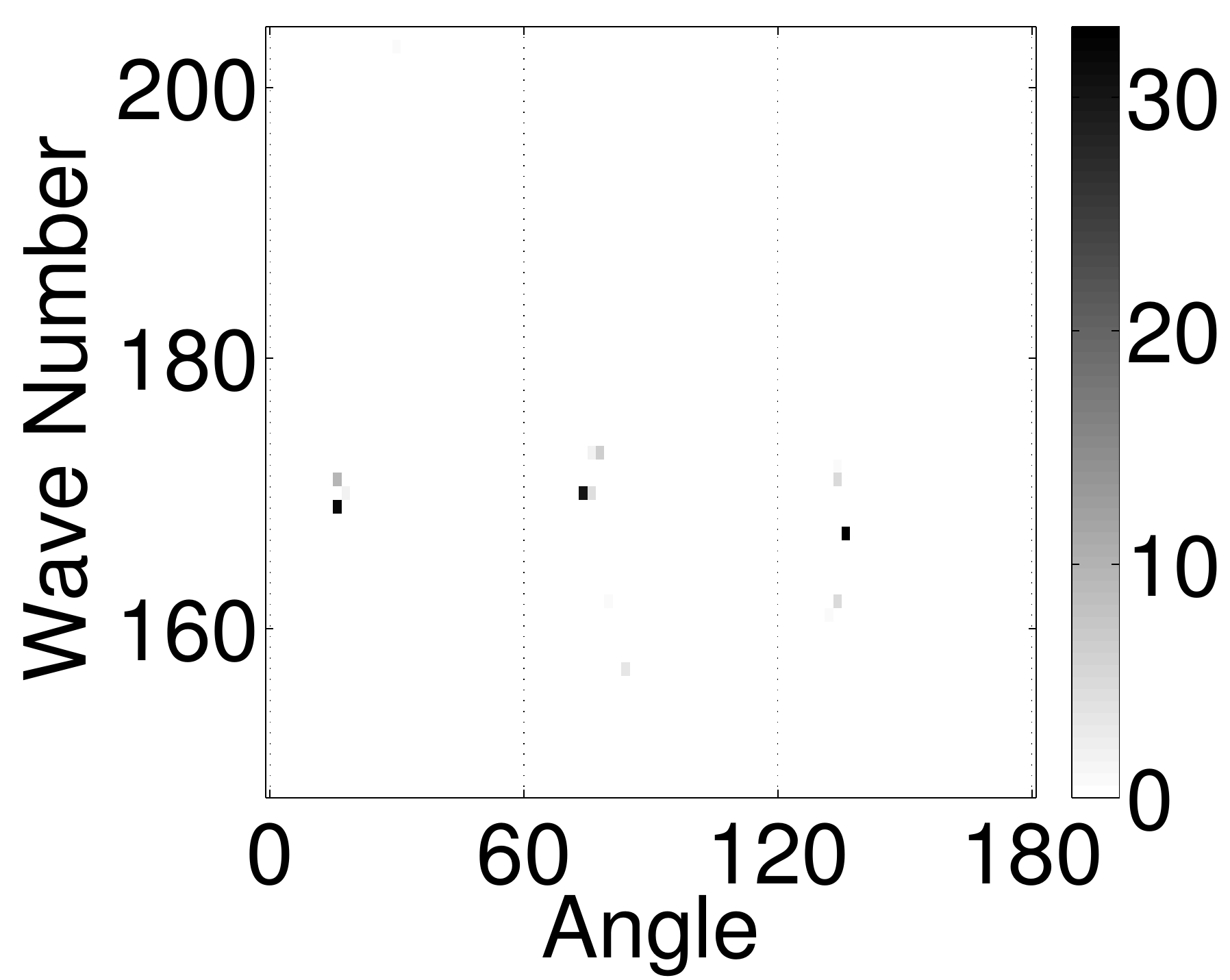} &  \includegraphics[height=1.2in]{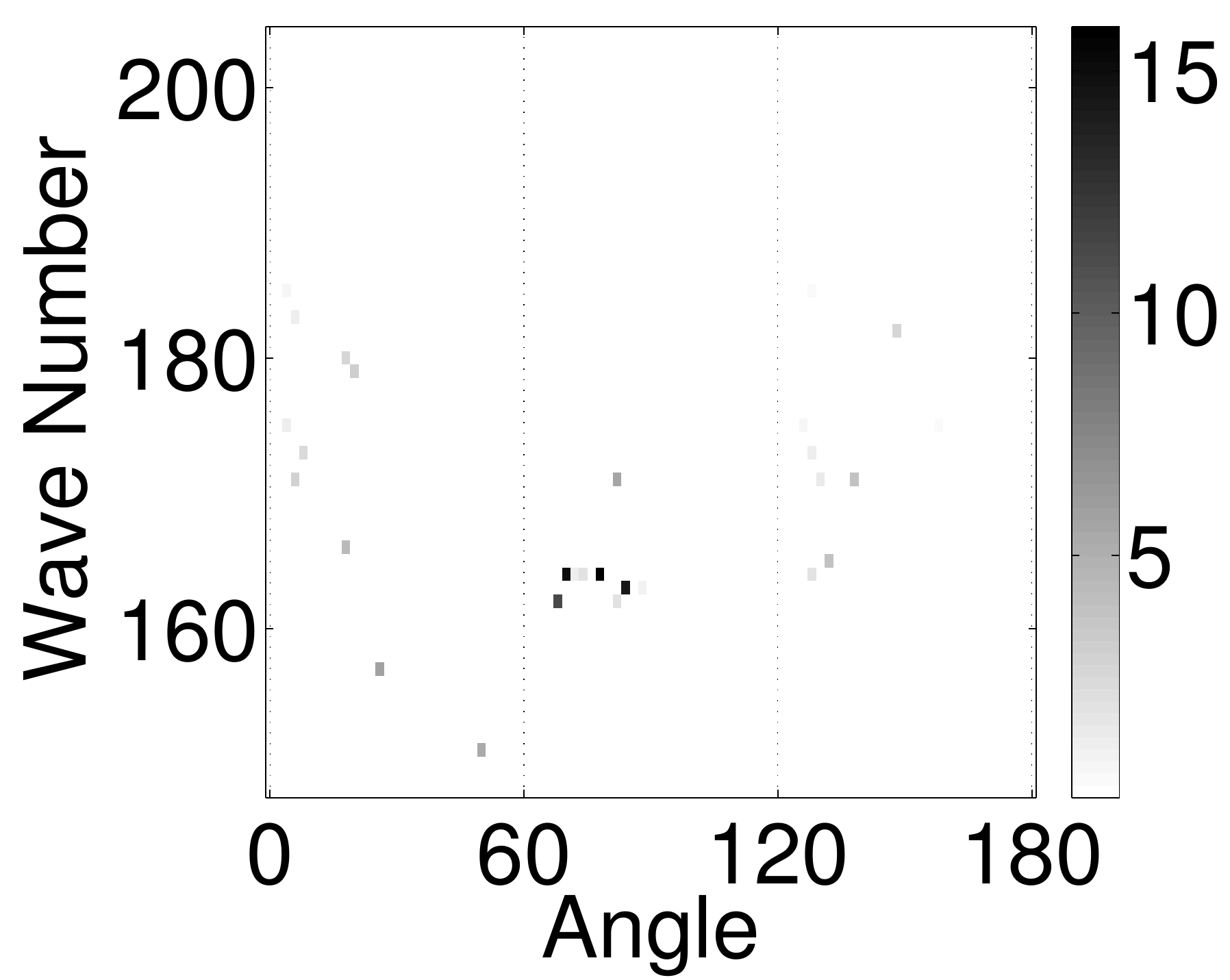} \\
       (e)&(f)
    \end{tabular}
  \end{center}
  \caption{(a) An example of a crystal image. (b) Windowed Fourier
    transform at a local patch indicated by a rectangle. (c) The defect
    indicator $\mass(x)$ provided by the SST. (d) Identified defect region for a threshold of $\eta=2$. (e) The SS energy
    distribution in polar coordinates at a point outside the defect
    region. (f) SS energy distribution in polar
    coordinates at a point in the defect region.}
  \label{fig:patch}
\end{figure}

\begin{figure}[ht!]
  \begin{center}
    \begin{tabular}{cc}
      \includegraphics[height=1.2in]{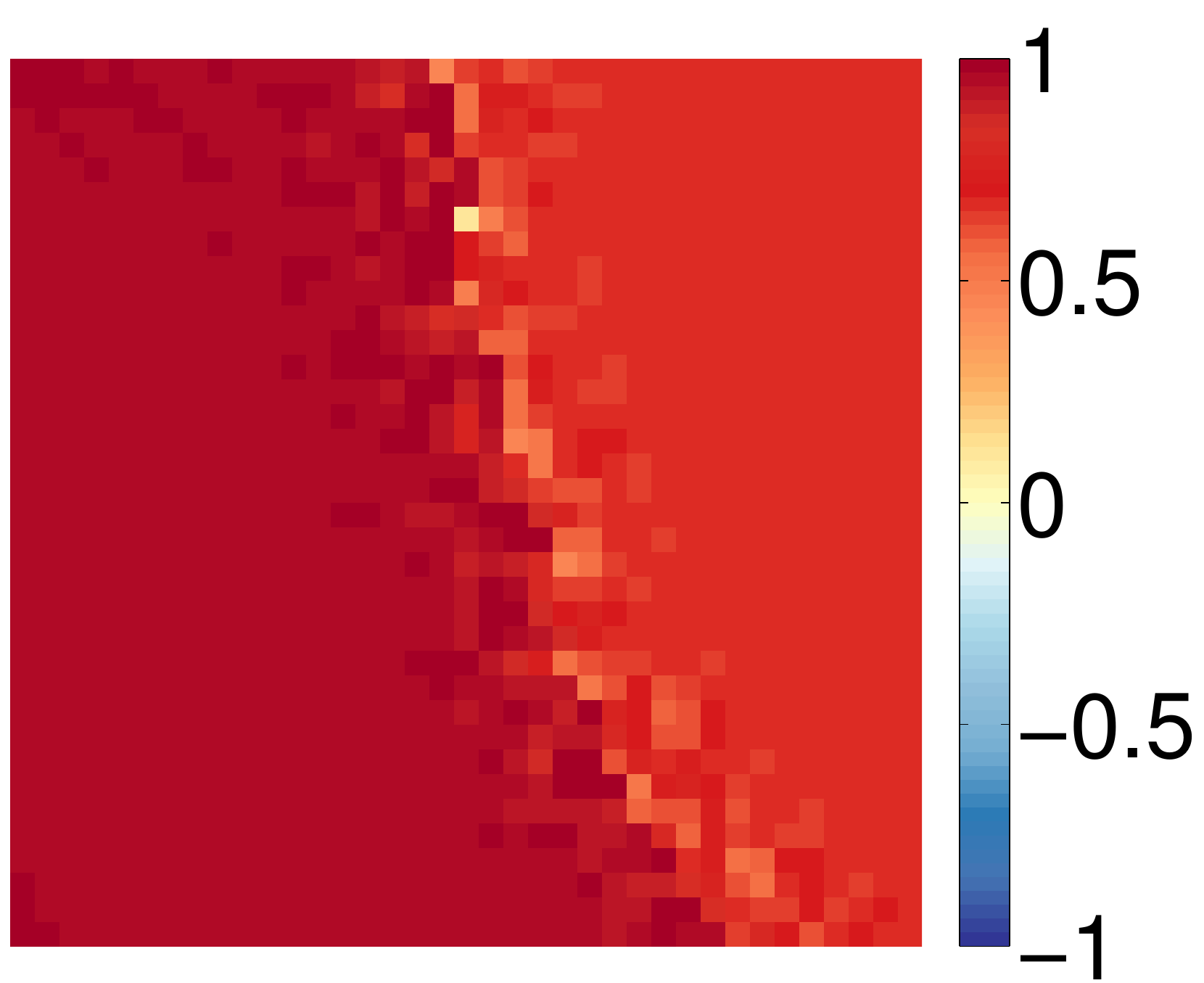} &  \includegraphics[height=1.2in]{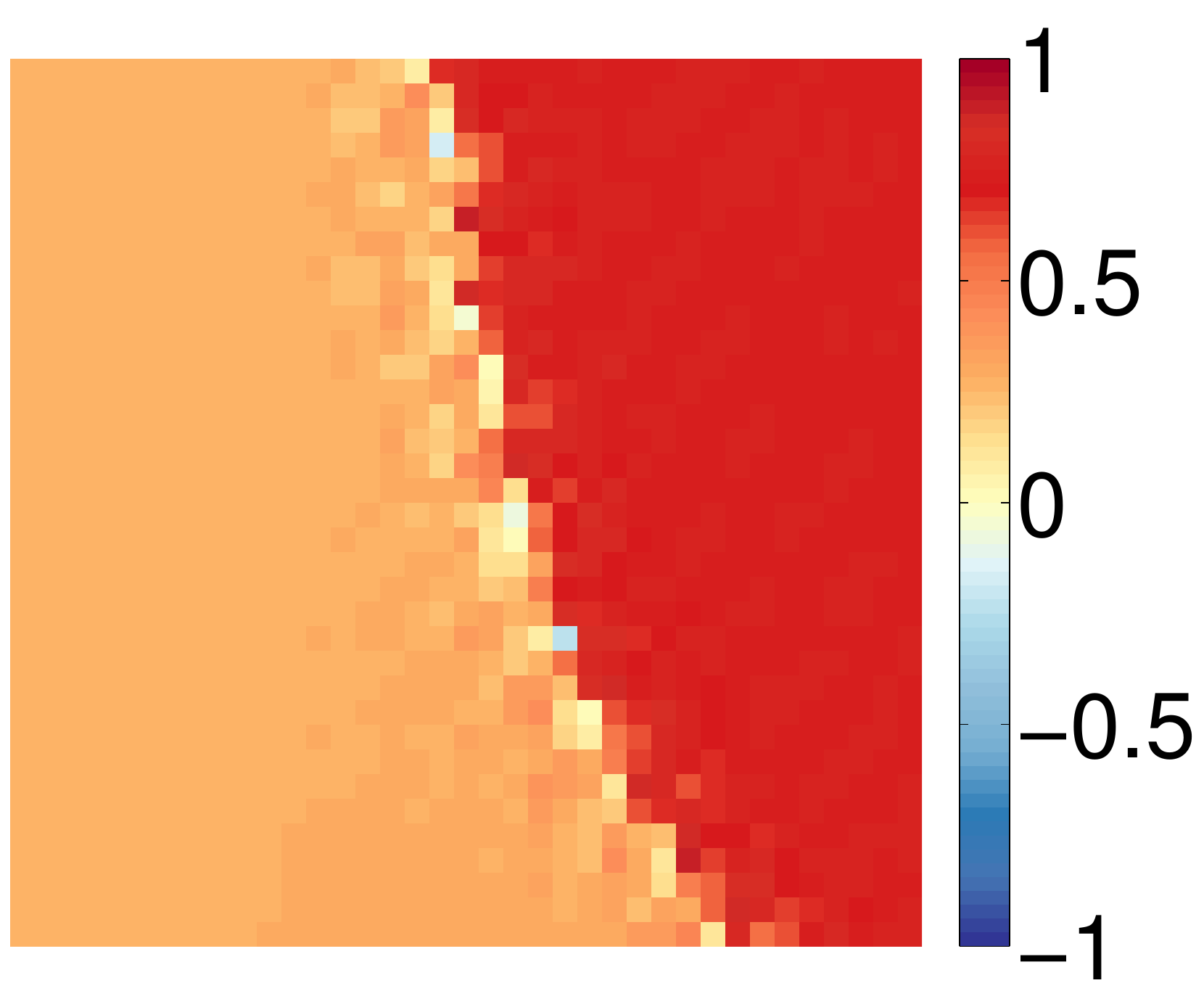} \\
      $G_0^{11}$& $G_0^{21}$ \\
        \includegraphics[height=1.2in]{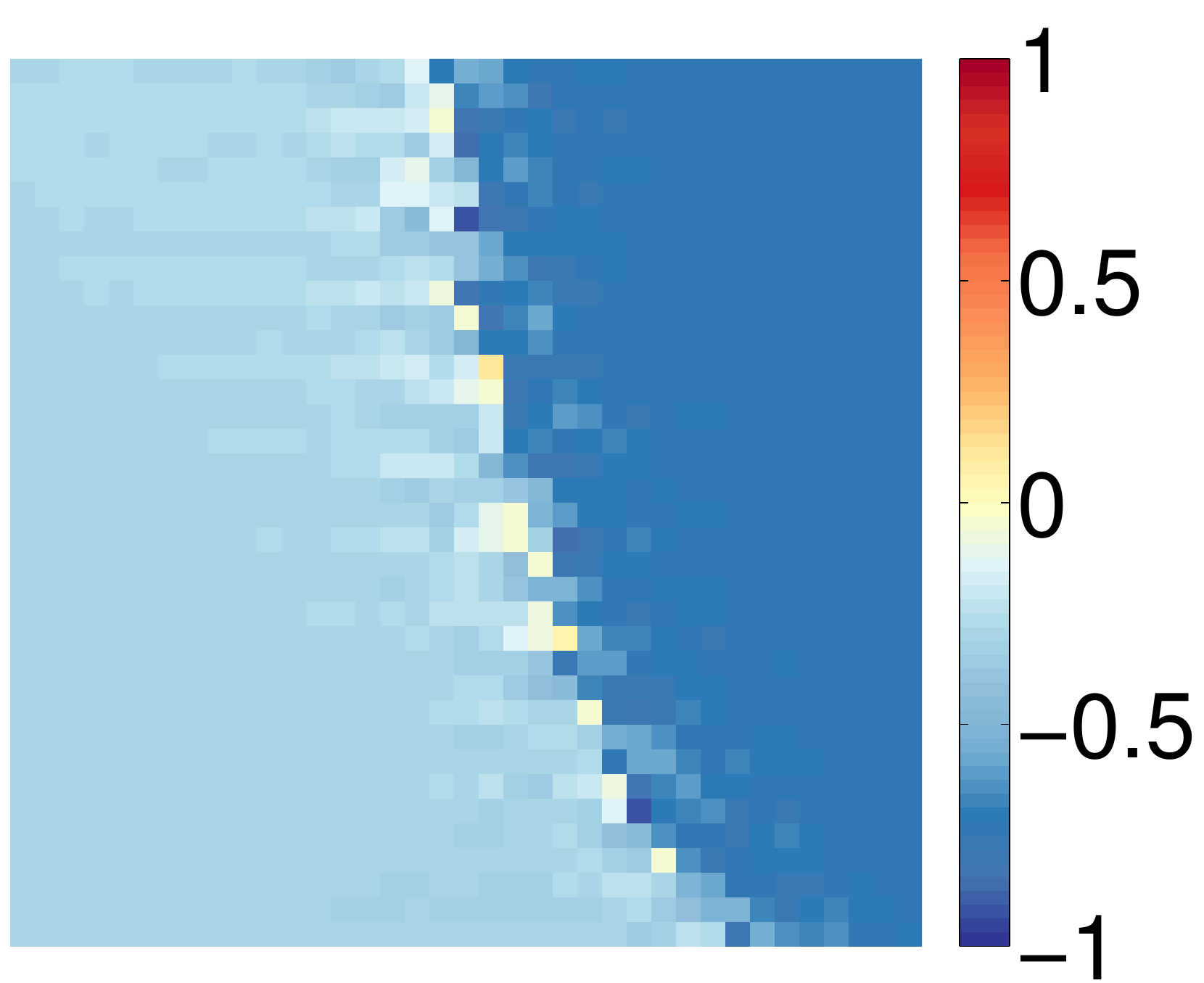} &  \includegraphics[height=1.2in]{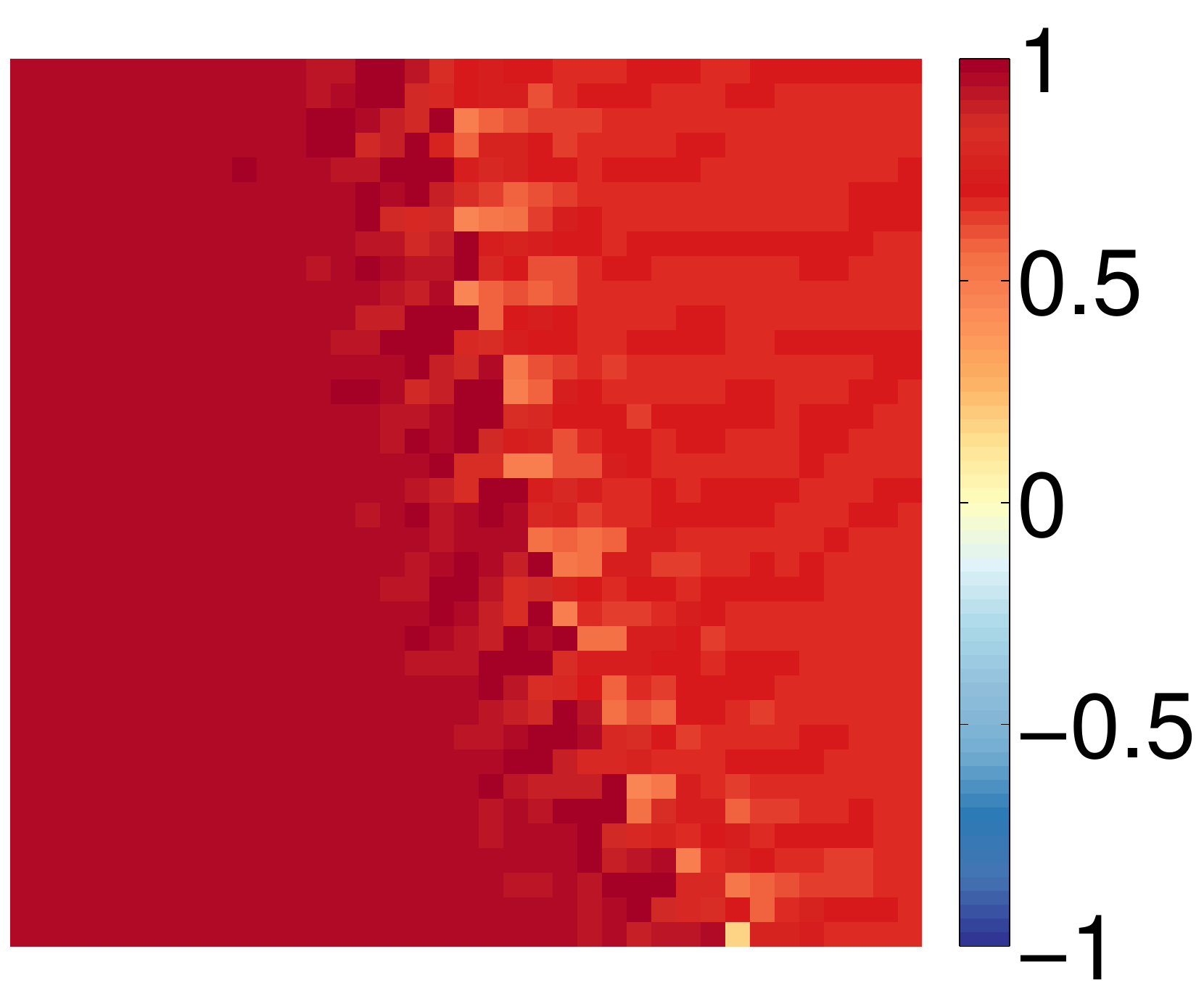}\\
       $G_0^{12}$ & $G_0^{22}$\\
    \end{tabular}
    \\
        \begin{tabular}{cc}
      \includegraphics[height=1.2in]{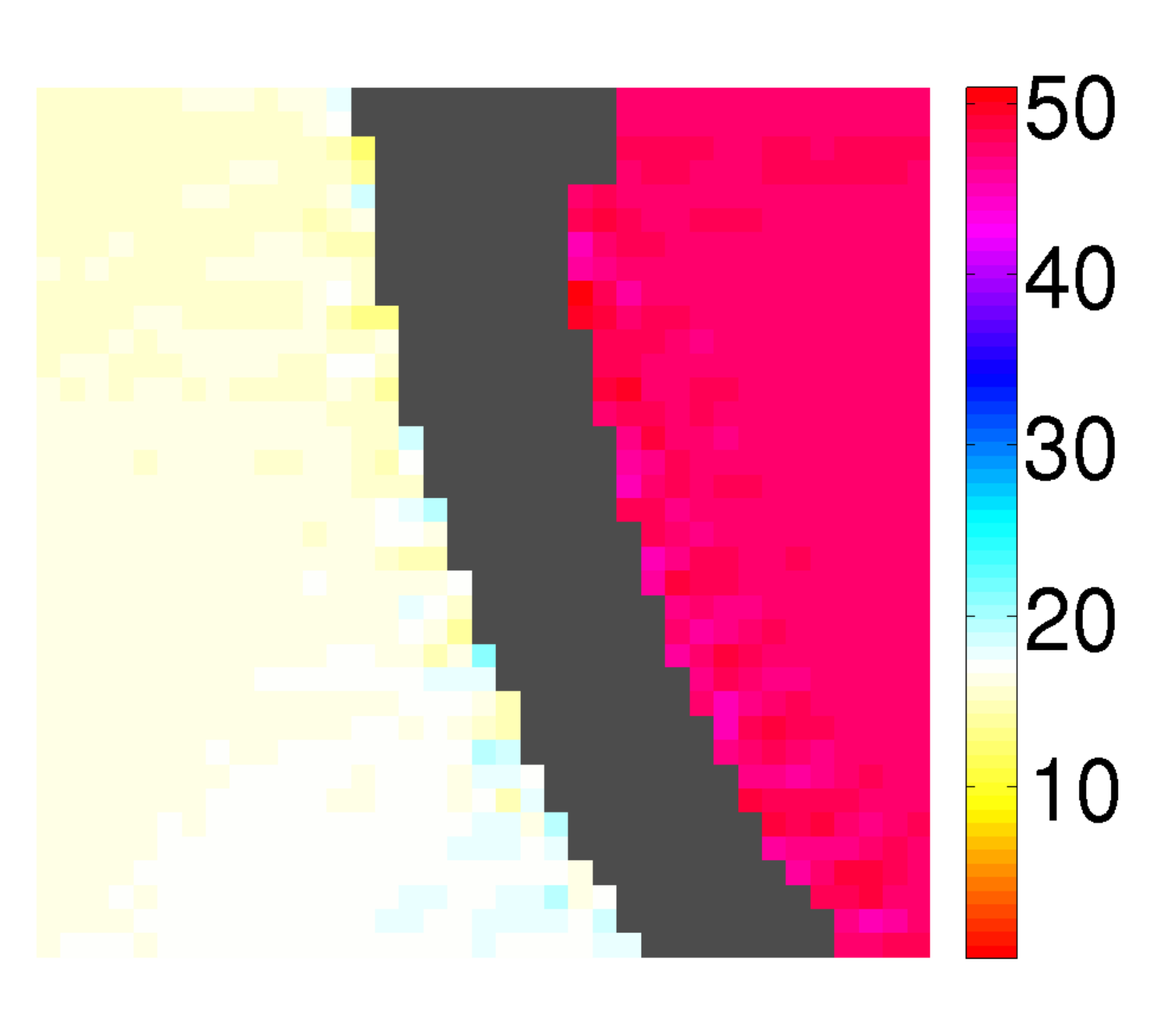} &  \includegraphics[height=1.2in]{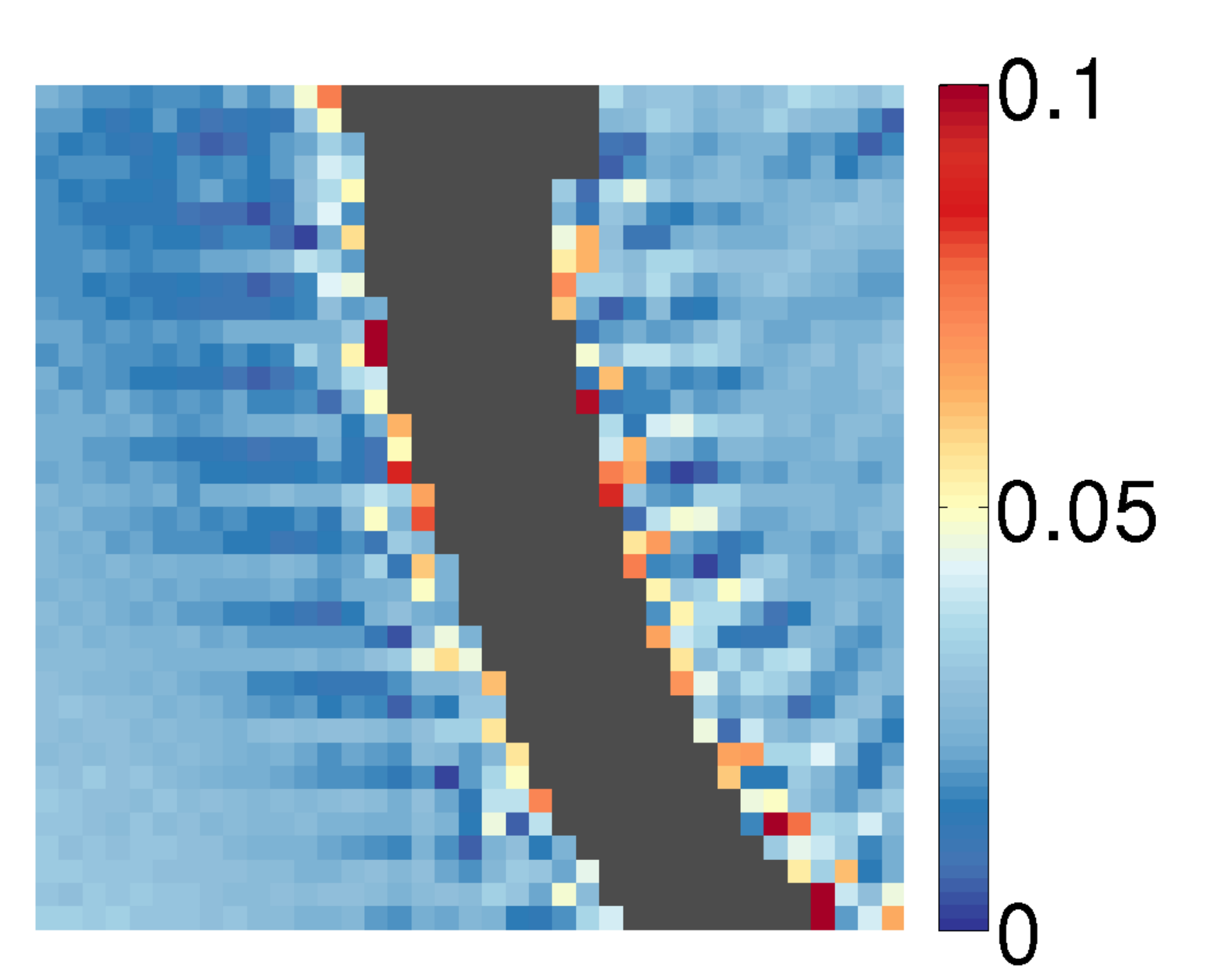} \\
      Crystal Orientation & Difference in principal stretches\\
    \end{tabular}
    \\
            \begin{tabular}{c}
      \includegraphics[height=1.2in]{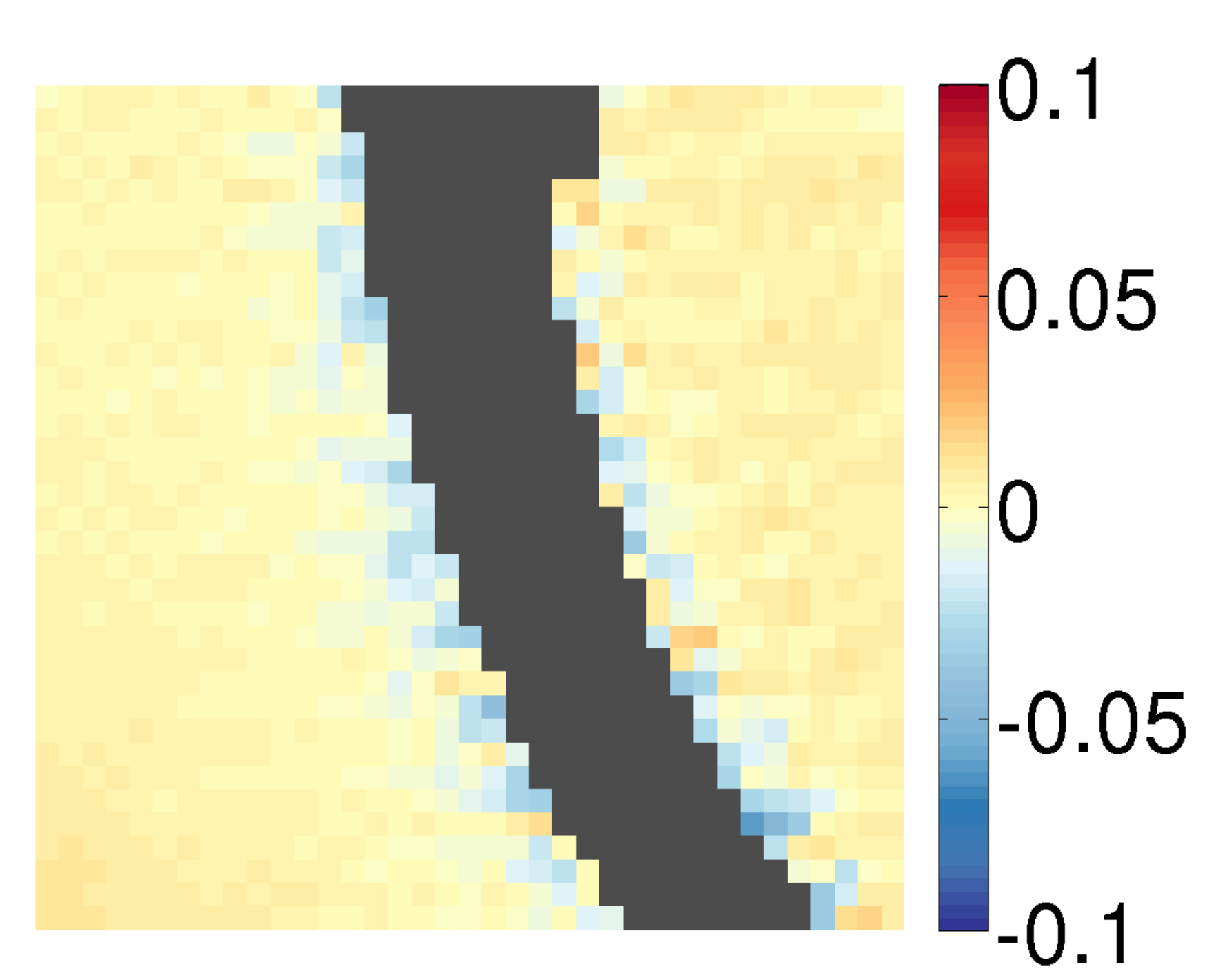} \\
      Volume Distortion
    \end{tabular}
  \end{center}
  \caption{Top panel: Estimated inverse deformation gradient $G_0\in \RR^{2\times 2}$ of the atomic crystal image in Figure \ref{fig:patch} (a). Bottom panel: The crystal orientation, the difference in principal stretches, and the volume distortion of $G_0$. The grey mask in these figures is the defect region identified in Figure~\ref{fig:patch} (d).}
  \label{fig:patchG}
\end{figure}

\subsection{Defect region, topological consistency, and Burgers vector identification}\label{sec:defect}

Based on the estimated $G_0$, we can determine the defect region
$\Omega_d$ and also the Burgers vectors or equivalently the curl field $b$ used in the constraint of
our variational formulation \eqref{eq:energyelastic}.  We give the
details in this section.

\subsubsection{Burgers vectors and $\curl G$}
As explained previously, away from defects, $G$ can be interpreted as the gradient $\nabla\phi$ of a (fictitious) deformation $\phi$ deforming the configuration of the given image into a perfect reference crystal of a fixed orientation.
Now the fact that gradient fields are always curl-free can be exploited as a constraint
\begin{equation}\label{eqn:zeroCurl}
\curl G=0\quad\text{on }\Omega\setminus\Omega_d
\end{equation}
in the variational method.
In the defect region $\Omega_d$, however, the interpretation of $G$ as a deformation gradient breaks down since there is no smooth deformation of the crystal that can undo the lattice defect.
In fact, denoting the connected components of $\Omega_d$ by $\Omega_d^1,\ldots,\Omega_d^l$, it is relatively simple to see that the integral
\begin{equation}\label{eqn:BurgersVector}
B_i=\int_{\Omega_d^i}\curl G\,\ud x
\end{equation}
is the Burgers vector associated with the defect in $\Omega_d^i$. 
Indeed, consider any closed injective curve $\gamma:[0,1]\to\RR^2$ that encloses $\Omega_d^i$ counterclockwise, but does not intersect $\Omega_d$.
Figure~\ref{fig:BurgersVector}, left, shows a specific example of such a curve, which is chosen such that it connects a sequence of neighboring atoms.
Denote the curve interior by $\tilde\Omega$.
Since $\partial\tilde\Omega$ lies in the defect-free region,
for every $x\in\partial\tilde\Omega$ there is a deformation $\phi$ (the gradient of which is given by $G$)
that deforms a neighborhood of $x$ into the perfect reference configuration of the crystal.
Starting at $\gamma(0)$ we can thus iteratively map little segments of the curve $\gamma$ into the reference crystal via $\gamma\circ\phi^{-1}$,
resulting in a curve $\hat\gamma:[0,1]\to\RR^2$ on the reference configuration.
Note that $\hat\gamma$ is in general no longer closed or injective.
By Stokes' Theorem, using the tangent vector $n^\perp$ to $\partial\tilde\Omega$, we have
\begin{multline*}
\int_{\Omega_d^i}\curl G\,\ud x
=\int_{\tilde\Omega}\curl G\,\ud x
=\int_{\partial \tilde\Omega}Gn^\perp\,\ud\sigma\\
=\int_0^1G(\gamma(t))\dot\gamma(t)\,\ud t
=\int_0^1\dot{\hat\gamma}(t)\,\ud t
=\hat\gamma(1)-\hat\gamma(0)\,.
\end{multline*}
The vector on the right-hand side is the defect's Burgers vector and is obviously independent of the originally chosen curve $\gamma$.
If $\Omega_d^i$ contains an isolated dislocation surrounded by a regular lattice, $B_i$ just represents the Burgers vector of that dislocation.
If $\Omega_d^i$ contains multiple dislocations or even a section of a high angle grain boundary, $B_i$ represents the accumulated Burgers vector of all defects in $\Omega_d^i$
(note that a high angle grain boundary may be thought of as a string of dislocations with distance smaller than the lattice spacing so that the single dislocations are not clearly spatially separated; Figure~\ref{fig:PFC1}(c) gives an impression of the size of the numerically found $\Omega_d^i$.

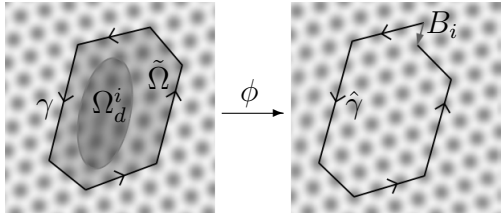
\begin{figure}
\setlength{\unitlength}{1.2\linewidth}
\begin{center}%
\includegraphics[width=.4\unitlength]{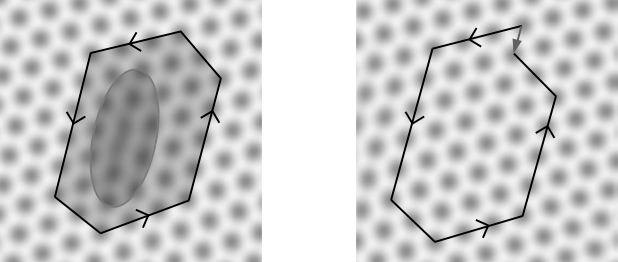}%
\begin{picture}(0,0)(.225,-.08)
\put(0,0){\vector(1,0){.05}}
\put(.015,.01){$\phi$}
\put(-.15,0){$\gamma$}
\put(.1,0){$\hat\gamma$}
\put(-.105,0){$\Omega_d^i$}
\put(-.06,.02){$\tilde\Omega$}
\put(.165,.065){$B_i$}
\end{picture}%
\end{center}%
\caption{A closed curve $\gamma$ around a defect in the crystal image breaks up when transported to the perfect crystal, the gap being the Burgers vector.}
\label{fig:BurgersVector}
\end{figure}

As in the case of $\Omega\setminus\Omega_d$, where we know $\curl G=0$, we also have a priori information on $\curl G$ in $\Omega_d^i$.
In particular we know that $B_i$ is a Burgers vector and thus must lie in the discrete set of Bravais lattice vectors of the perfect reference lattice $\mathcal L$. 
We thus identify $B_i$ by projecting the (potentially noisy) estimate $\int_{\Omega_d^i}\curl G_0\,\ud x$ onto $\mathcal L$
and then impose \eqref{eqn:BurgersVector} as a constraint on $G$.
In fact, instead of prescribing the accumulated curl in $\Omega_d^i$ via \eqref{eqn:BurgersVector} we may just as well prescribe
\begin{equation}\label{eqn:pointwiseCurlConstraint}
  \curl G=b_i\quad\text{on }\Omega_d^i
\end{equation}
for a function $b_i:\Omega_d^i\to\RR^2$ with
$\int_{\Omega_d^i}b_i\,\ud x=B_i$ (in mathematically more precise
terms, $b_i$ may be a distribution).  This is possible since we are
only interested in the field $G$ outside of $\Omega_d$ and since any
field $G:\Omega\to\RR^{2\times2}$ satisfying \eqref{eqn:zeroCurl} and
\eqref{eqn:BurgersVector} can be modified on $\Omega_d$ to a field
satisfying \eqref{eqn:pointwiseCurlConstraint}.  The function $b_i$ is here
simply chosen as $b_i=\mathrm{diag}(\alpha,\beta)\curl G_0$ with $\alpha,\beta\in\RR$
such that $\int_{\Omega_d^i}b_i\,\ud x=B_i$ (\textit{i.e.}, an overall
scaling of $\curl G_0$).  Summarizing, in our variational method to
extract $G$ from the noisy data $G_0$ we will prescribe the constraint
\begin{equation}\label{eqn:curlConstraint}
\curl G=b\qquad\text{for}\quad b=\begin{cases}0&\text{on }\Omega\setminus\Omega_d;\\b_i&\text{on }\Omega_d^i\,.\end{cases}
\end{equation}

\subsubsection{Refined defect regions}
On the one hand, the threshold $\eta$ to identify $\Omega_d$ should be chosen very low to yield thin and localized defect regions
(e.\,g.\ such that defect regions $\Omega_d^i$ around single dislocations stay separated from each other),
on the other hand, the thinner the identified defect region $\Omega_d^i$ the worse will the estimate of the Burgers vectors $B_i$ be.
A compromise is to first use thick defect regions $\tilde\Omega_d^i$ in order to estimate the Burgers vectors $B_i$
and then to impose the constraint \eqref{eqn:curlConstraint} with a much finer estimate of the $\Omega_d^i$.
However, it may happen that a thick patch $\tilde\Omega_d^i$ contains multiple thin connected components $\Omega_d^{i_1},\ldots,\Omega_d^{i_k}$.
In that case the $B_{i_j}$ are defined as the closest projections of $\int_{\Omega_d^{i_j}}\curl G_0\,\ud x$ onto $\mathcal L$ under the constraint $B_{i_1}+\ldots+B_{i_k}=\tilde B_i$,
where $\tilde B_i$ is the accumulated Burgers vector of the patch $\tilde\Omega_d^i$.
In order to obtain a very thin and localized $\Omega_d$ we simply identify the ridge of $3-\mass(b)$ inside the thick $\tilde\Omega_d$ and then dilate this ridge by a few pixels.
The ridge computation turns out to be sufficiently robust even in the presence of noise, as can be seen from the similarity between the detected defect regions with and without noise in Figures~\ref{fig:PFC1} and \ref{fig:PFC2}.

\subsubsection{Introduction of jump sets for topological consistency}
A Bravais lattice does typically not only exhibit translational, but
also rotational symmetry.  The so-called point group
$\pointgroup\subset SO(2)$ comprises all those rotations which leave
the reference lattice invariant.  This leads to an ambiguity in the
deformation gradient $G$: if $G$ correctly describes the local
configuration of the crystal, then $RG$ for any $R\in\pointgroup$ does
so as well.  Even though the constraint \eqref{eqn:curlConstraint} has the
effect that the matrix field $G$ will be locally consistent (in the
sense that for any $y$ in a neighborhood of $x\in\Omega$ we have
$|G(y)-G(x)| = \min_{R\in\pointgroup}|RG(y)-G(x)|$), global
consistency is often not guaranteed.  Indeed,
Figure~\ref{fig:pointGroup} shows a situation in which along a closed
path $\gamma\subset\Omega$, $G$ changes continuously from the identity
$I$ to an element $R\neq I$ of the point group.  Since $R$ describes
the same local crystal configuration as $I$, the curl where $G$ jumps
from $R$ to $I$ is spurious.  As in \cite{ElseyWirth:MMS}, this
inconsistency can be remedied by introducing a cut set $\jumpset$
across which the tangential component of $G$ is allowed to jump by a point group element,
\begin{equation*}
G^-n^\perp=RG^+n^\perp
\end{equation*}
for some $R:\jumpset\to\pointgroup$,
where $G^-$ and $G^+$ denote the value of $G$ on either side of $\jumpset$ and $n^\perp:\jumpset\to\RR^2$ denotes the tangent to $\jumpset$.

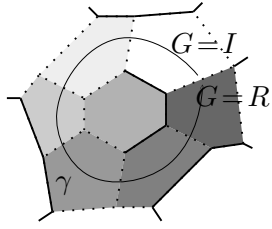
\begin{figure}
\centering
\setlength{\unitlength}{.8\linewidth}
\begin{picture}(.3,.3)
\put(0,0){\includegraphics[width=.3\unitlength]{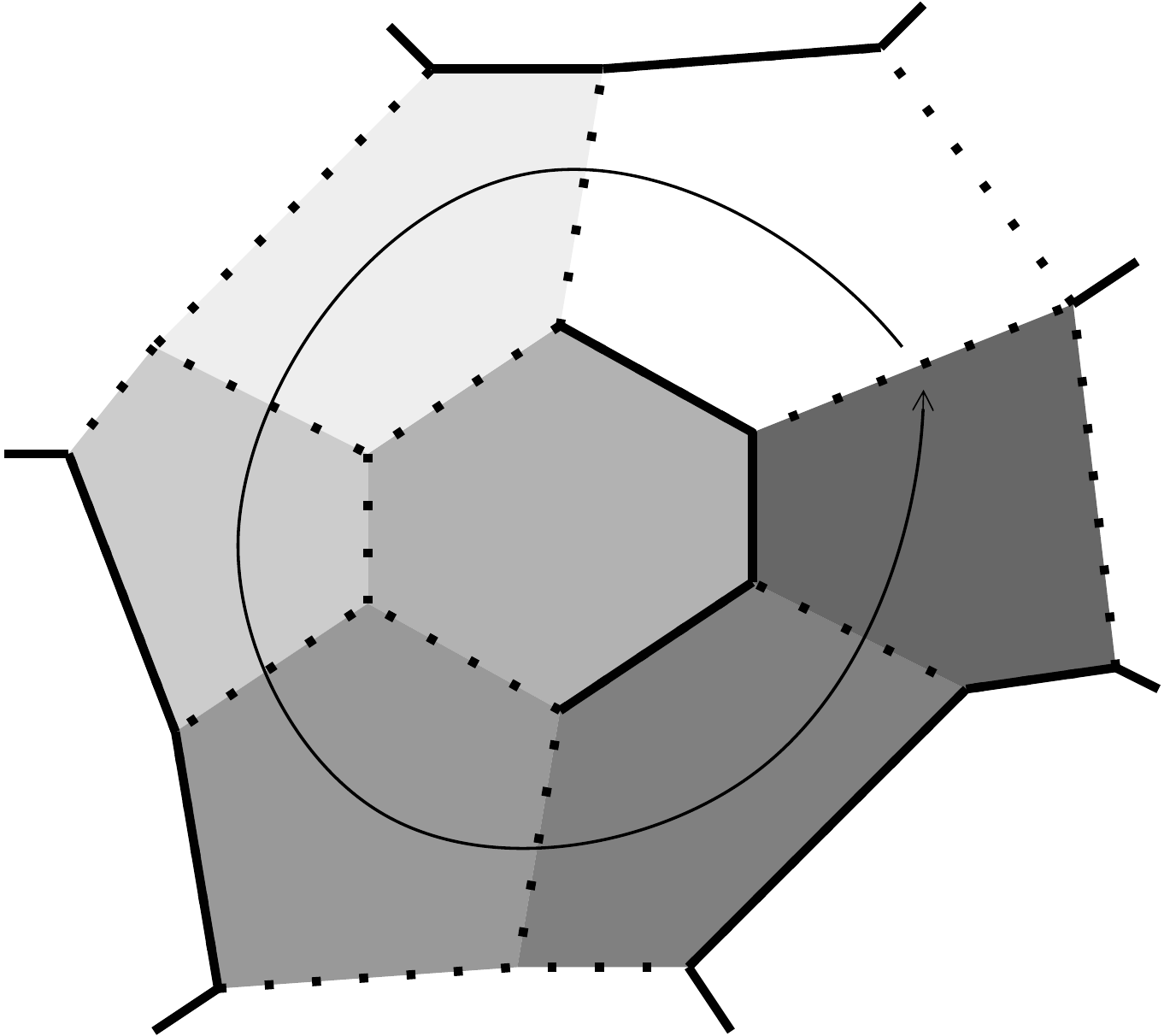}}
\put(.06,.04){\small$\gamma$}
\put(.2,.205){\small$G\!=\!I$}
\put(.23,.14){\small$G\!=\!R$}
\end{picture}
\caption{Along a closed path $\gamma$ traversing a sequence of crystal grains, the deformation gradient $G$ changes continuously from $I$ to $R\neq I$.
The gray shade indicates the local crystal orientation from the identity $I$ (white) to $R$ (dark gray).
Dots represent point dislocations; lines indicate high angle grain boundaries.
Along the path $\gamma$ all grains are connected by low angle grain boundaries.}
\label{fig:pointGroup}
\end{figure}

Note that we have large freedom to choose the cut set $\jumpset$.
For instance we could take $\jumpset$ to be composed of smooth lines
connecting the different connected components $\Omega_d^i$ of $\Omega_d$.
Those lines can easily be chosen in such a way that they do not intersect
and that the connected components, say $\Omega^i$, of $\Omega\setminus(\Omega_d\cup S)$ are simply connected.
Within each $\Omega^i$, the crystal is defect-free,
thus there is a deformation $\phi^i$ mapping $\Omega^i$ onto the reference crystal.
The matrix field $G$ will be the gradient of $\phi^i$ inside $\Omega^i$,
which is automatically consistent all over $\Omega^i$.
If $G$ jumps across $\jumpset$ between two neighboring components $\Omega^i$ and $\Omega^j$,
then for all $x\in S^{ij}:=\partial\Omega^i\cap\partial\Omega^j$ we must have $G^i(x)n^\perp(x)=RG^j(x)n^\perp(x)$ with a constant $R\in\pointgroup$,
since $S^{ij}$ lies inside a defect-free crystal region
($G^i$ and $G^j$ denote the value of $G$ in $\Omega^i$ and $\Omega^j$, respectively, while $n^\perp$ is the tangent vector to $\jumpset$).
Therefore, not only is $S$ smooth, but even $R$ constant on each curve segment of $S$.
Furthermore, up to multiplications with point group elements,
the optimal $G$ does not depend on the particular position and geometry of $S$.
Indeed, if $S^{ij}$ is shifted, producing new $\hat\Omega^i$, $\hat\Omega^j$, $\hat S^{ij}$,
then the new optimal field $\hat G$ is given by
\begin{equation*}
\hat G(x)=\begin{cases}
G(x)&\text{if }x\in\hat\Omega^i\cap\Omega^i\text{ or }x\in\hat\Omega^j\cap\Omega^j,\\
RG(x)&\text{if }x\in\hat\Omega^i\cap\Omega^j,\\
R^{-1}G(x)&\text{if }x\in\hat\Omega^j\cap\Omega^i.
\end{cases}
\end{equation*}%

For simplicity, the cut set $\jumpset$ in our computation is algorithmically chosen after the problem discretization.
First we sweep over all image pixels and reassign the pixel values via a fast marching type method.
Let $Q$ denote the set of already treated pixels (consisting initially of only one pixel).
Amongst the not yet treated neighbors of all pixels in $Q$, we pick the pixel $x$ with the smallest value of $\min_{R\in\pointgroup}|RG_0(x)-G_0(x_n)|$,
where $x_n\in Q$ denotes the neighboring pixel.
We replace its value $G_0(x)$ with $RG_0(x)$ for the minimizer $R\in\pointgroup$ and set $Q=Q\cup\{x\}$.
After all pixels have been swept through in this way, we set $S$ to be the union of interfaces between neighboring pixels $x_1,x_2$
that have $\min_{R\in\pointgroup}|RG_0(x_1)-G_0(x_2)|<|G_0(x_1)-G_0(x_2)|$.
Note that we also take care that $\Omega_d\cap\jumpset=\emptyset$ so that the
estimation of Burgers vectors is not impaired.

\subsubsection{Variational principle with topological jump set}

While we will use $G$ on the whole domain $\Omega$ in the numerical
algorithm, as the energy \eqref{eq:energyelastic} depends only on $G$
on $\Omega \setminus \Omega_d$, it is more convenient
to reformulate the problem with a bounded Lipschitz
domain $\Omega_0 := \Omega\setminus( \Omega_d \cup \jumpset)$ and
consider the Hilbert space
\begin{equation}
  H(\curl, \Omega_0)  = \Bigl\{ u
  \in 
  (L^2(\Omega_0))^2 \; \Big\vert \; \curl u \in L^2(\Omega_0) \Bigr\}, 
\end{equation}
with norm given by 
\begin{equation}
  \norm{u}_{H(\curl, \Omega_0)} = \norm{u}_{(L^2(\Omega_0))^2} + \norm{\curl u}_{L^2(\Omega_0)}.
\end{equation}
Note that $\partial \Omega_0 = S \cup \partial \Omega_d = S \cup
\bigcup_{i} \partial \Omega_d^i$, where we recall that $\Omega_d^i$ are
the connected components of $\Omega_d$. We define the tangential trace
operators $\gamma_t^-$ and $\gamma_t^+$ on the two sides of $S$ for $u
\in (C^{\infty}(\Omega_0))^2$ as
\begin{equation}
  \gamma_t^{-}(u) = 
  n^{\perp} \cdot  u \vert_{S^-}
  \quad \text{and} \quad 
  \gamma_t^{+}(u) = 
  n^{\perp} \cdot  u \vert_{S^+}, 
\end{equation}
where $n = (n_1, n_2)$ is the unit normal of $S$ and $n^{\perp} =
(-n_2, n_1)$. It is standard (see e.\,g., \cite{DuvautLions:77}) using
Green formula that $\gamma_t^{\pm}$ can be extended to bounded linear
operators from $H(\curl, \Omega_0)$ to the fractional Hilbert space $H^{-1/2}(S)$. Similarly, we
may also define the trace operator $\gamma_t$ on $\partial \Omega_d$.

For the (inverse) deformation gradient $G$, it is then natural to
consider the space $H(\curl, \Omega_0)^2$ as the rows of $G$: $G^{1:}$
and $G^{2:}$ lie in $H(\curl, \Omega_0)$. Thus, we generalize
the tangential trace operators $\gamma_t^{\pm}$ and $\gamma_t$ also to
$G$, such that they act on each row of $G$ and map $H(\curl,
\Omega_0)^2$ to $(H^{-1/2}(S))^2$ and $(H^{-1/2}(\partial
\Omega_d))^2$, respectively.  In particular, as we are only concerned with $G$ on
$\Omega_0$, the constraint $\curl G = b$ on $\Omega_d$ is equivalent to
$\int_{\partial \Omega_d^i} \gamma_t(G) \ud \sigma = \int_{\Omega_d^i}
b_i \ud x$.

With this, we may formulate \eqref{eq:energyelastic} with the
topological jump set more precisely as
\begin{align}\label{eq:energySH}
  \min_{G \in H(\curl, \Omega_0)^2} &\int_{\Omega\setminus\Omega_{d}} \abs{G - G_0}^2 + W(G) \ud x \\
  \text{s.\,t.\ }& \curl G = 0 \text{ on }\Omega_0\, \notag\\
  & \int_{\partial \Omega_d^i} \gamma_t(G) \ud \sigma =
  \int_{\Omega_d^i} b_i \ud x\, \notag \\
  & \gamma_t^-(G) = R \gamma_t^+(G) \text{ on }\jumpset\,,\notag
\end{align}
for some fixed piecewise constant $R: S \to P$. Here the last equation
$\gamma_t^-(G) = R \gamma_t^+(G)$ is understood in the space of
$(H^{-1/2}(S))^2$ (this is well defined as $R f \in (H^{1/2}(S))^2$
for any $f \in (H^{1/2}(S))^2$ due to the regularity of $R$).

\begin{theorem}
  Assume that the topological jump set $\jumpset$ and the boundary of
  $\Omega_d$ are Lipschitz.  Given $G_0 \in \mc{M}_{2 \times
    2}(L^2(\Omega \setminus \Omega_d))$, $R: S \to P$ piecewise
  constant, and $b$ a finite measure on $\Omega_d$, the minimizer to
  the variational problem \eqref{eq:energySH} exists.
\end{theorem}
\begin{proof}
  We consider the admissible space of $G$:
  \begin{multline*}
    \mc{A} = \Bigl\{ G \in H(\curl, \Omega_0)^2 \;\Big\vert\; \curl G = 0 \text{ on } \Omega_0,\, 
    \int_{\partial \Omega_d^i} \gamma_t(G) \ud \sigma =
    \int_{\Omega_d^i} b_i \ud x, \, \\
    \gamma_t^-(G) = R \gamma_t^+(G) \text{ on } \jumpset \Bigr\}.
  \end{multline*}
  Note that this is a closed linear subspace of $H(\curl, \Omega_0)^2$, and
  hence any bounded sequence is weakly compact in $\mc{A}$. Given
  $\{G^{(n)}\}$ a minimizing sequence of \eqref{eq:energySH}, as the
  energy is uniformly bounded, we have
  \begin{equation*}
    \int_{\Omega \setminus \Omega_{d}} \abs{G^{(n)} - G_0}^2 \ud x \leq C 
  \end{equation*}
  for some constant $C$. Together with $\curl G^{(n)} = 0$, we obtain
  that $\norm{G^{(n)}}_{H(\curl, \Omega_0)^2}$ is uniformly bounded, and
  hence weakly convergent (in $H(\curl, \Omega_0)^2$) to $G^{(\infty)}
  \in \mc{A}$. The existence then follows from the lower
  semi-continuity of $\int \abs{G^{(n)} - G_0}^2\ud x$ and $\int W(G)\ud x$ with
  respect to the weak $(L^2(\Omega_0 \setminus
  \Omega_d))^{2\times 2}$ topology (which is clearly weaker than the weak topology
  of $H(\curl, \Omega_0)^2$).
\end{proof}

\subsection{Constrained minimization algorithm}\label{sec:minimize}
In this section we describe the numerical algorithm to solve the
minimization problem \eqref{eq:energyelastic} or rather the
corresponding version with topological jump set $\jumpset$
\eqref{eq:energySH}. It is more convenient to
describe the algorithm for the following simplified formulation,
which is equivalent after discretization.
\begin{align}
\label{eq:energyS}
\min_{G:\Omega\to\RR^{2\times2}}&\int_{\Omega\setminus\Omega_{d}} \abs{G - G_0}^2 + W(G) \ud x \\
\text{s.\,t.\ }&\curl G = b\text{ on }\Omega\setminus\jumpset\,,\quad
G^- = R G^+ \text{ on }\jumpset\,,\nonumber
\end{align}
where $G^-$ and $G^+$ denote the value of $G$ on either side of $\jumpset$.
For simplicity we use the same notation for the continuous and discrete objects.

\subsubsection{Finite difference discretization}
The image domain $\Omega$ is discretized via $M\cdot N$ Cartesian pixels,
indexed by $x$-$y$-position $(m,n)\in\Omega=\{1,\ldots,M\}\times\{1,\ldots,N\}$
(the pixel spacing is assumed to be one).
For an index $m$ we denote by $m^+=m+1$ and $m^-=m-1$ the next larger or next smaller index,
where for simplicity we assume periodic boundary conditions and use cyclic indexing, i.\,e.\ $(M^+,n)=(1,n)$, $(m,N^+)=(m,1)$, $(1^-,n)=(M,n)$, $(m,1^-)=(m,N)$ for all $(m,n)\in\Omega$.
The matrix fields are discretized accordingly, $G,G_0\in(\RR^{2\times2})^{M\times N}$.
The jump set $\jumpset$ follows the edges between the pixels
and is represented as a collection of horizontal or vertical pixel pairs,
$\jumpset\subset\{((m,n),(k,l))\in\Omega\times\Omega\ :\ (k,l)=(m^+,n)\text{ or }(k,l)=(m,n^+)\}$.
Furthermore, we define the function $R:S\to\pointgroup$
such that $R_{(m,n),(k,l)}$ is the point group element with smallest distance to $(G_0)_{m,n}(G_0)_{k,l}^{-1}$.
$R$ is extended to $\Omega\times\Omega$ by the identity in $P$.
Derivatives in $x$- and $y$-direction are replaced by finite differences that respect the point group equivalence across $\jumpset$,
\begin{eqnarray*}
(\partialDisc_x G)_{m,n}^{ij}=(R_{(m,n),(m^+,n)}G_{m^+,n})^{ij}-G_{m,n}^{ij}\,,\\
(\partialDisc_y G)_{m,n}^{ij}=(R_{(m,n),(m,n^+)}G_{m,n^+})^{ij}-G_{m,n}^{ij}\,,
\end{eqnarray*}
where superscript $ij$ denotes the $(i,j)$-matrix entry.
Throughout, a superscript $d$ denotes discrete differential operators.
In particular, the discrete curl and Laplacian are defined as
\begin{align*}
&\curlDisc:(\RR^{2\times2})^{M\times N}\to(\RR^2)^{M\times N}\,,\\
&\hspace{20ex}(\curlDisc G)_{m,n}=\partialDisc_xG_{m,n}^{:2}-\partialDisc_yG_{m,n}^{:1}\,,\\
&\DeltaDisc:(\RR^2)^{M\times N}\to(\RR^2)^{M\times N}\,,\\
&\hspace{20ex}\DeltaDisc=-\curlDisc(\curlDisc)^*\,,
\end{align*}
where $G_{m,n}^{:i}$ denotes the $i$\textsuperscript{th} column of the
matrix $G_{m,n}$ and the superscript $*$ denotes the adjoint operator, which
in this particular case is given by
\begin{align*}
&(\curlDisc)^*:(\RR^2)^{M\times N}\to(\RR^{2\times2})^{M\times N}\,,\\
&((\curlDisc)^*V)_{m,n}=\\
&\left[\left.V_{m,n}-R_{(m,n^-),(m,n)}^TV_{m,n^-}\right|R_{(m^-,n),(m,n)}^TV_{m^-,n}-V_{m,n}\right]\,.
\end{align*}

\subsubsection{Projected descent in constraint space}
The discrete optimization problem reads
\begin{eqnarray*}
&\min_{G\in C_b}E[G]\\
&\text{for }E[G]=\sum_{(m,n)\in\Omega}\left(\abs{G_{m,n} - (G_0)_{m,n}}^2 + W(G_{m,n})\right)\\
&\text{and }C_b=\{G\in(\RR^{2\times2})^{M\times N}\ :\ \curlDisc G=b\}\,.
\end{eqnarray*}
The constraint space $C_b$ is an affine space and can be expressed as $\hat G+C_0$ for a $\hat G$ with $\curlDisc\hat G=b$.
Hence, the energy can be minimized using a standard projected nonlinear conjugate gradient (NCG) descent in $C_b$.
In more detail, we employ a Fletcher--Reeves NCG method
in which the derivative $\partialDisc_GE$ of $E$ with respect to $G$ is always orthogonally projected onto $C_0$
(i.\,e.\ onto its component parallel to $C_b$)
so that the algorithm is performed within the subspace $C_b$.
Due to accumulating numerical errors we also have to project the current estimate $G$ back onto $C_b$ from time to time.
Denoting the projection onto $C_b$ by $\proj{C_b}$,
the NCG algorithm is initialized with $\hat G=\proj{C_b}G_0$.

The projection $\proj{C_b}F$ is the solution to the constraint minimization $\min_{\curlDisc G=b}\sum_{m,n}\abs{G_{m,n}-F_{m,n}}^2$,
which satisfies the optimality conditions
\begin{equation*}
b=\curlDisc F\,,\qquad0=F-G+(\curlDisc)^*\Lambda
\end{equation*}
for a Lagrange multiplier $\Lambda\in(\RR^2)^{M\times N}$.
Applying $\curlDisc$ to the second equation we obtain
\begin{equation*}
\curlDisc G-b=-\DeltaDisc\Lambda\,.
\end{equation*}
Note that $\ker(\DeltaDisc)\perp\range(\curlDisc)$.
Denoting by $(-\DeltaDisc)^{-1}:\range(\DeltaDisc)\to\ker(\DeltaDisc)^\perp$ the inverse of $-\DeltaDisc$, we obtain $\Lambda=(-\DeltaDisc)^{-1}(\curlDisc G-b)$ and thus
\begin{equation*}
\proj{C_b}G=F=G-(\curlDisc)^*(-\DeltaDisc)^{-1}(\curlDisc G-b)\,.
\end{equation*}
Once $\hat G$ is computed, the projection onto $C_b$ can also be obtained as
\begin{equation*}
\proj{C_b}F=\hat G+\proj{C_0}(F-\hat G)\,.
\end{equation*}

\subsubsection{Fast projection onto constraint space}
For the above projection we need to invert the discrete Laplacian operator.
Using periodic boundary conditions this would be very fast using FFT if there was no jump set $\jumpset$.
However, with nonempty jump set $\jumpset$, our finite difference operators do not turn into pointwise multiplications in Fourier space.
In order to obtain a fast inversion we decompose $\DeltaDisc=\DeltaDisc_0+J$,
where $\DeltaDisc_0$ is the standard discrete Laplacian and $J$ the linear operator accounting for the point group,
\begin{eqnarray*}
&\DeltaDisc_0:(\RR^2)^{M\times N}\to(\RR^2)^{M\times N}\,,\\
&(\DeltaDisc_0V)_{m,n}=V_{m^-,n}+V_{m^+,n}+V_{m,n^-}+V_{m,n^+}-4V_{m,n}\,,
\end{eqnarray*}
\begin{align*}
&J:(\RR^2)^{M\times N}\to(\RR^2)^{M\times N}\,,\ \\
(JV)_{m,n}&=(R_{(m^-,n),(m,n)}^T-I)V_{m^-,n}\\
&+(R_{(m,n),(m^+,n)}-I)V_{m^+,n}\\
&+(R_{(m,n^-),(m,n)}^T-I)V_{m,n^-}\\
&+(R_{(m,n),(m,n^+)}-I)V_{m,n^+}\,.
\end{align*}
Note that $J$ is symmetric, and it is highly sparse and thus has a very small range $\rngJ=\range J$ and a large kernel $\ker J=\ker J^\TT=\rngJ^\perp$.
If we decompose $V\in(\RR^2)^{M\times N}$ into
\begin{equation*}
V=V_\rngJ+V_{\rngJ^\perp}\in\rngJ\oplus\rngJ^\perp\,,
\end{equation*}
then we obtain
\begin{equation}\label{eqn:splitSystem}
-\DeltaDisc V=B
\qquad\Leftrightarrow\qquad
-\DeltaDisc_0V_{\rngJ^\perp}=(\DeltaDisc_0+J)V_\rngJ+B\,.
\end{equation}
The solvability condition tells us that $\ker\DeltaDisc_0$ is orthogonal to the right-hand side,
so we can just as well project the right-hand side onto $(\ker\DeltaDisc_0)^\perp$ by subtracting the mean,
\begin{equation*}
-\DeltaDisc_0V_{\rngJ^\perp}=\proj{(\ker\DeltaDisc_0)^\perp}\left((\DeltaDisc_0+J)V_\rngJ+B\right)\,.
\end{equation*}
Now choosing a basis $\rngJ=\langle v_1,\ldots,v_K\rangle$ and $\ker\DeltaDisc_0=\langle s_1,s_2\rangle$, we write
\begin{eqnarray*}
V_\rngJ=\sum_{i=1}^K\lambda_iv_i\,,
\end{eqnarray*}
\begin{align*}
V_{\rngJ^\perp}&= a_1s_1+a_2s_2+(-\DeltaDisc_0)^{-1}\proj{(\ker\DeltaDisc_0)^\perp}\left((\DeltaDisc_0+J)V_\rngJ+B\right)\\
&=a_1s_1+a_2s_2+(-\DeltaDisc_0)^{-1}\proj{(\ker\DeltaDisc_0)^\perp}B\\
&\qquad +\sum_{i=1}^K\lambda_i\left((-\DeltaDisc_0)^{-1}\proj{(\ker\DeltaDisc_0)^\perp}J v_i-v_i\right)\,,
\end{align*}
where $(-\DeltaDisc_0)^{-1}:(\ker\DeltaDisc_0)^\perp\to(\ker\DeltaDisc_0)^\perp$ denotes the inverse of $-\DeltaDisc_0$.
The degrees of freedom $\lambda_1,\ldots,\lambda_K,a_1,a_2$ now have to satisfy the $K+2$ equations
$v_j\cdot V_{\rngJ^\perp}=0$, $j=1,\ldots,K$, (due to $v_j\in\rngJ$) and $s_j\cdot \left((\DeltaDisc_0+J)V_\rngJ+B\right)=0$, $j=1,2$, (the solvability condition for \eqref{eqn:splitSystem})
where the dot denotes the dot product.
In detail, the equations are given by
\begin{equation}\label{eqn:invertLaplace}
\left.
\begin{array}{rcl}
0&=&a_1v_j\cdot s_1+a_2v_j\cdot s_2+v_j\cdot (-\DeltaDisc_0)^{-1}\proj{(\ker\DeltaDisc_0)^\perp}B\\
& &+\sum_{i=1}^K\lambda_i\left(\left((-\DeltaDisc_0)^{-1}\proj{(\ker\DeltaDisc_0)^\perp}v_j\right)\cdot J v_i-v_j\cdot v_i\right)\,,\\
0&=&s_j\cdot B+\sum_{i=1}^K\lambda_is_j\cdot J v_i\,.
\end{array}
\right\}
\end{equation}
To solve for $\lambda:=(\lambda_1,\ldots,\lambda_K,a_1,a_2)^T$, we write the linear system as a matrix vector equation.
Since we need to solve systems of the form $-\DeltaDisc V=B$ several times for different values of $B$, we perform a QR decomposition,
\begin{equation*}
QRP\lambda=\text{rhs}\,,
\end{equation*}
where $P$ is a permutation matrix such that the lower rows of $R$ are zero
(this is possible since the solution to $-\DeltaDisc V=B$ is only uniquely specified up to a two-dimensional subspace due to $\dim\ker\DeltaDisc=2$)
and all other rows have nonzero diagonal elements.
Note that this decomposition is the bottleneck of the algorithm with a complexity of $O(K^3)$.
The degrees of freedom corresponding to the zero-rows can now be chosen freely (say, as zero), and the resulting
\begin{align*}
V&=a_1s_1+a_2s_2+(-\DeltaDisc_0)^{-1}\proj{(\ker\DeltaDisc_0)^\perp}B\\
& \qquad +\sum_{i=1}^K\lambda_i(-\DeltaDisc_0)^{-1}\proj{(\ker\DeltaDisc_0)^\perp}J v_i\\
&= a_1s_1+a_2s_2+(-\DeltaDisc_0)^{-1}\proj{(\ker\DeltaDisc_0)^\perp}\left[\sum_{i=1}^K\lambda_iJ v_i+B\right]
\end{align*}
solves $B=-\DeltaDisc V$.
Note that $(-\DeltaDisc_0)^{-1}$ can readily be computed via FFT
and that the $v_i$ can be chosen as shifted versions of $\hat v^j\in(\RR^2)^{M\times N}$, $j=1,2$,
with $\hat v^j=0$ except for the $j$\textsuperscript{th} entry of $\hat v^j_{1,1}$ being one.
Thus, $(-\DeltaDisc_0)^{-1}\proj{(\ker\DeltaDisc_0)^\perp}v_i$ can be evaluated efficiently as just a shift of $(-\DeltaDisc_0)^{-1}\proj{(\ker\DeltaDisc_0)^\perp}\hat v^j$.

\section{Numerical results}

In this section, we present several numerical examples for images
coming from both computer simulations and real experiments to
illustrate the performance of our method. The corresponding code is open source and available as SynCrystal at
\url{https://github.com/SynCrystal/SynCrystal}. We first apply our method to
analyze synthetic crystal images of phase field crystals
(PFC)~\cite{ElderGrantMartin:04}. To show the robustness of our
method, both noiseless and noisy examples of PFC images are
presented. Moreover, we examine different kinds of crystal images
from real experiments: (1) a TEM-image of GaN with a series of
isolated defects that leads to a large angle grain boundary; (2) a
photograph of a bubble raft with a thick transition region at grain
boundaries; (3) a TEM-image of Al with a noisy and irregular grain
boundary.

In all of these examples, we compare the initial estimated strain $G_0$ provided by the SST-based analysis described in \S\ref{sec:synsquez} and the improved results $G$ from the variational method \eqref{eq:energyS}, where each time we display crystal orientations, difference in principal stretches, and volume distortion. For better visualization we mask out the identified defect regions (in which there is no meaningful notion of strain). 
The curl of $G_0$ (which in general violates the physical constraint of being zero outside $\Omega_d$ and of being compatible with the defects' Burgers vectors)
as well as the average $\curl G$ per connected defect region (which is compatible with the defects' Burgers vectors) are also shown.

The main parameters for the SST-based analysis are two geometric scaling parameters $s$ and $t$ in the 2D SST (for details see \cite{Robustness}). Smaller scaling parameters result in better robustness of SSTs while larger scaling parameters give more accurate estimates in noiseless cases \cite{Robustness}. Hence, we adopt $t=s\approx 1$ in the examples with less noise and use $t=s\approx 0.8$ when the crystal image is noisy. As discussed in \cite{Robustness}, for images with heavy noise, the synchrosqueezed transform can still provide reasonable initial guess via a highly redundant transform with more computational cost. The variational model parameters $\lambda$ and $\mu$ in \eqref{eq:elasticEnergy} are simply set to $1$ in all of the following examples.

During the synchrosqueezed transform \notinclude{ very efficient as shown in \cite{SSCrystal}. It }we also downsample the original image in the frequency domain and reduce the mesh size for the variational optimization. As a result, the computational mesh of the optimization model is actually independent of the resolution of the crystal images.
To achieve a short computation time, we typically choose the downsampling rate as large as possible such that defects can still be localized with a precision below one atom spacing.
Figure\,\ref{fig:comp} will show a comparison between different downsampling rates.

The main computational cost arises from inverting the Laplace operator via the solution of \eqref{eqn:invertLaplace}.
The number of variables is proportional to the number $K$ of pixel pairs of the singularity set $\jumpset$, and the solution cost scales like $K^3$.
Note that, since $\jumpset$ has codimension $1$, $K$ behaves roughly like the number of image pixels along one direction (and not like the total number of pixels),
so the computations are still feasible for quite large images.
The runtime for the examples in Figure\,\ref{fig:Real1img} is less than $10$ seconds,
the runtime for the $1024\times1024$ PFC image in Figure\,\ref{fig:PFC1img}, left, is $20$ seconds.
The runtime for the example in Figure\,\ref{fig:Real2img} with heavy noise is about $1$ minute due to extra effort to obtain robustness of the synchrosqueezed transform.

\subsection{Synthetic crystal images}

\begin{figure}
\begin{center}
   \begin{tabular}{cc}
       \includegraphics[height=2.4in]{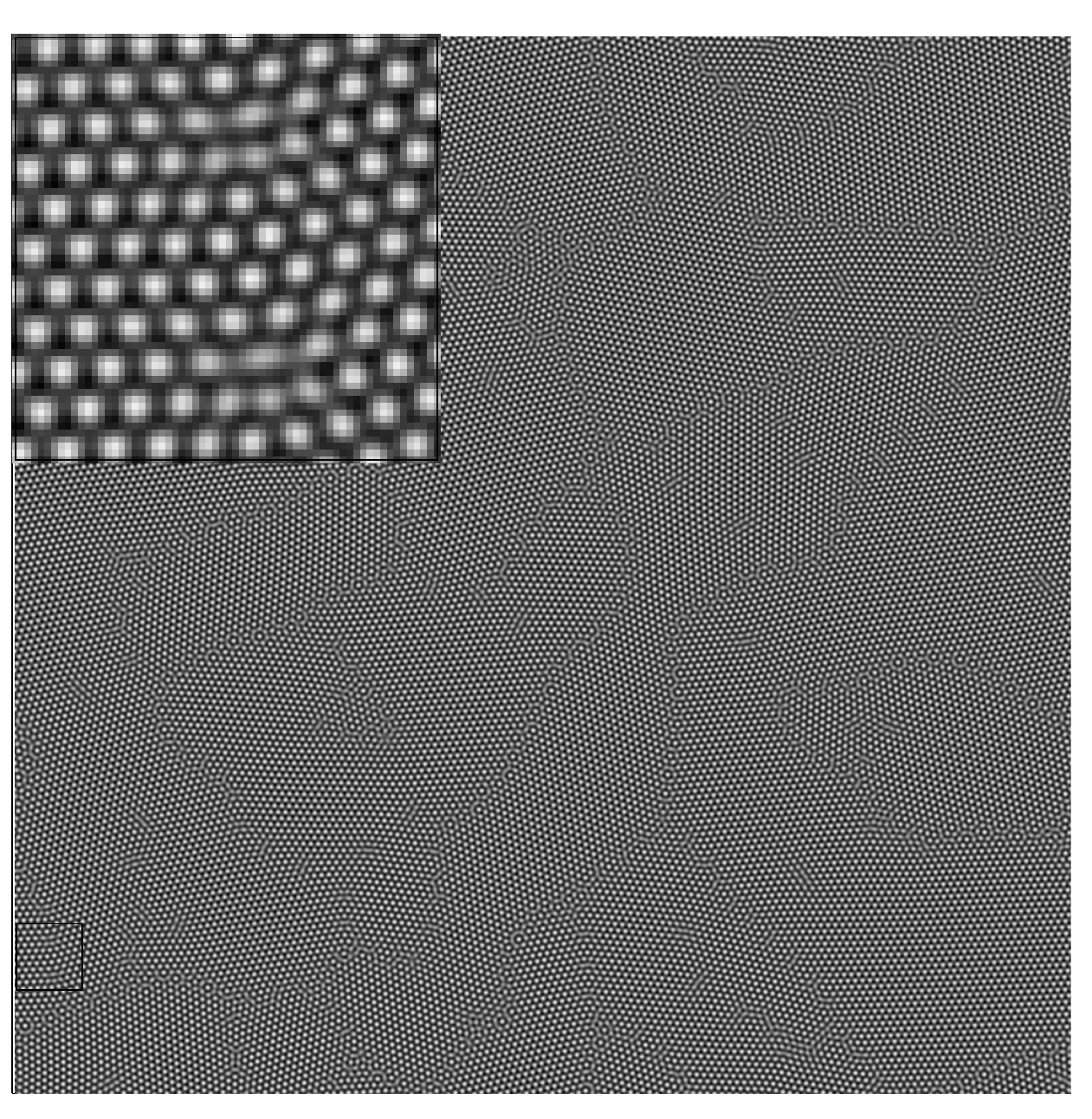}&\includegraphics[height=2.4in]{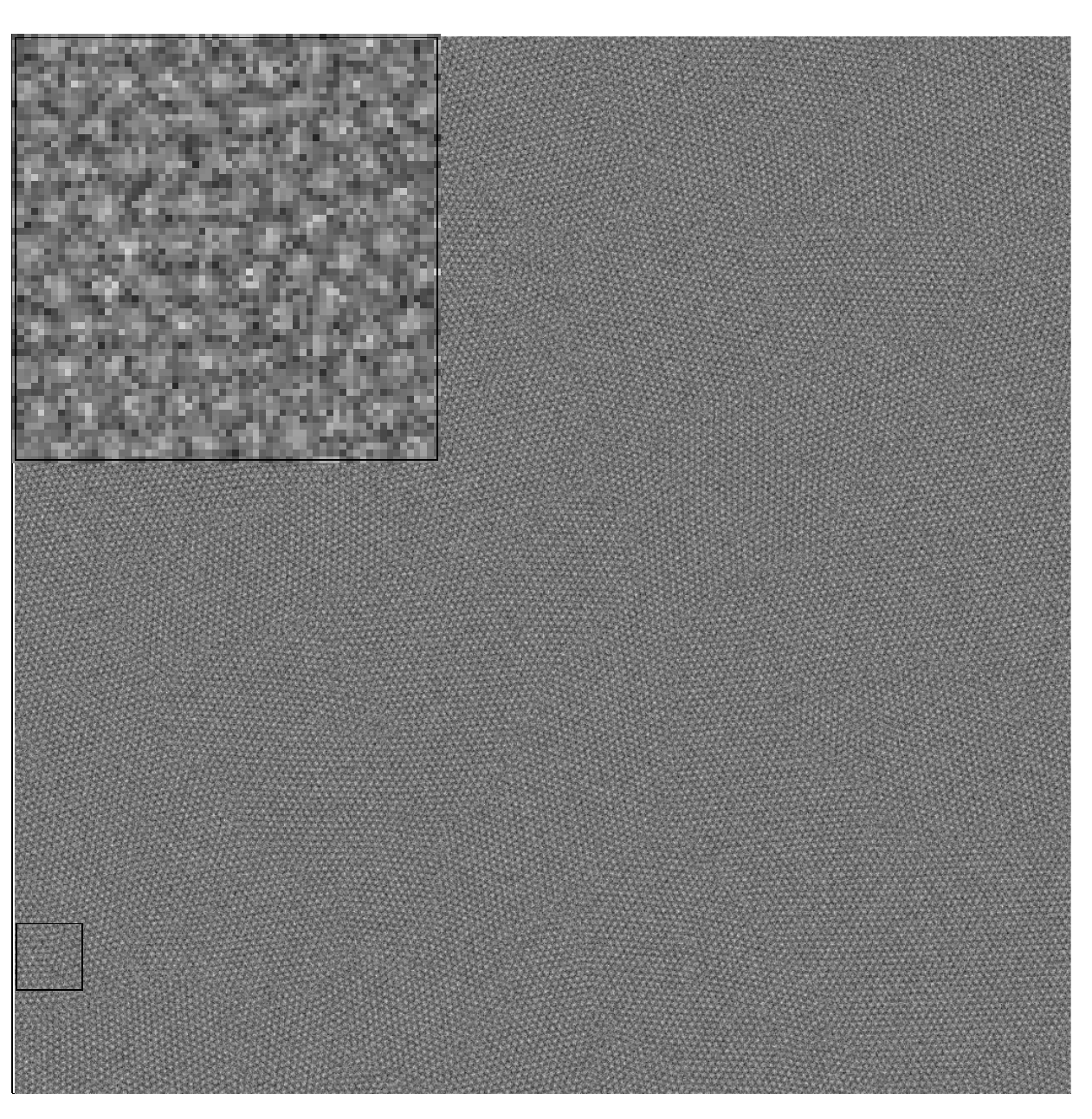}
    \end{tabular}
 \end{center}
 \caption{A noise-free PFC image (left) and its noisy version (right) with a zoom-in detailing the marked part.}
 \label{fig:PFC1img}
  \end{figure}
       
\begin{figure}
\begin{center}
   \begin{tabular}{cc}
        \includegraphics[height=1.2in]{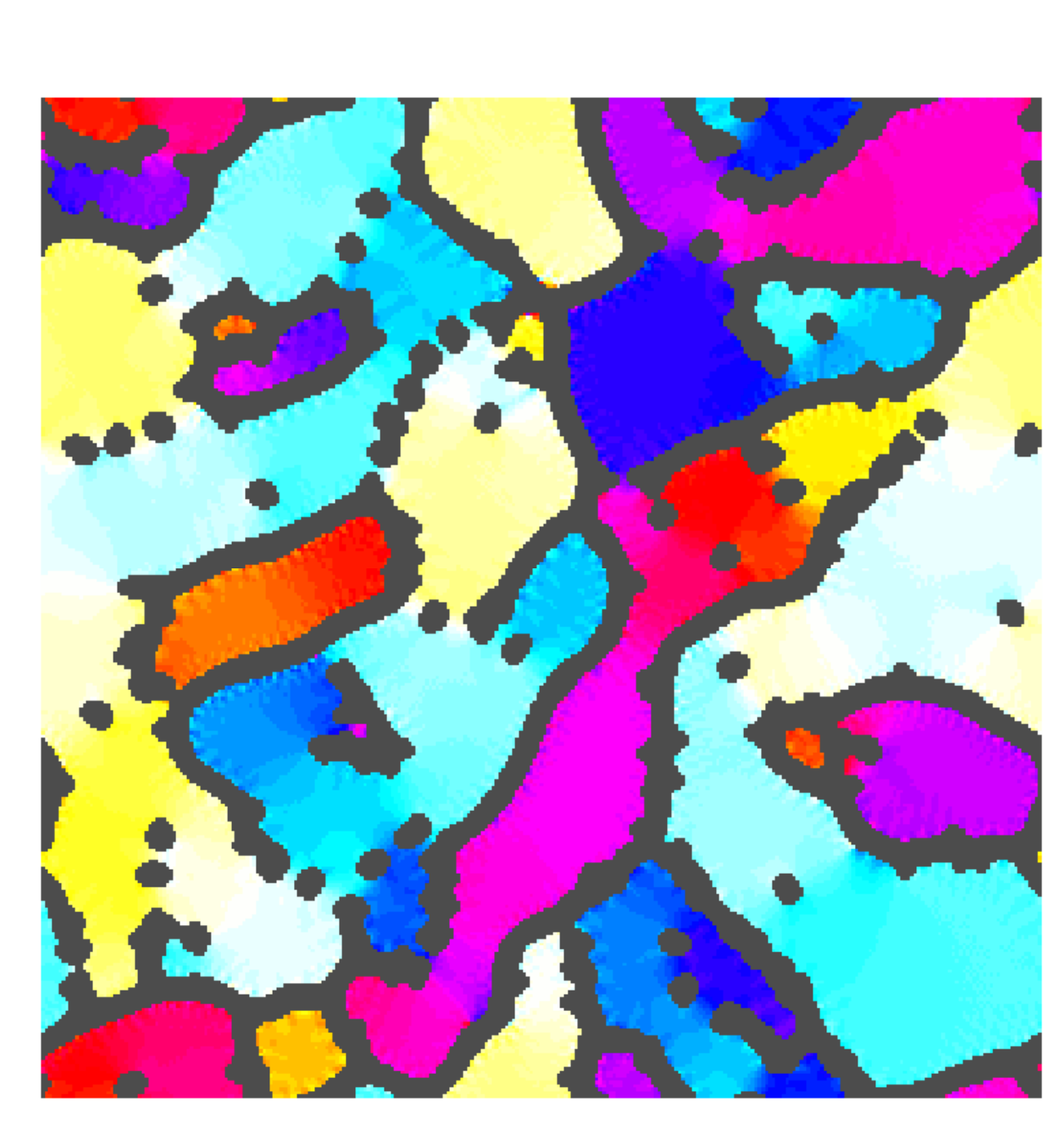} \includegraphics[height=1.2in]{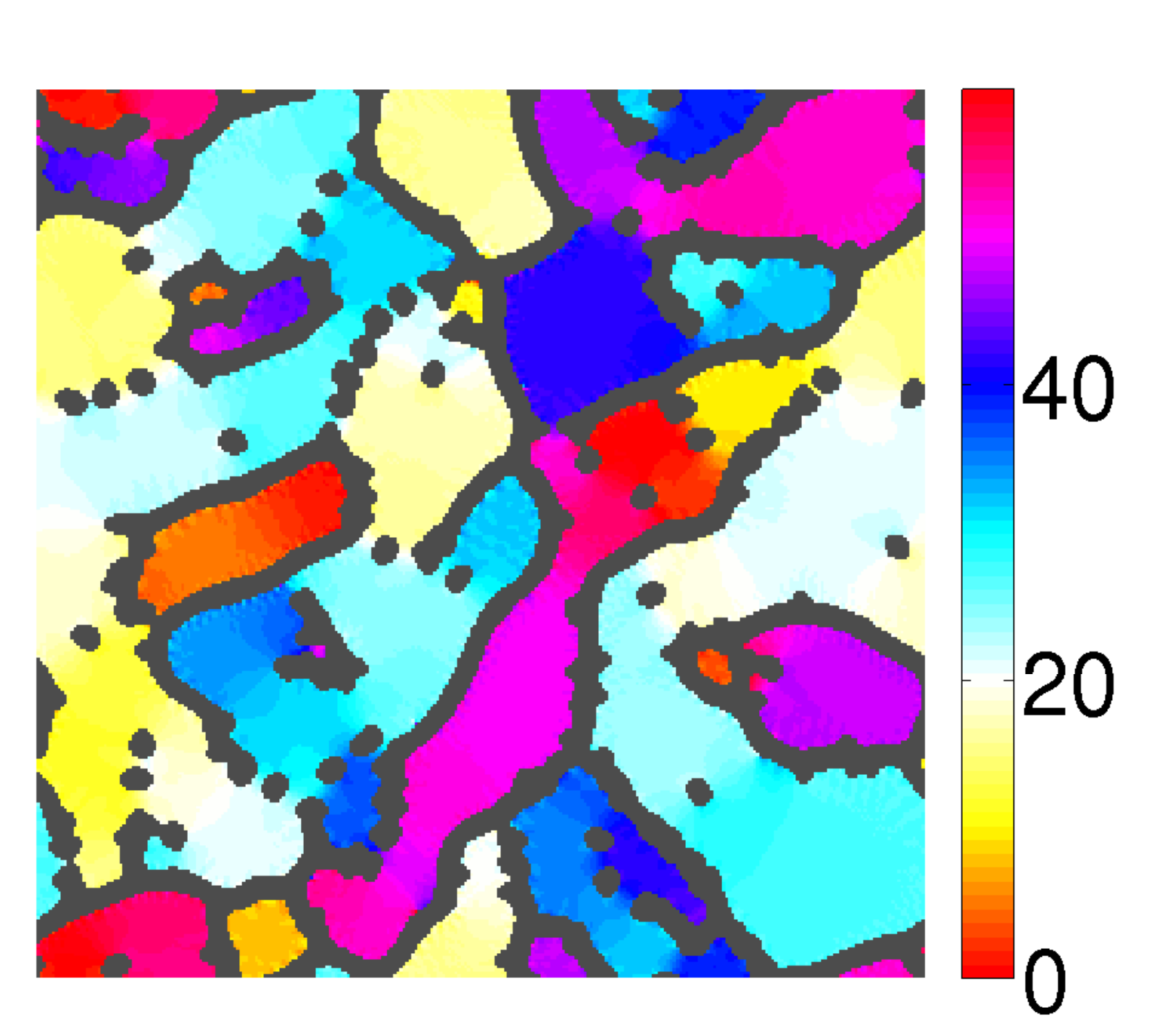}  &
         \includegraphics[height=1.2in]{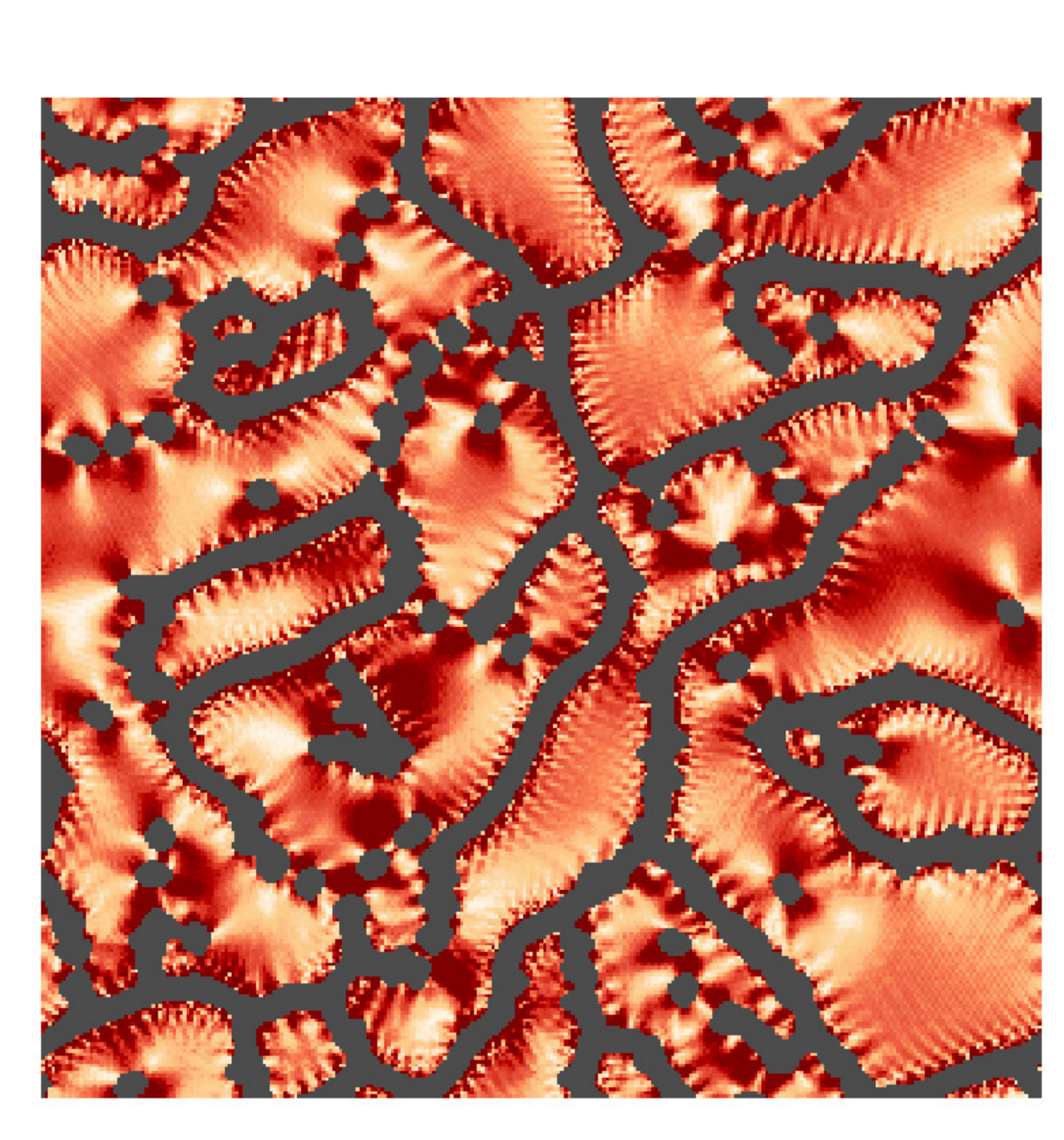} \includegraphics[height=1.2in]{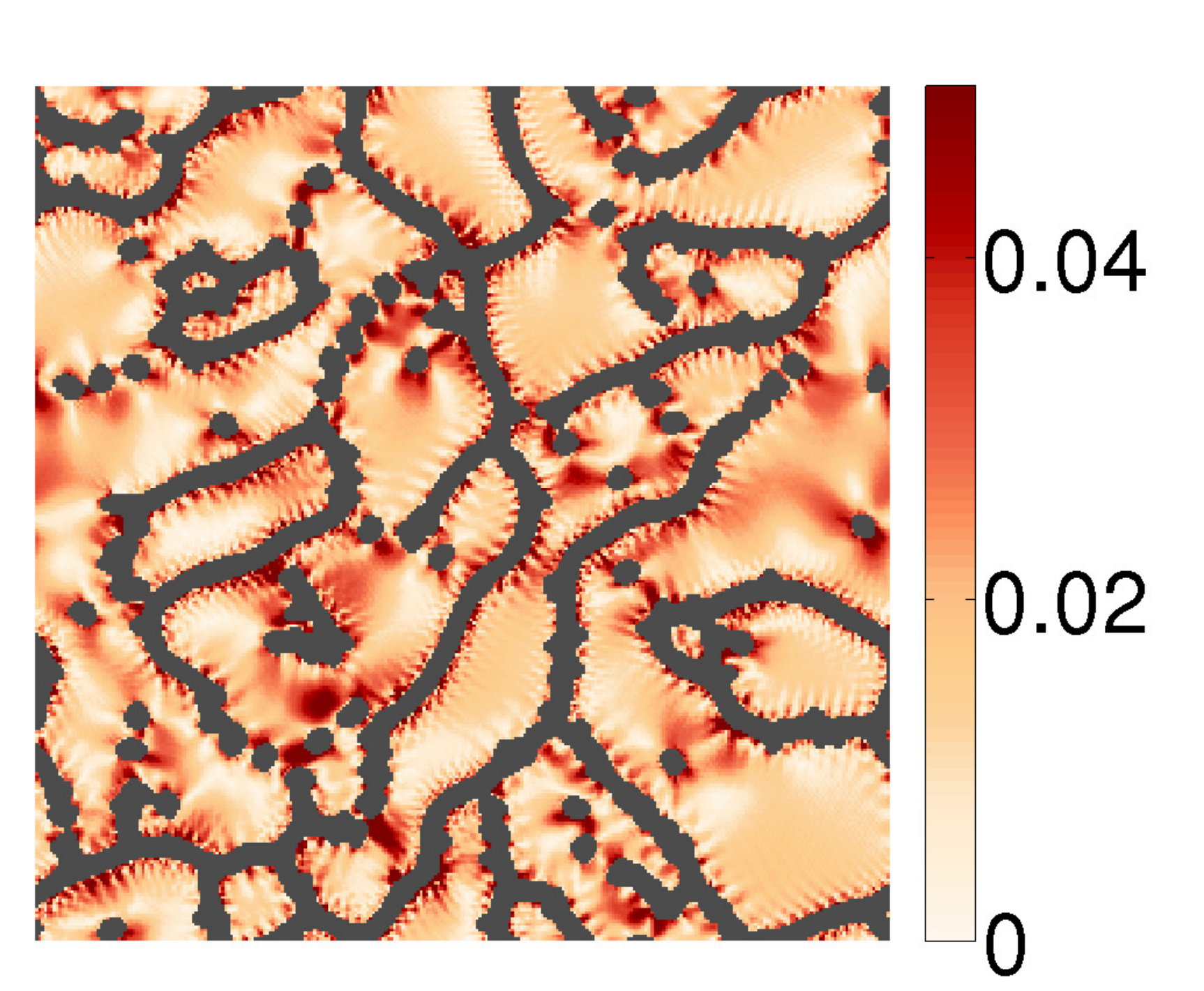} \\
         (a) & (b)\\
         \includegraphics[height=1.2in]{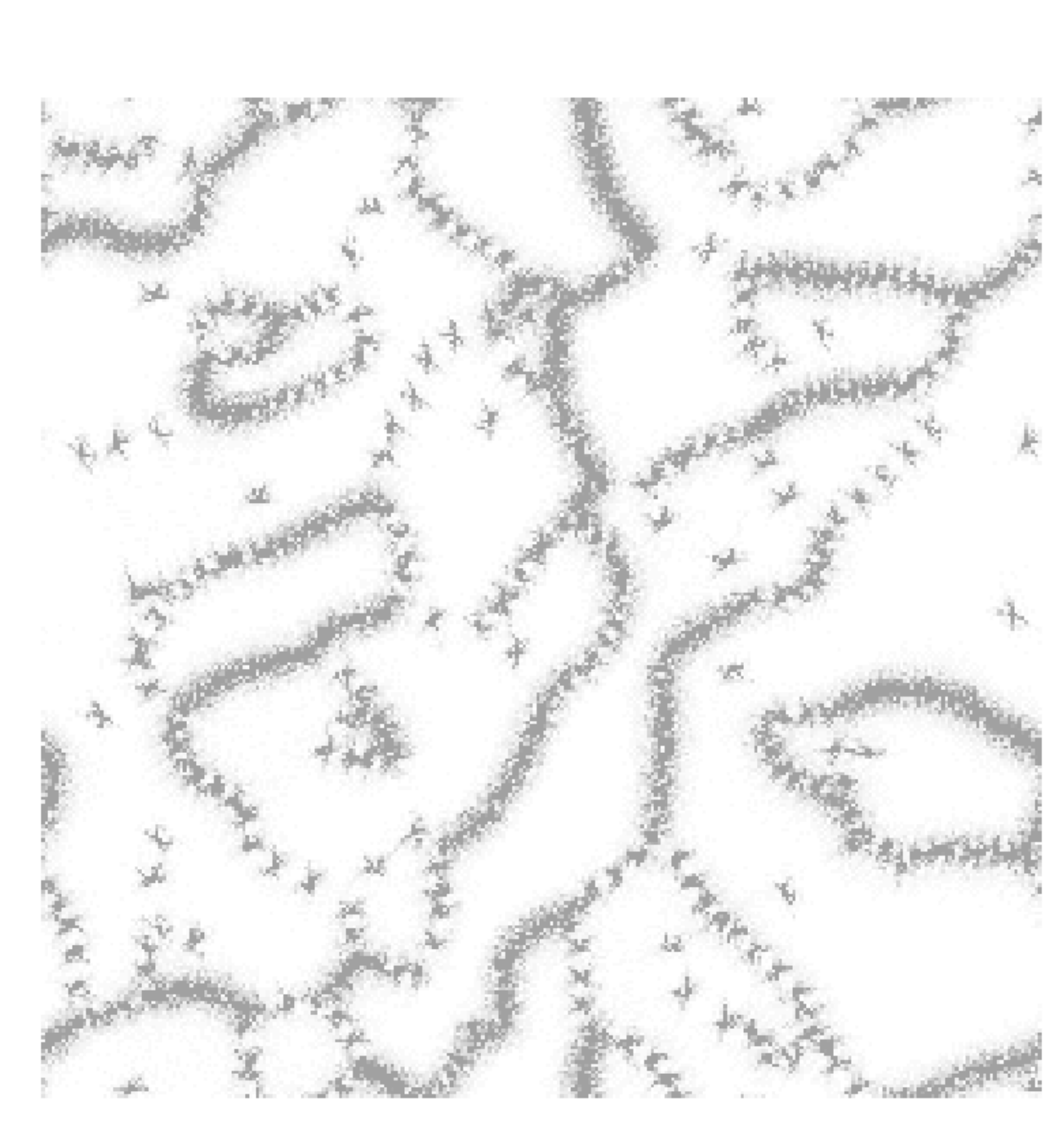}  \includegraphics[height=1.2in]{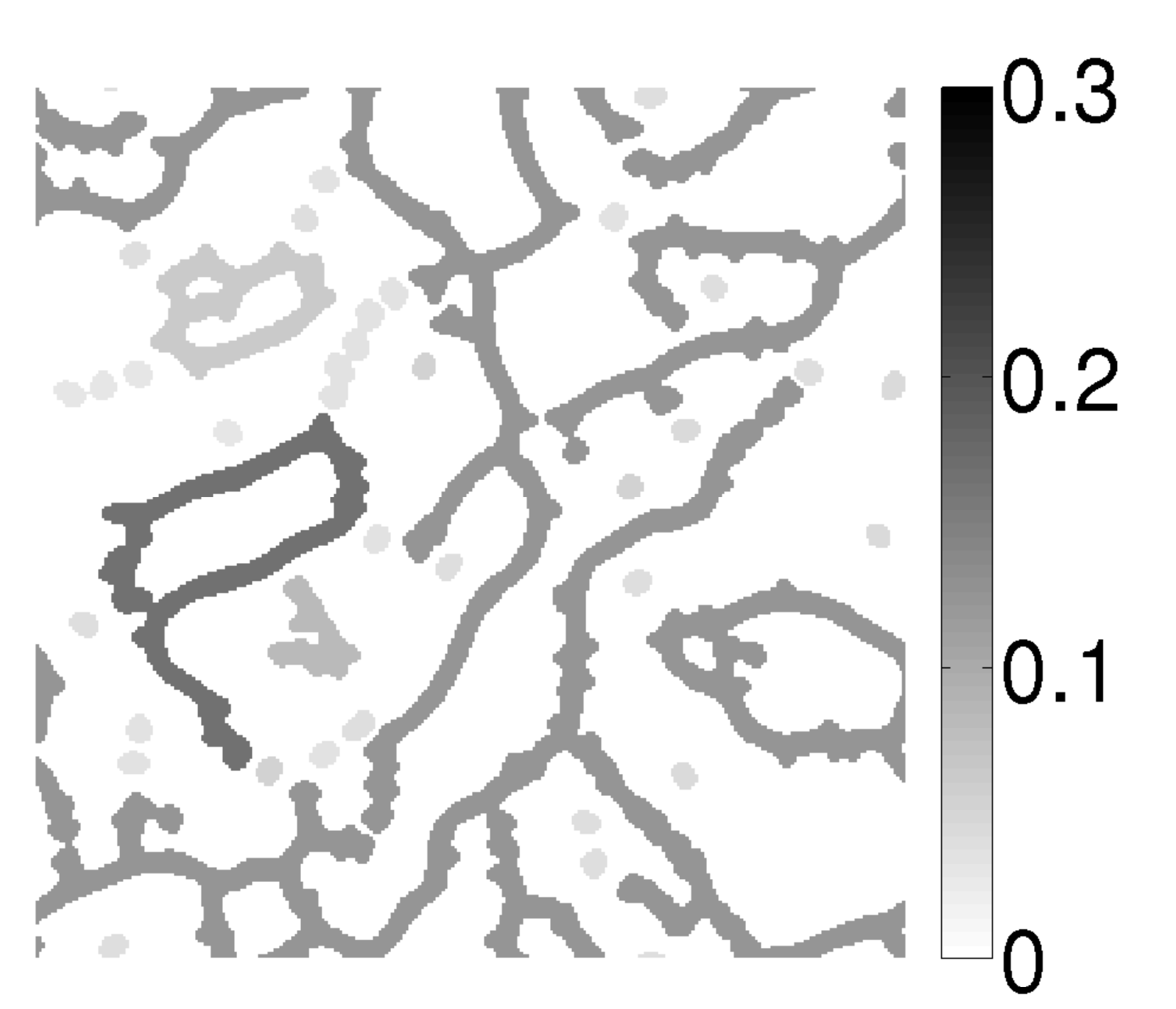} &
         \includegraphics[height=1.2in]{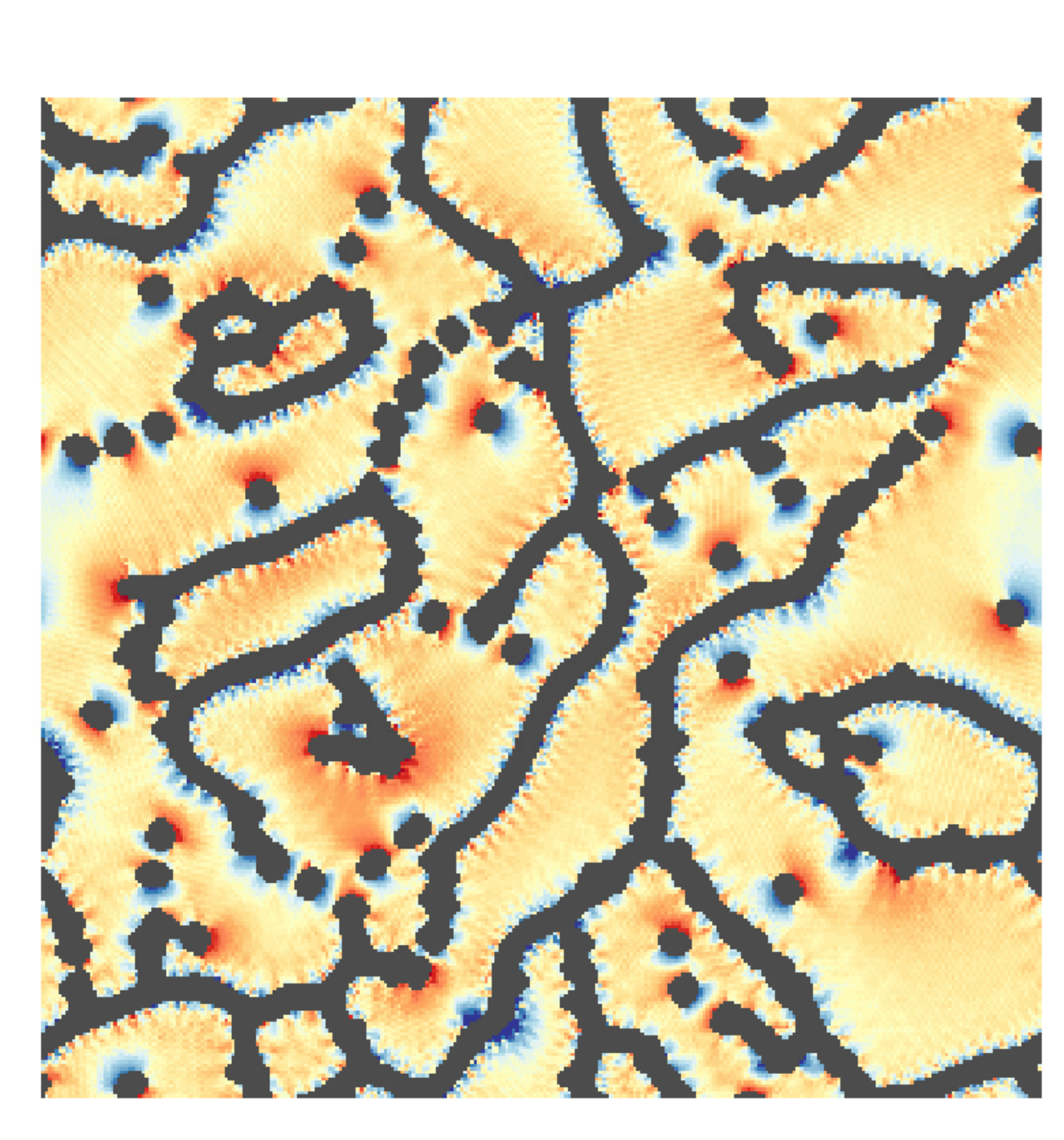}  \includegraphics[height=1.2in]{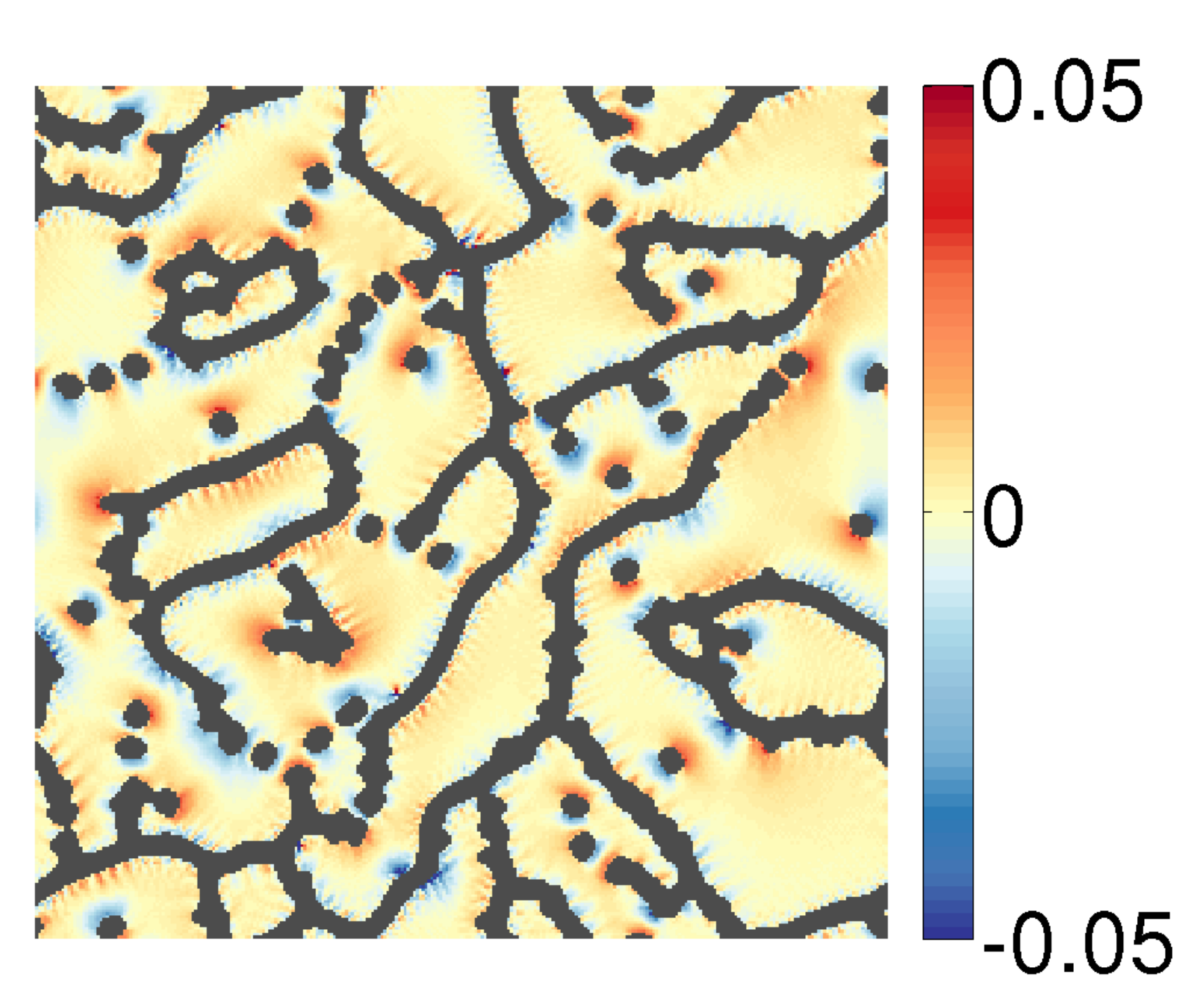} \\
         (c) & (d)
\end{tabular}
\end{center}
\caption{Results of the image in Figure\,\ref{fig:PFC1img} (left). From panel (a) to (d): crystal orientation, difference in principal stretches, curl of the inverse deformation gradient, and volume distortion. In each panel, the left figure shows the initial results from SST and the right one shows the optimized results from the variational method. Particularly in (c)-right, the average curl on each connected defect region is shown.}
\label{fig:PFC1}
\end{figure}

\begin{figure}
 \begin{center}
   \begin{tabular}{cc}
        \includegraphics[height=1.2in]{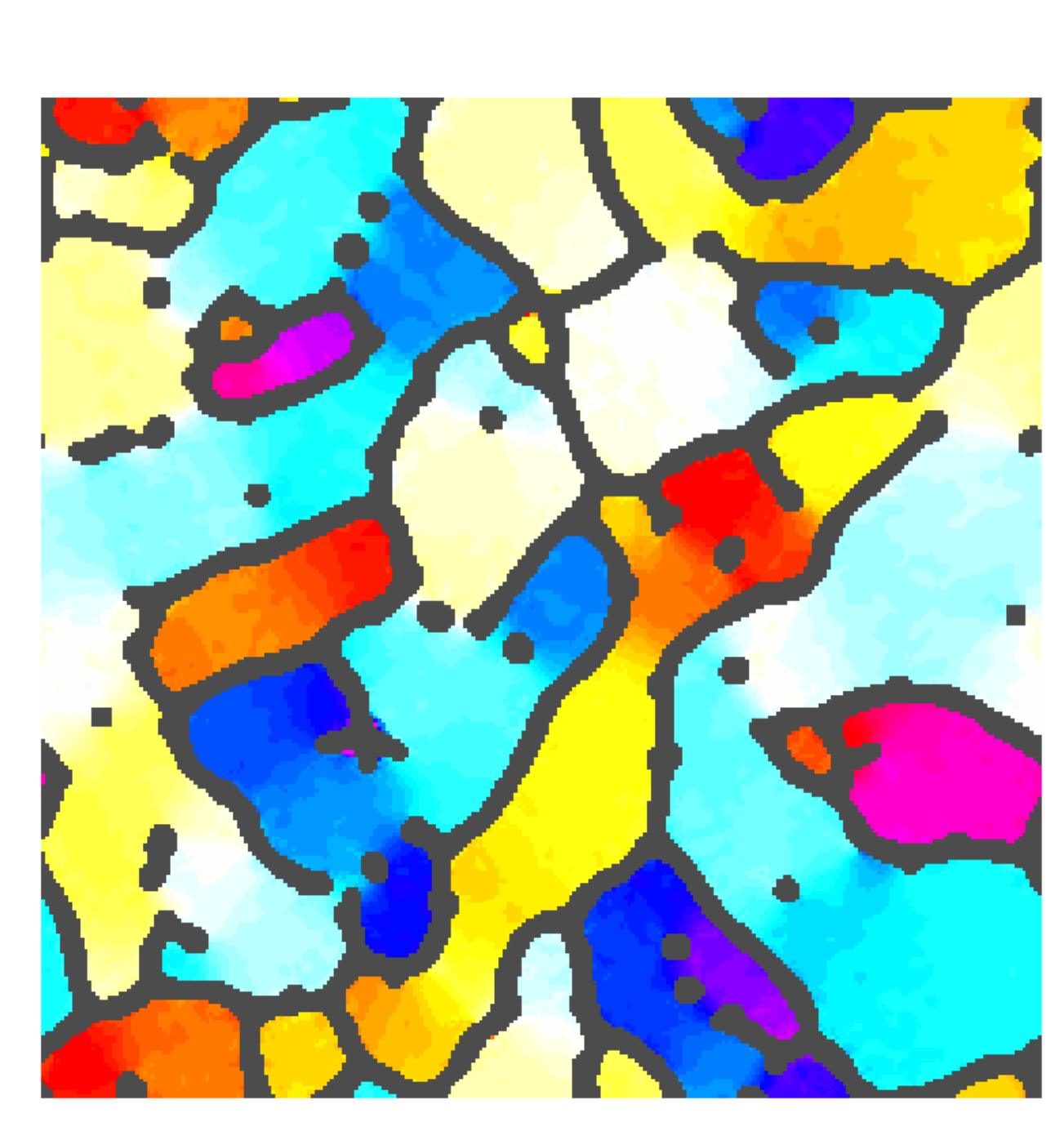} \includegraphics[height=1.2in]{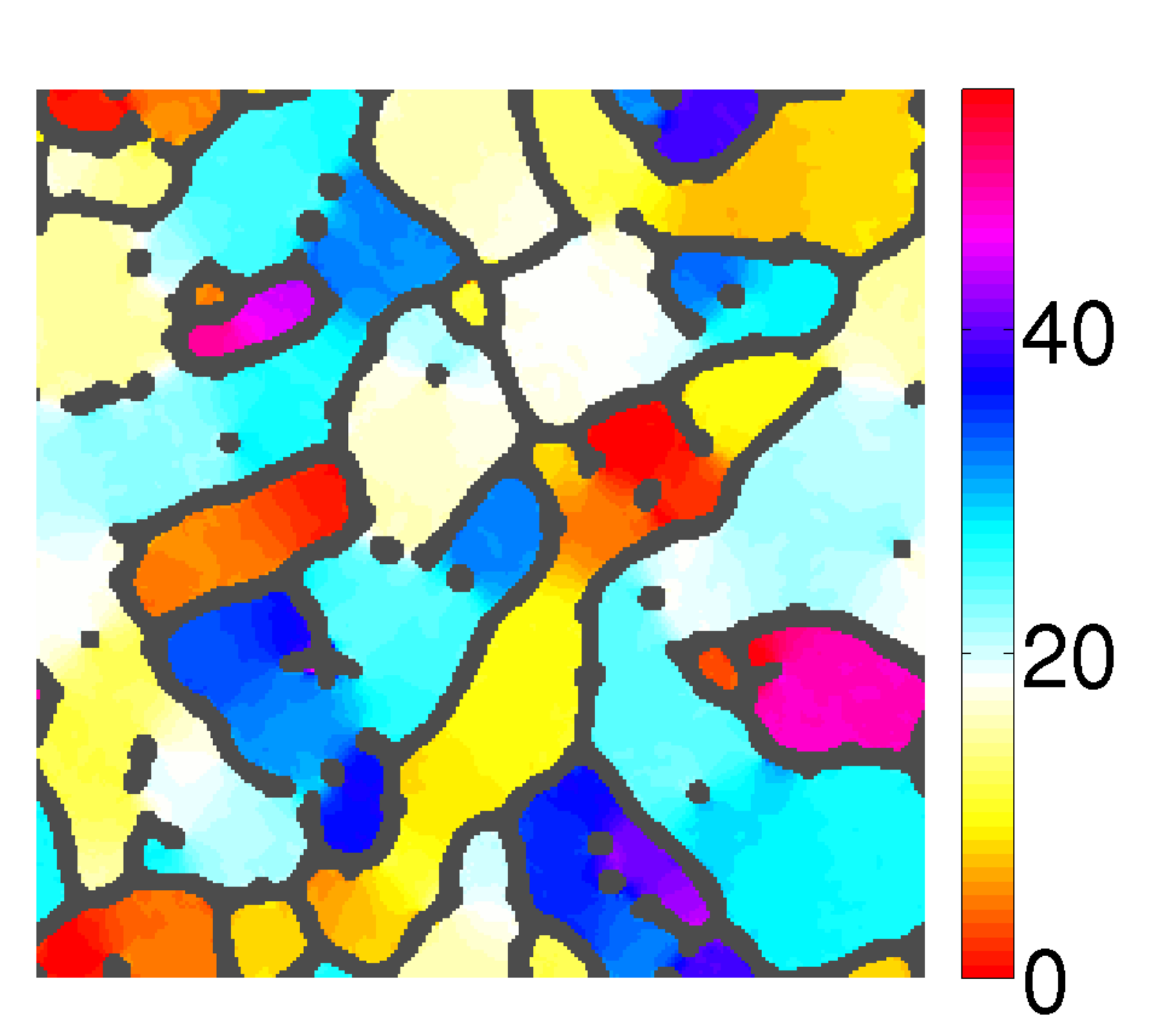} &
         \includegraphics[height=1.2in]{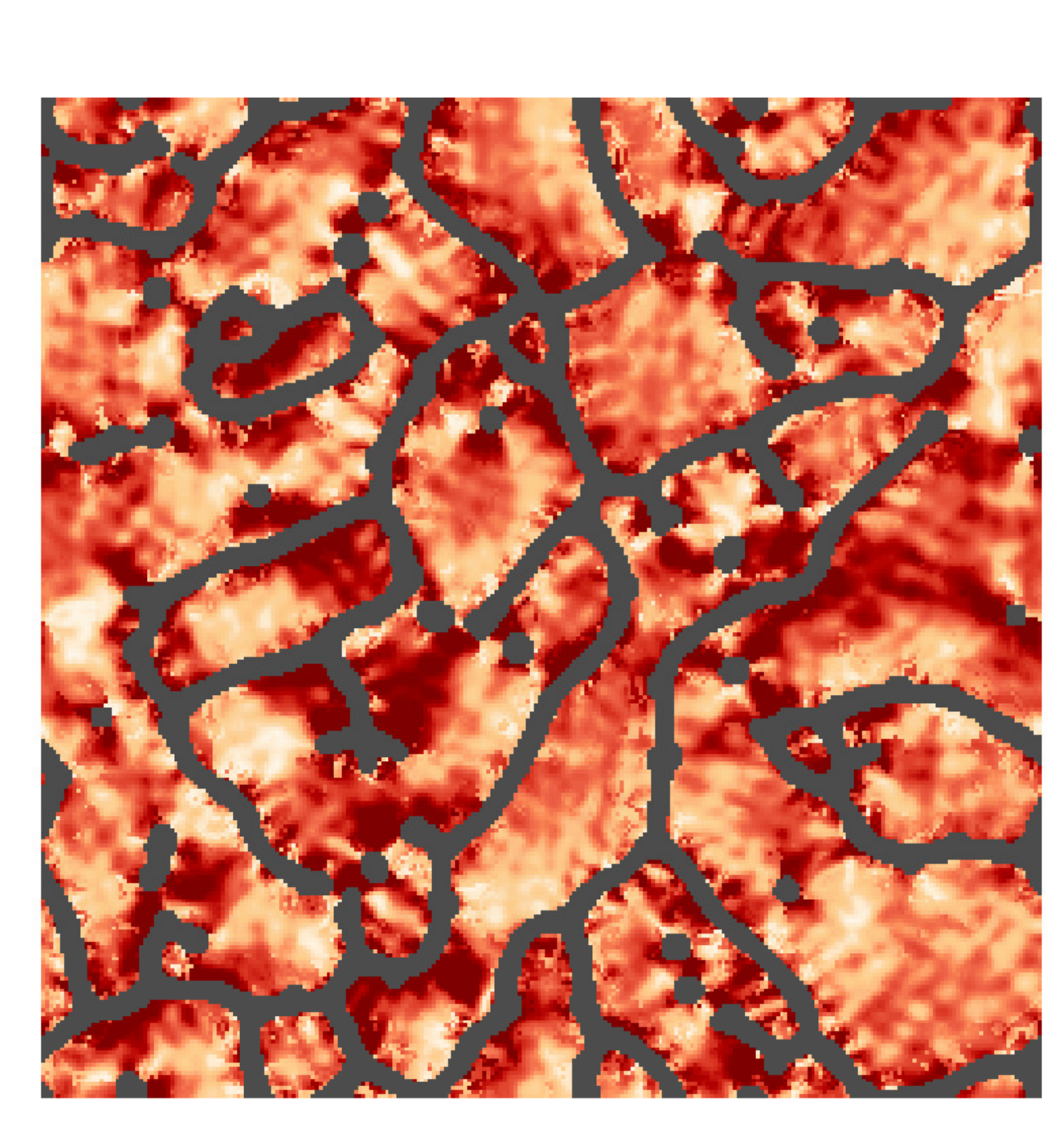} \includegraphics[height=1.2in]{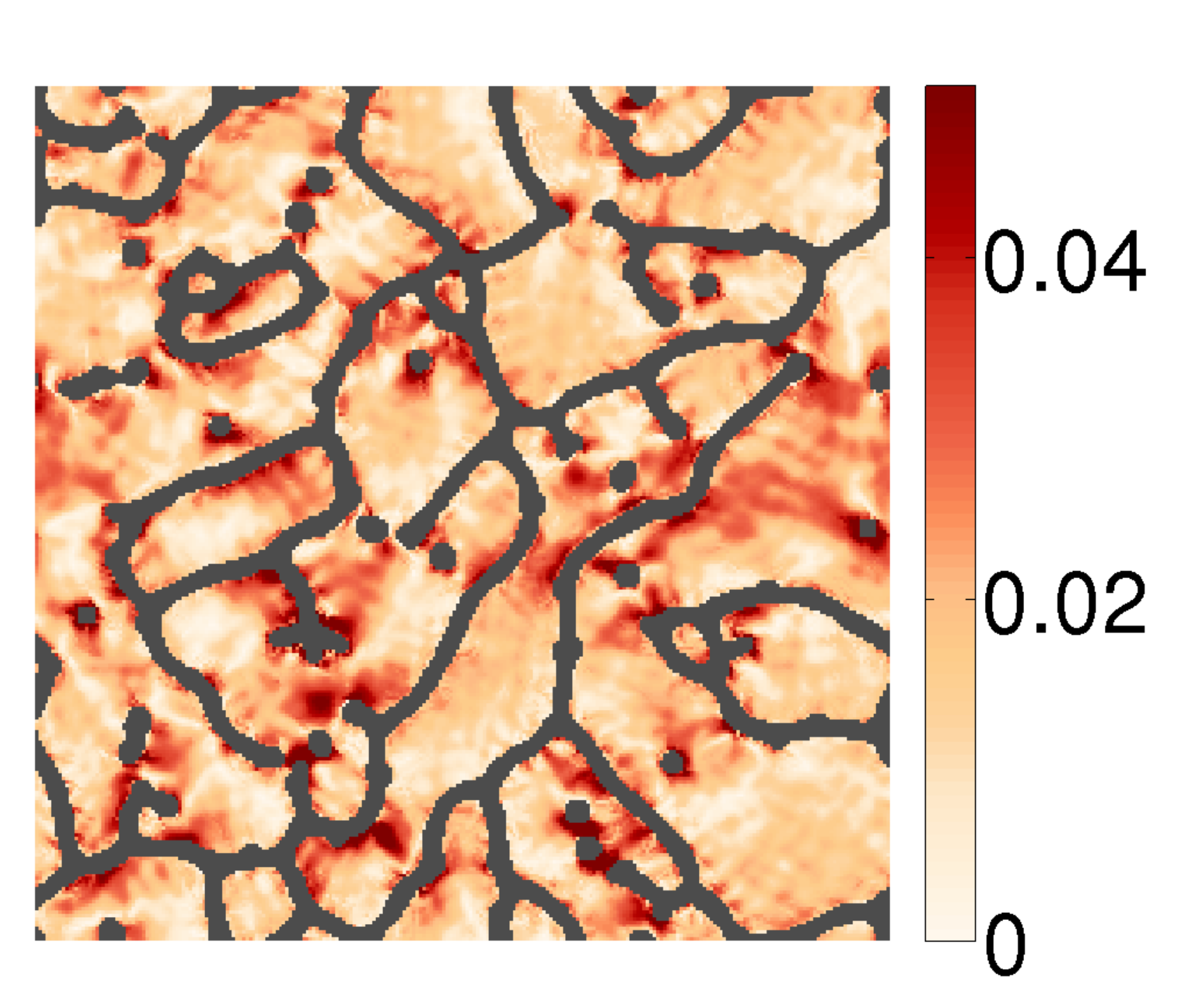} \\
         \includegraphics[height=1.2in]{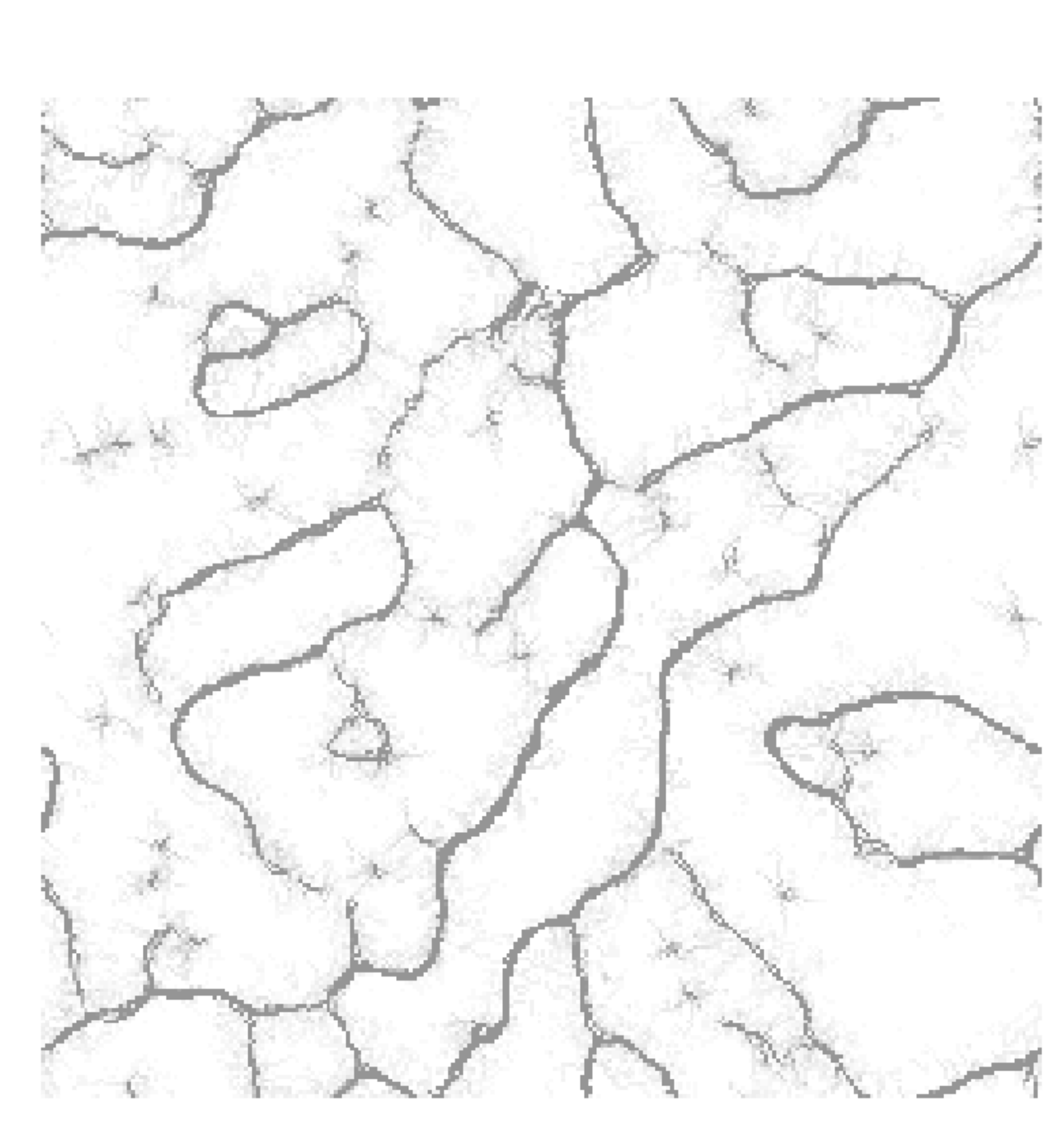} \includegraphics[height=1.2in]{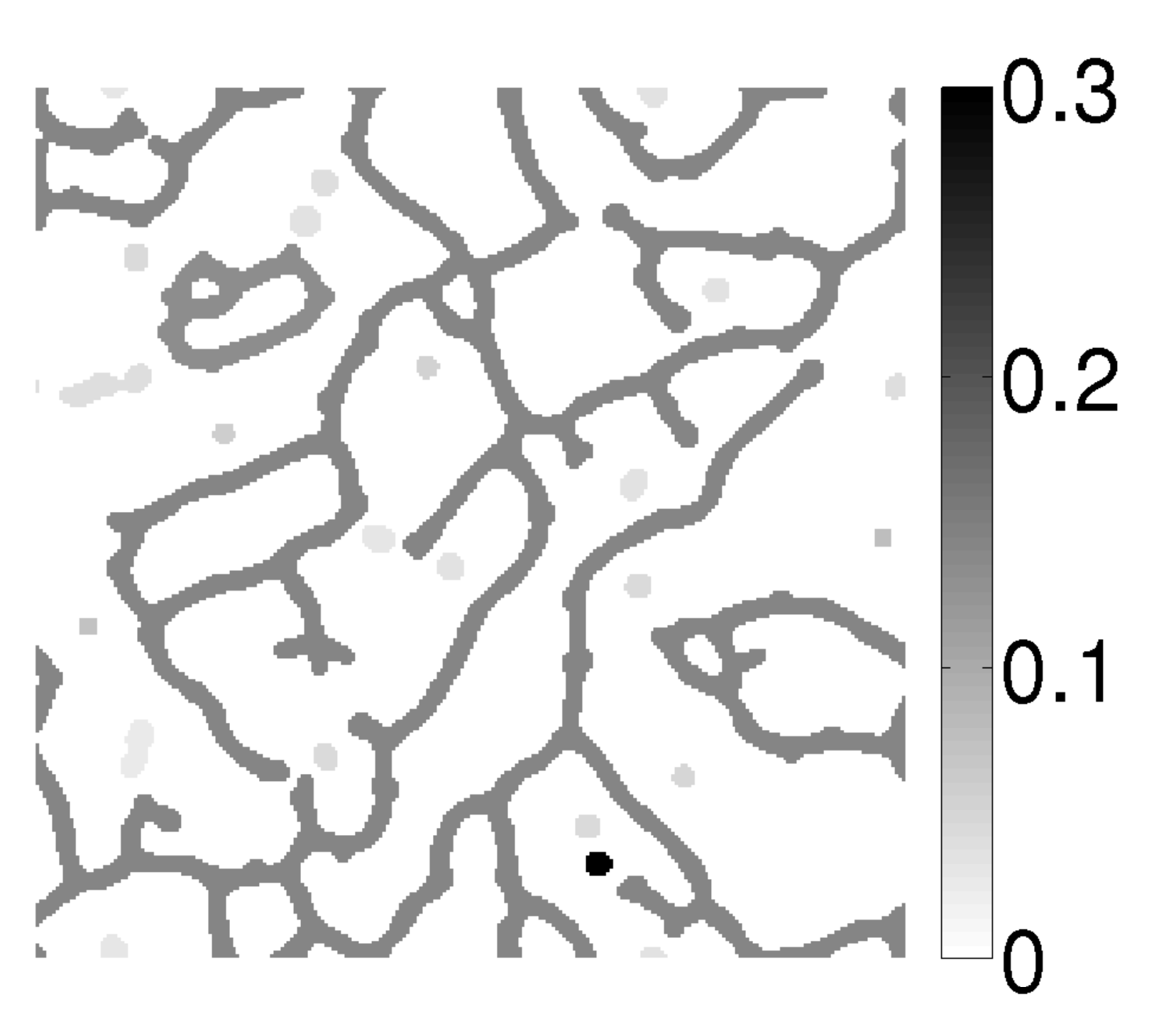} &
         \includegraphics[height=1.2in]{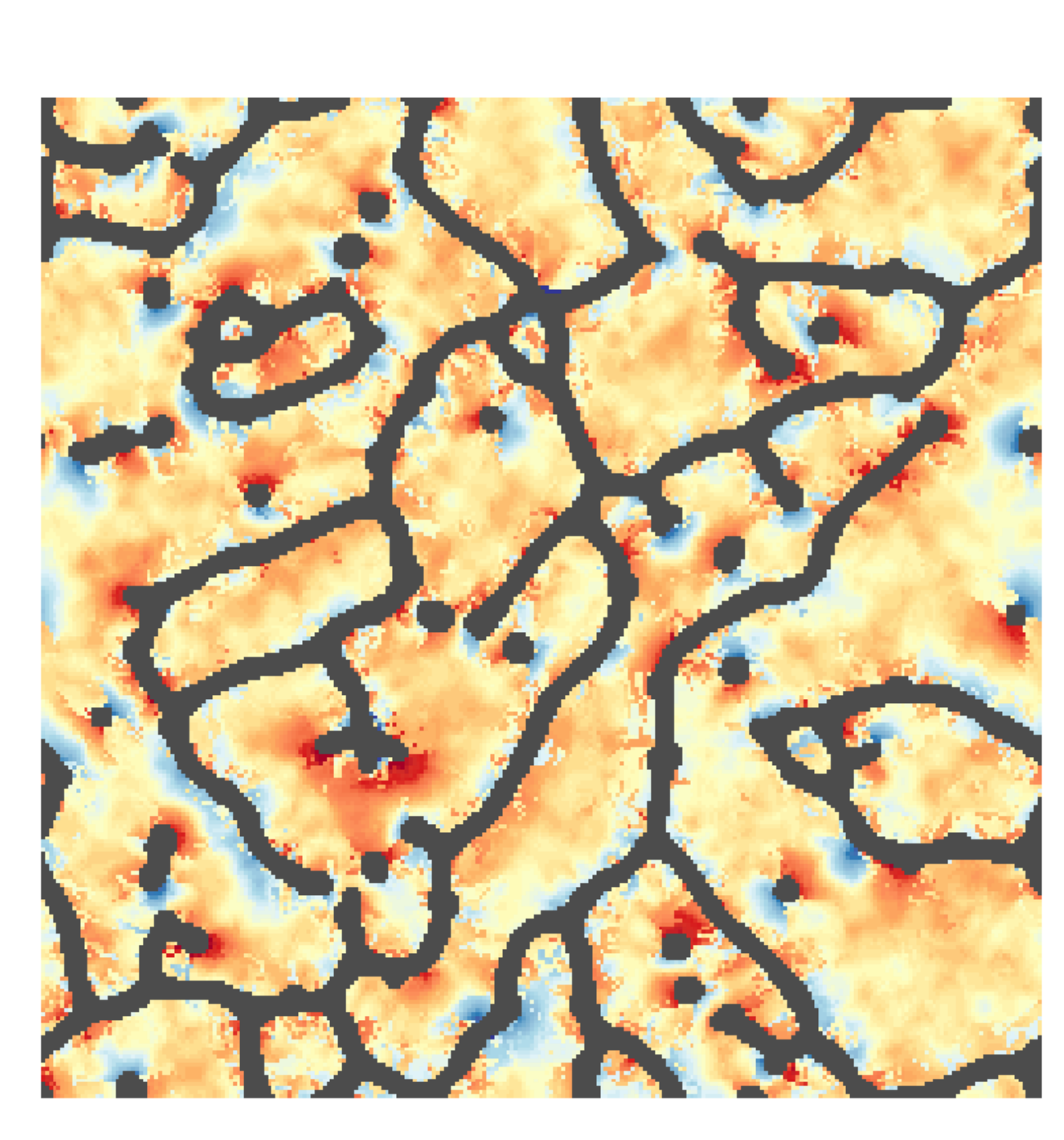}  \includegraphics[height=1.2in]{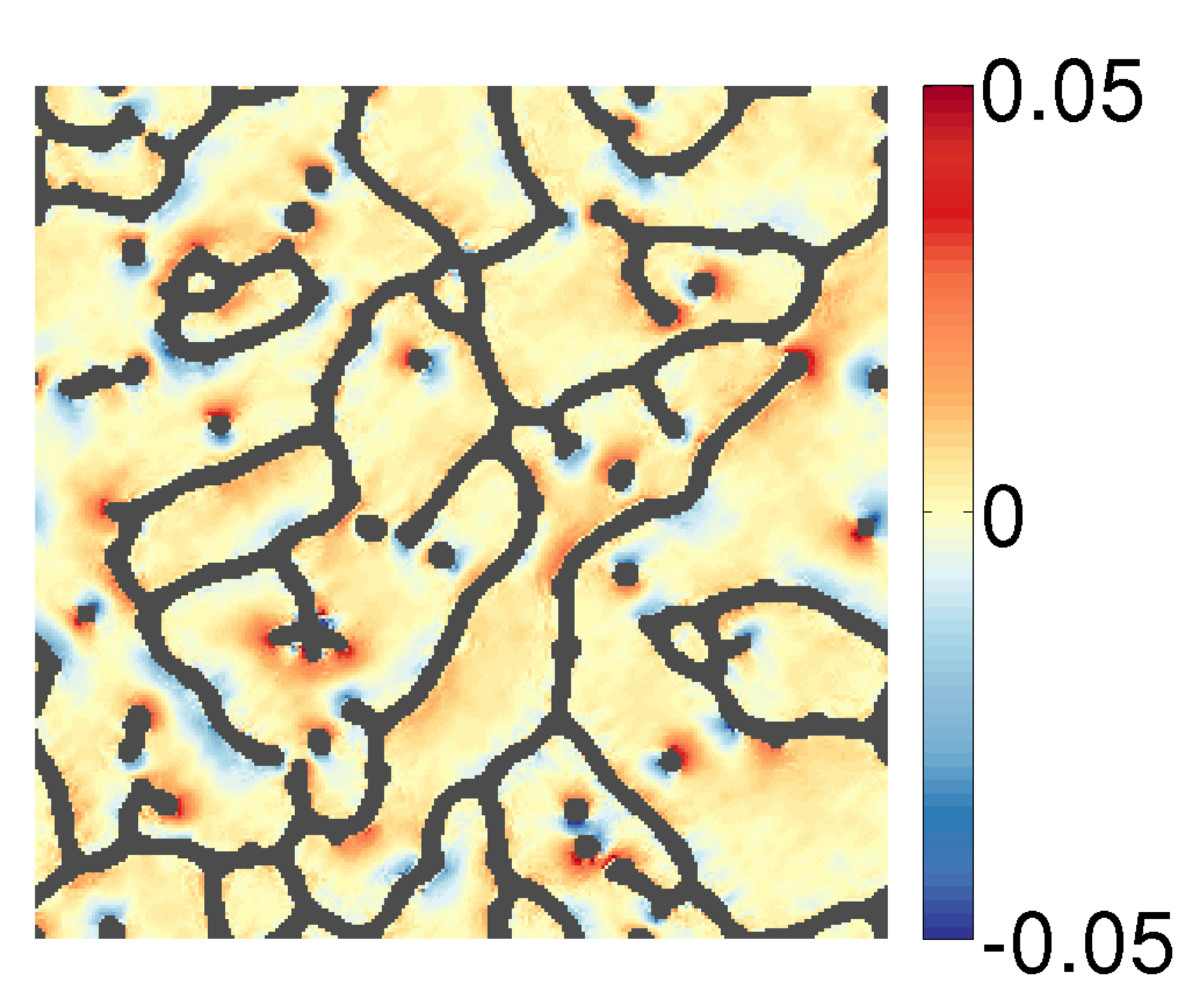} \\
           \end{tabular}
 \end{center}
\caption{Results of Figure\,\ref{fig:PFC1img} (right) using the same visualization as in Figure\,\ref{fig:PFC1}.}
\label{fig:PFC2}
\end{figure}

\begin{figure}
  \begin{tabular}{ccc}
  \includegraphics[height=1.6in]{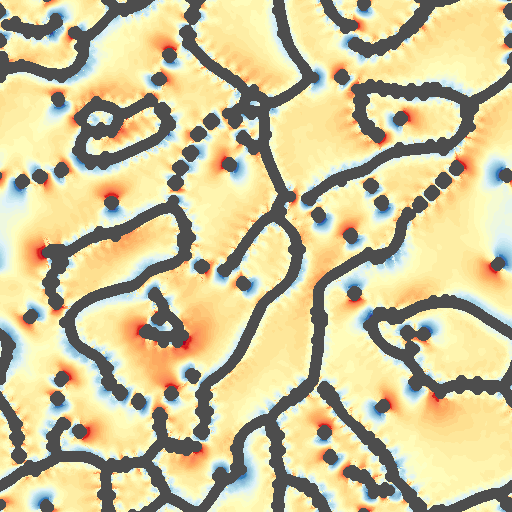}&
  \includegraphics[height=1.6in]{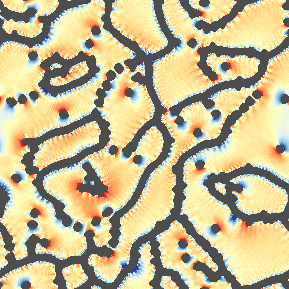}&
  \includegraphics[height=1.6in]{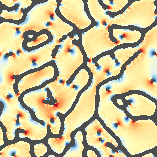}\\
  \includegraphics[height=1.6in]{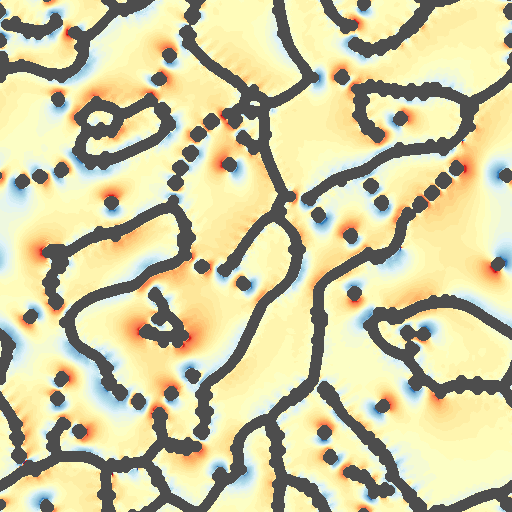}&
  \includegraphics[height=1.6in]{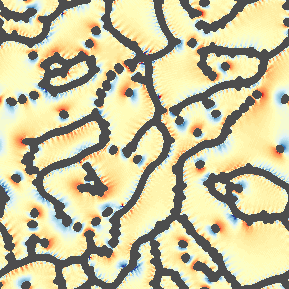}&
  \includegraphics[height=1.6in]{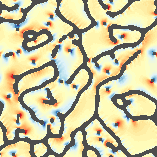}
\end{tabular}
\caption{A comparison of the results obtained with different downsampling rates. From left to right, the synchrosqueezed transform downsamples the original image by a factor of $2$, $4$, and $8$, and the runtime is about $50$ minutes, $13$ minutes and $20$ seconds. Top panel: initial volume distortion provided by the synchrosqueezed transform. Bottom panel: optimized volume distortion. }
\label{fig:comp}
\end{figure}
The PFC model is a well-established method to simulate elastic and plastic deformations, free surfaces, and multiple crystal orientations in nonequilibrium processes. Its simulation can provide synthetic crystal images with various topological defects. A noiseless PFC image is given in Figure~\ref{fig:PFC1img} (left). It contains isolated defects and large as well as small angle grain boundaries, the latter showing up as a string of dislocations.

Compared to the initial strain estimate $G_0$, the optimized $G$ is smoother, exhibits a much smaller overall volume distortion and shear (as visualized by the difference in principal stretches), and sharpens the compression-dilation dipoles around each single dislocation, as shown in Figure\,\ref{fig:PFC1}.

The noisy PFC image of Figure~\ref{fig:PFC1img} (right) is generated by adding
50$\%$ Gaussian random noise. 
Obviously, this leads to strong artifacts in the estimated deformation $G_0$.
Remarkably, after the optimization we retrieve an estimate $G$ almost as good as in the noiseless case as shown in Figure\,\ref{fig:PFC2}, 
which demonstrates the robustness of our method.

To show that the performance of our method is not really
  depending on the mesh size chosen for the optimization algorithm, we
  compare the results obtained with different downsampling rates in
  the synchrosqueezed transform. Figure \ref{fig:comp} summarizes the
  comparison for the PFC example. Even though the grain boundary
  geometry is complicated in this example, our algorithm obtains
  similar results for different downsampling rates.

\subsection{Experimental crystal images}

The first real crystal image, a TEM-image of GaN  of size $420\times 444$ (Figure~\ref{fig:Real1img} left), contains a string of dislocations forming a large angle grain boundary.
The artificial strong spatial variation of crystal orientation, shear and volume distortion in the SST result is greatly reduced after applying the optimization as shown in Figure\,\ref{fig:GaN}.
Even more importantly, the unphysical curl away from the defects is completely removed
such that the curl of $G$ is fully concentrated in the single defect regions around each dislocation (the total curl of each region equalling the dislocation's Burgers vector).

\begin{figure}
\begin{center}
   \begin{tabular}{cc}
       \includegraphics[height=2.4in]{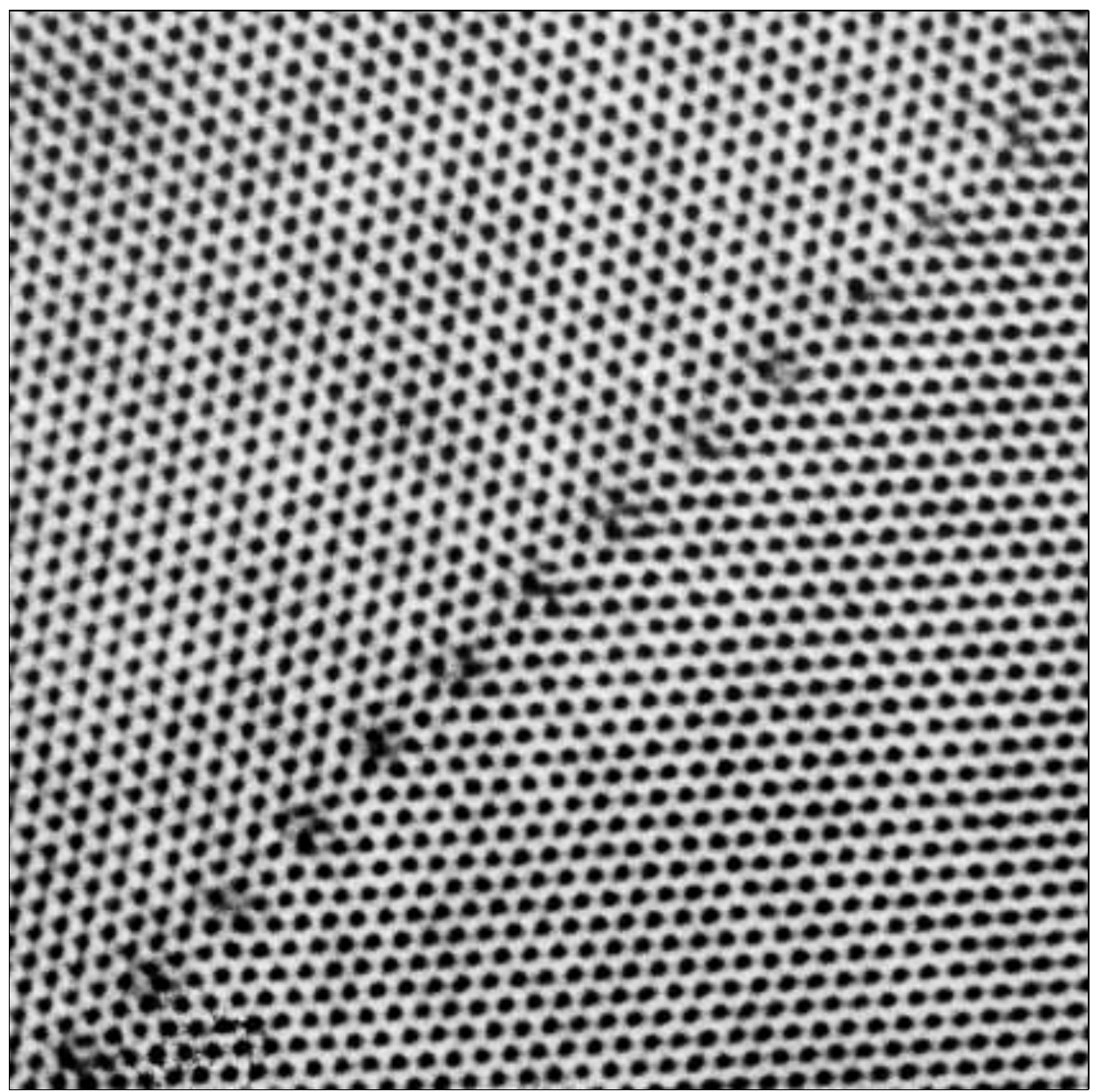}&\includegraphics[height=2.4in]{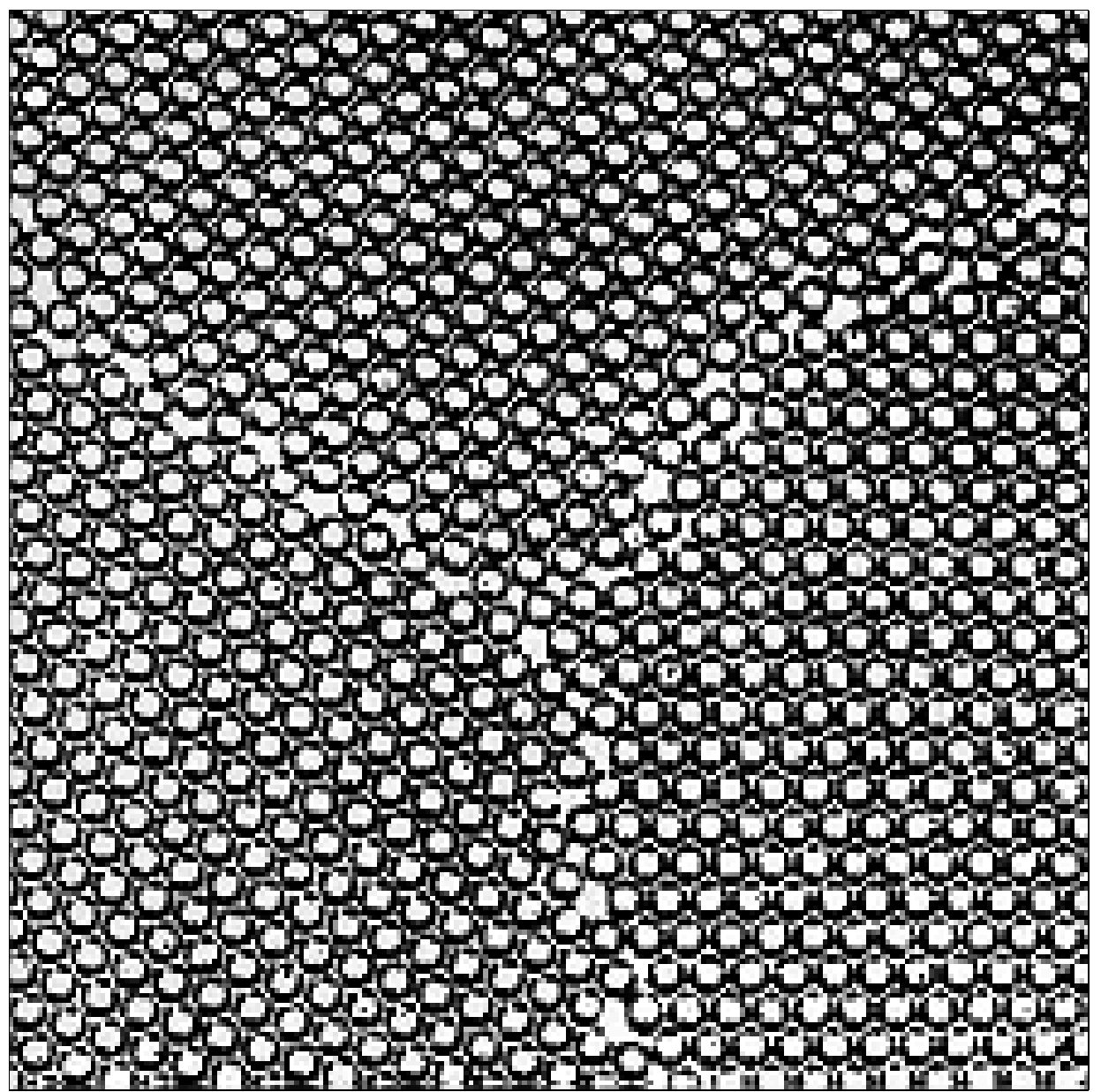}
       \end{tabular}
   \end{center}
 \caption{Left: a TEM-image of GaN (courtesy of David M.\ Tricker, Department of Material Science and Metallurgy at the University of Cambridge). Right: a photograph of a bubble raft (courtesy of Barrie S.\ H.\ Royce at Princeton University).}
 \label{fig:Real1img}
\end{figure}
       
 \begin{figure}
 \begin{center}
       \begin{tabular}{cc}
        \includegraphics[height=1.2in]{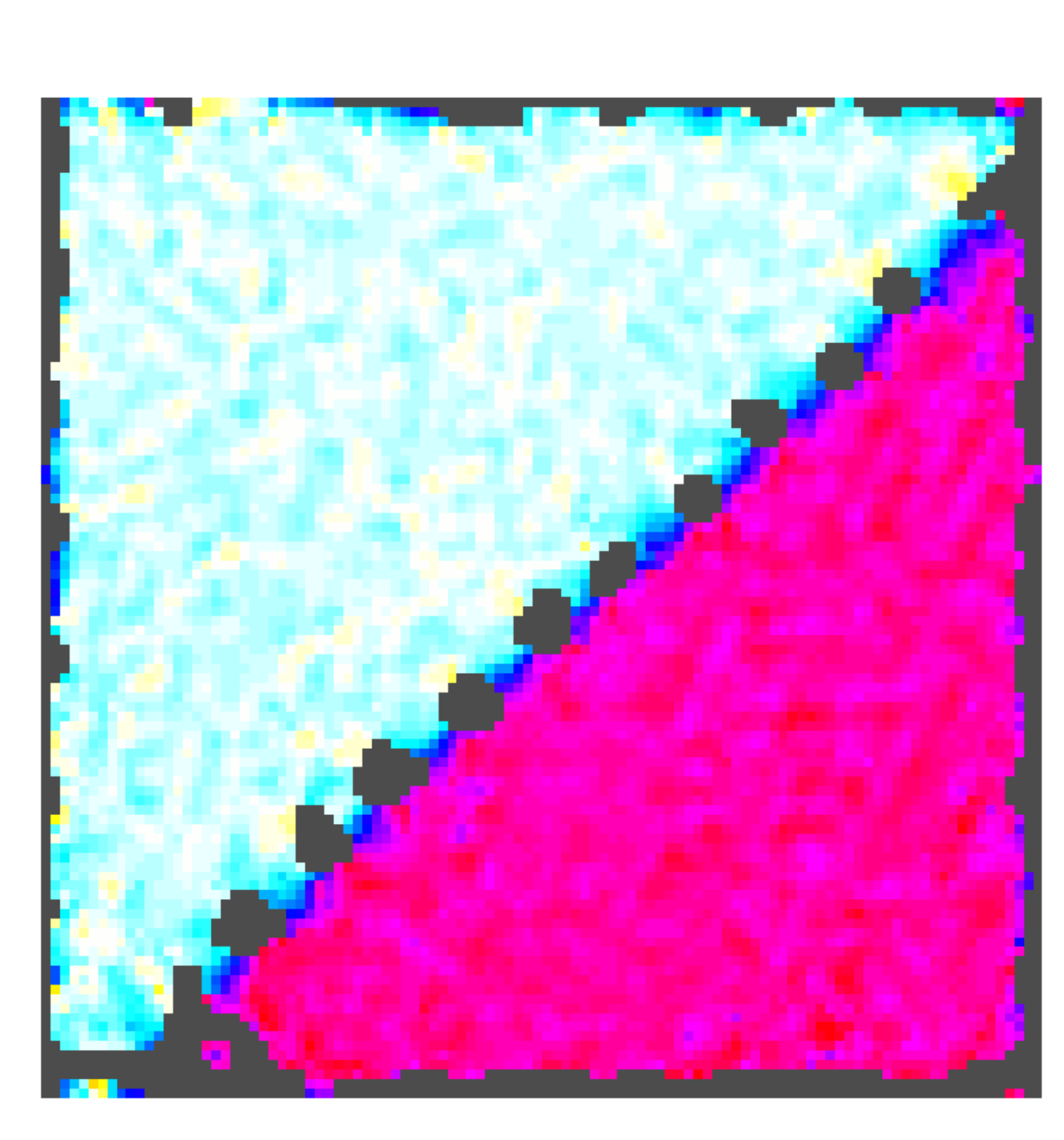} \includegraphics[height=1.2in]{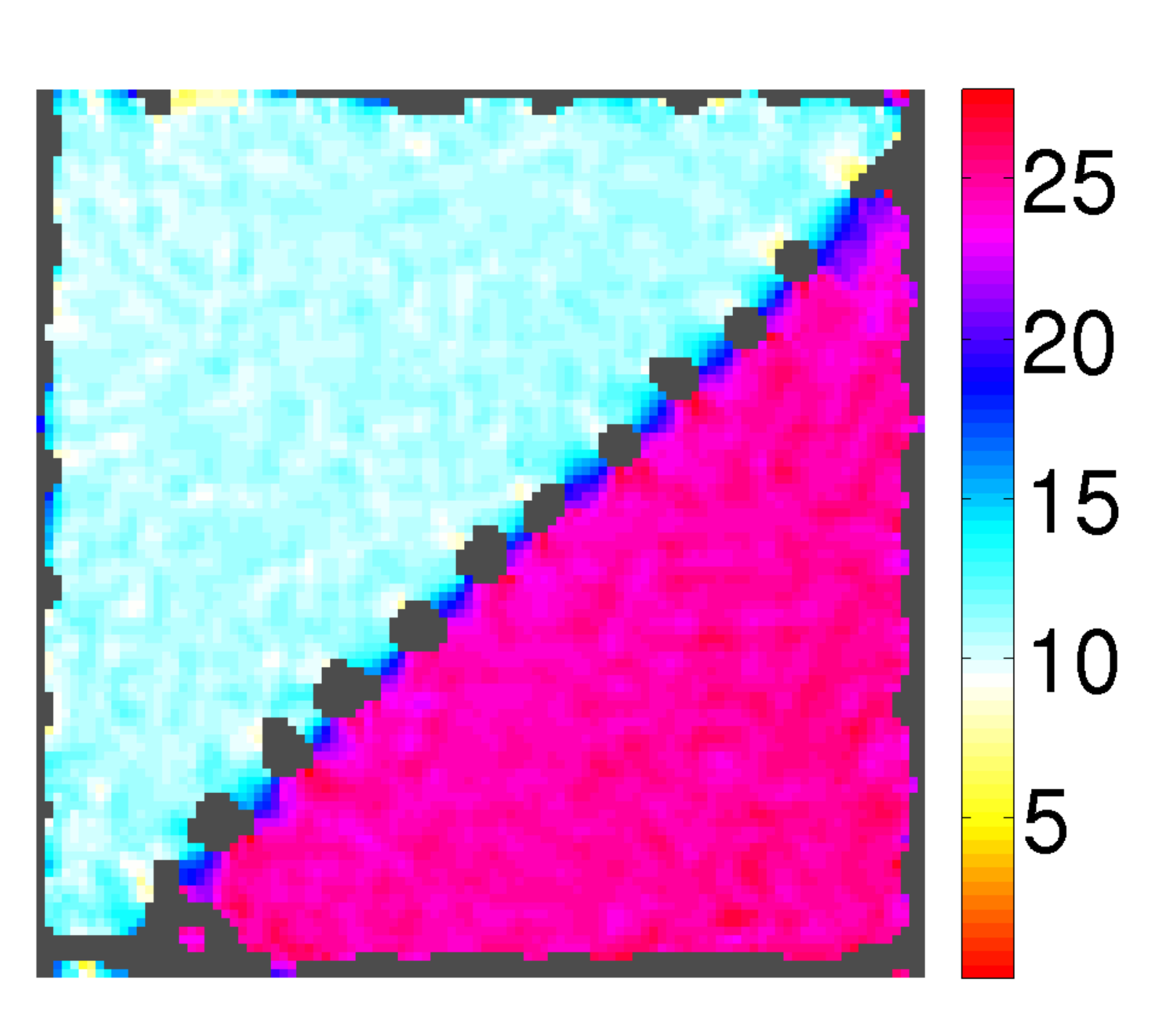} &
         \includegraphics[height=1.2in]{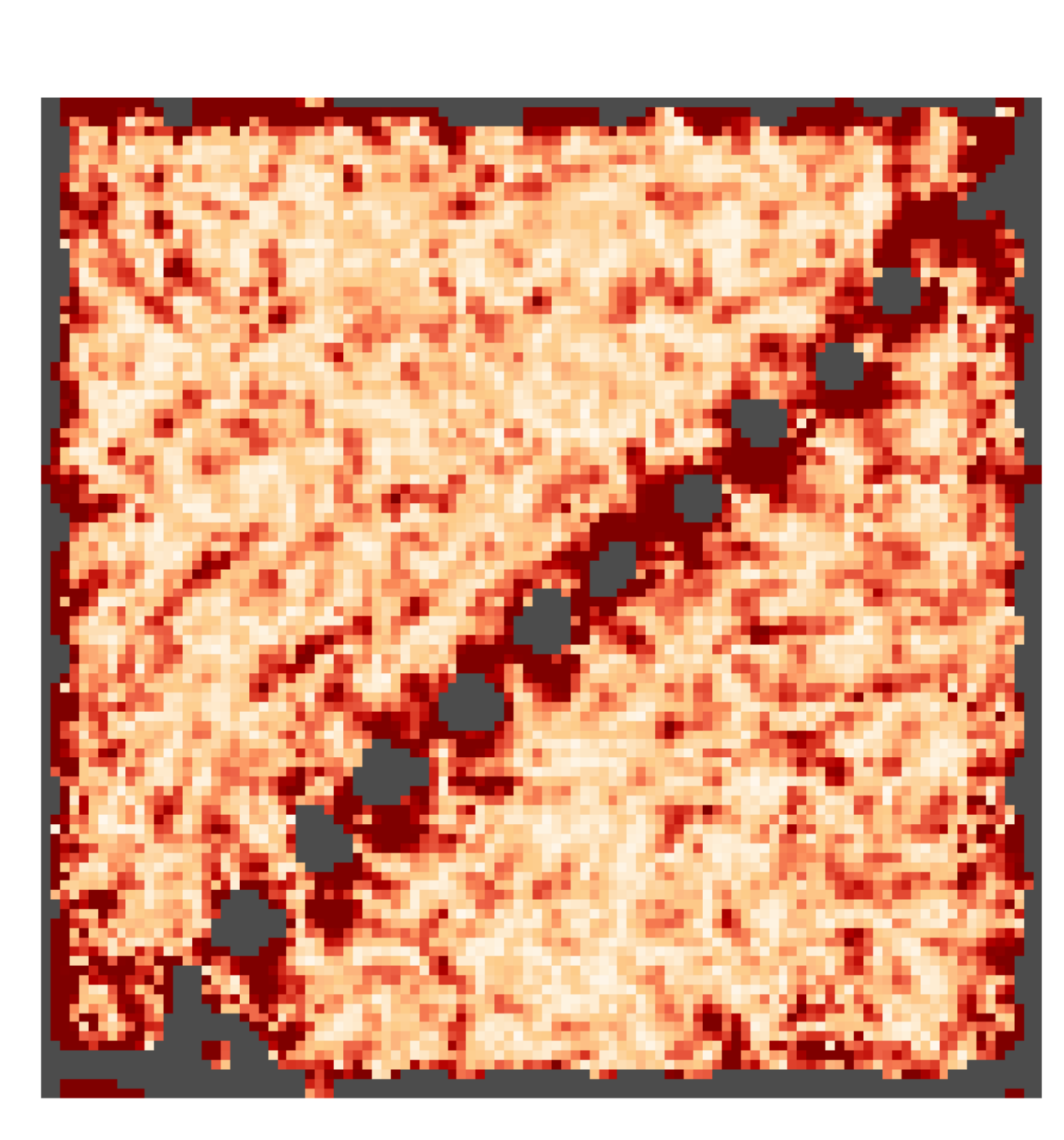} \includegraphics[height=1.2in]{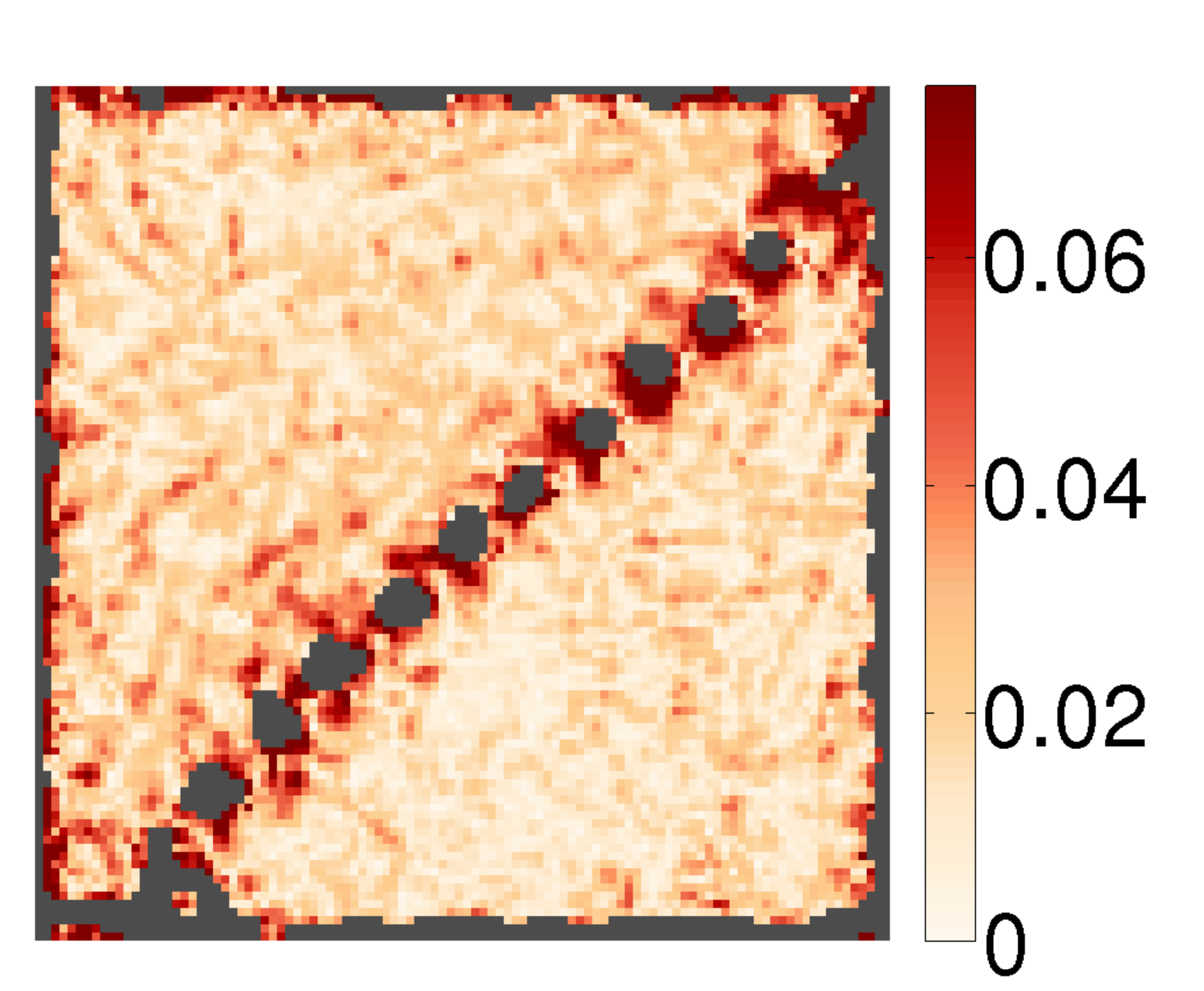} \\
         \includegraphics[height=1.2in]{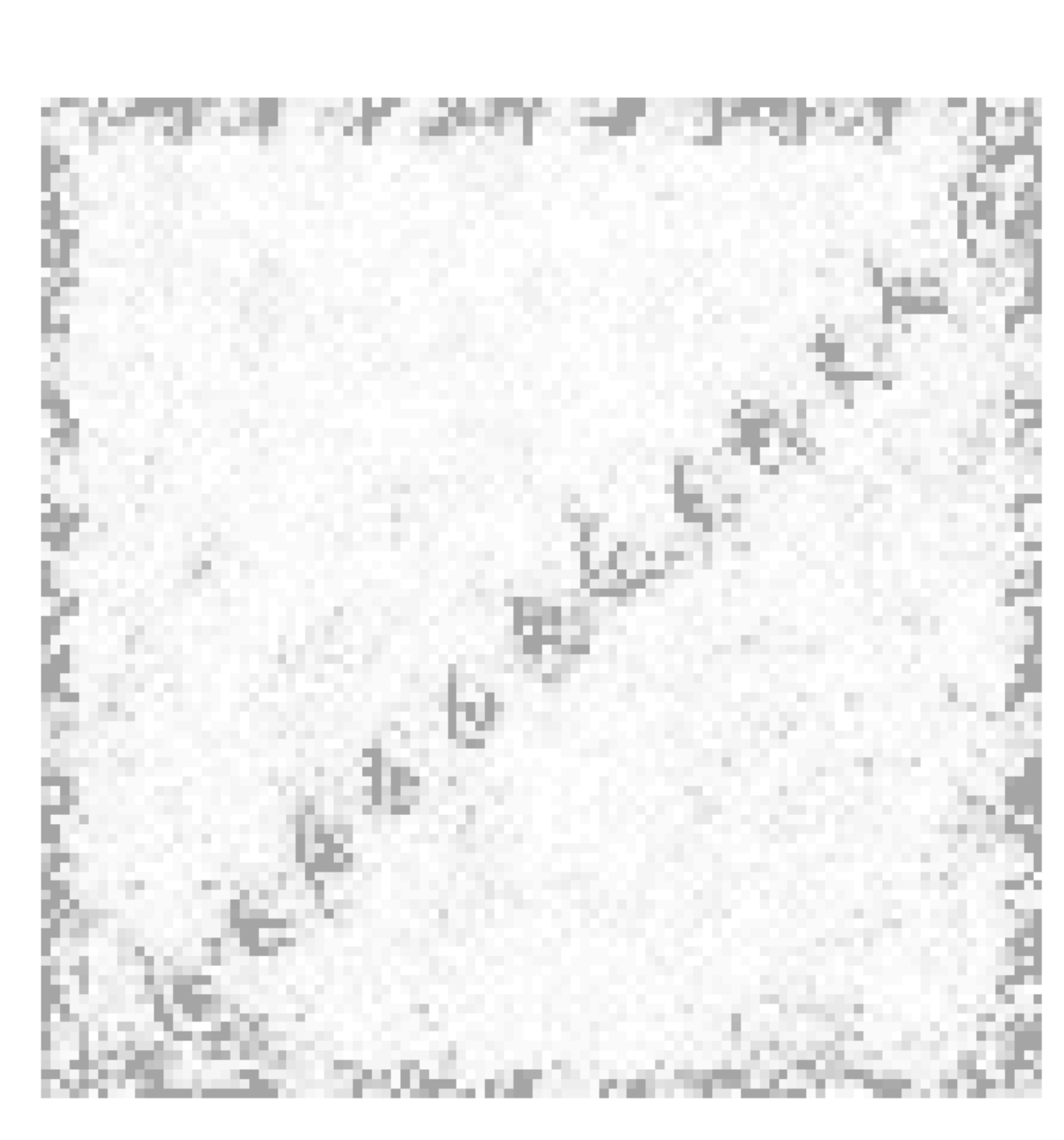}  \includegraphics[height=1.2in]{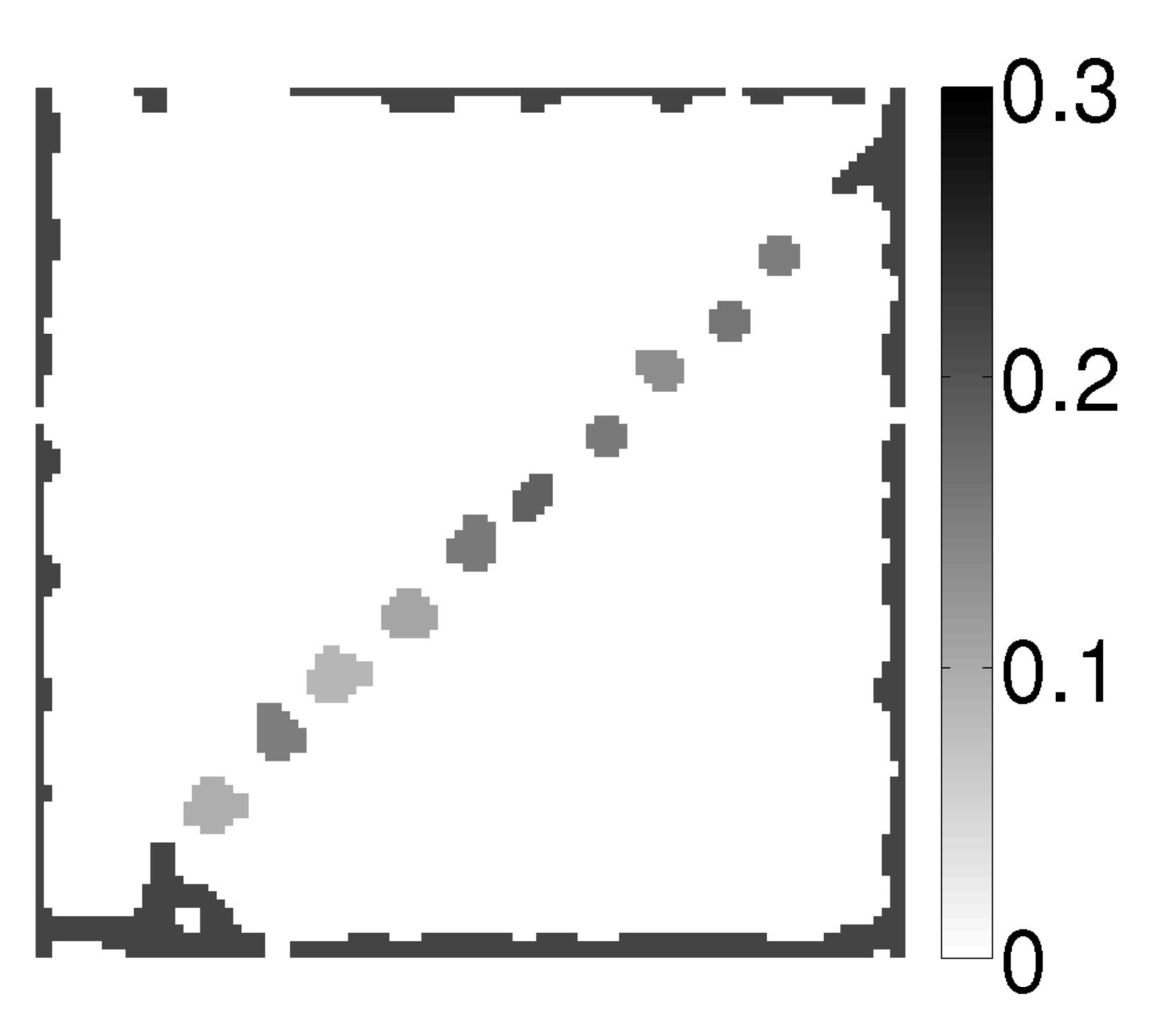} &
         \includegraphics[height=1.2in]{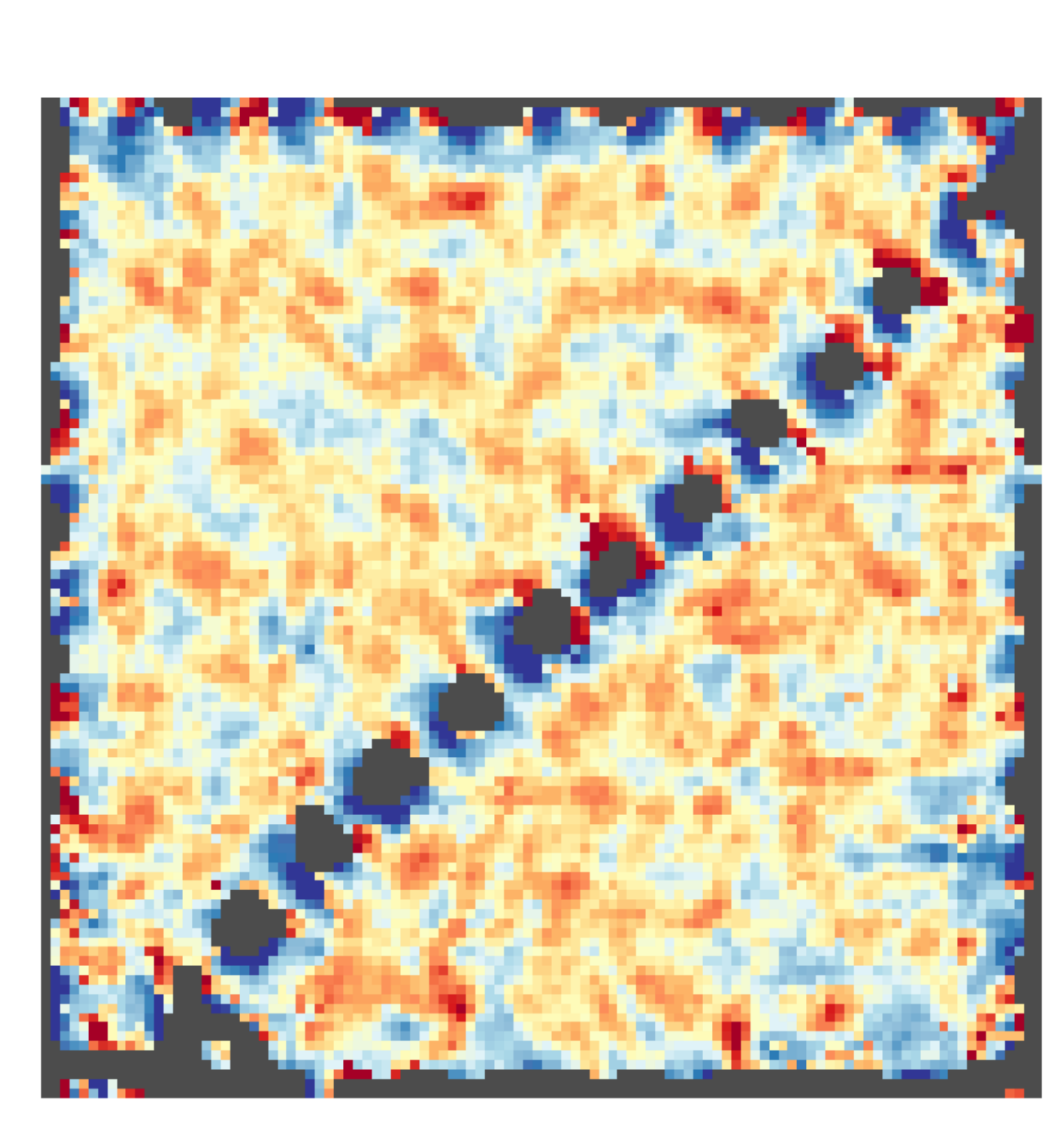}  \includegraphics[height=1.2in]{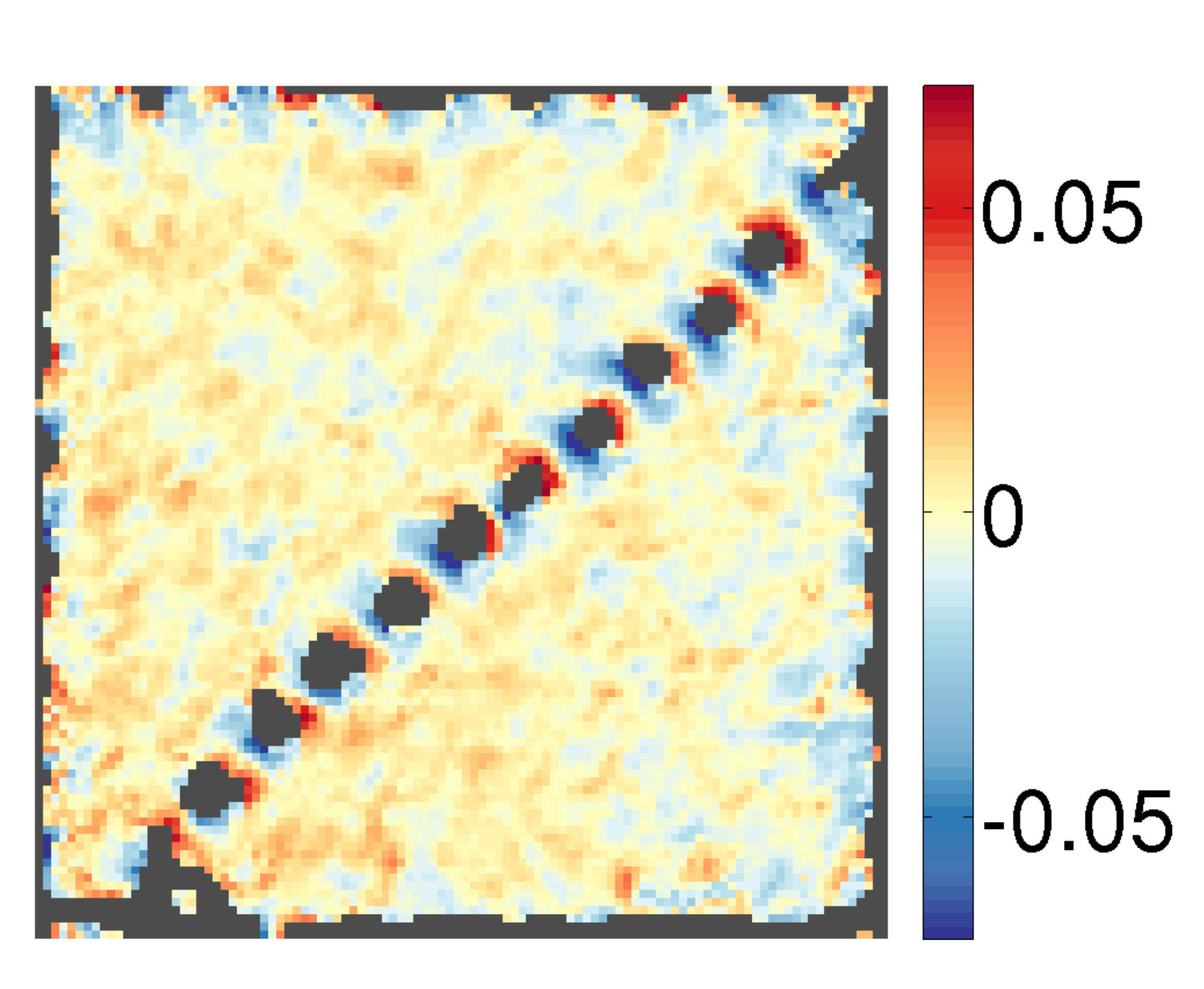} 
   \end{tabular}
\end{center}
\caption{Results of crystal analysis for Figure\,\ref{fig:Real1img} (left), using the same visualization as in Figure\,\ref{fig:PFC1}.}
\label{fig:GaN}
\end{figure}

 \begin{figure}
 \begin{center}
       \begin{tabular}{cc}
        \includegraphics[height=1.2in]{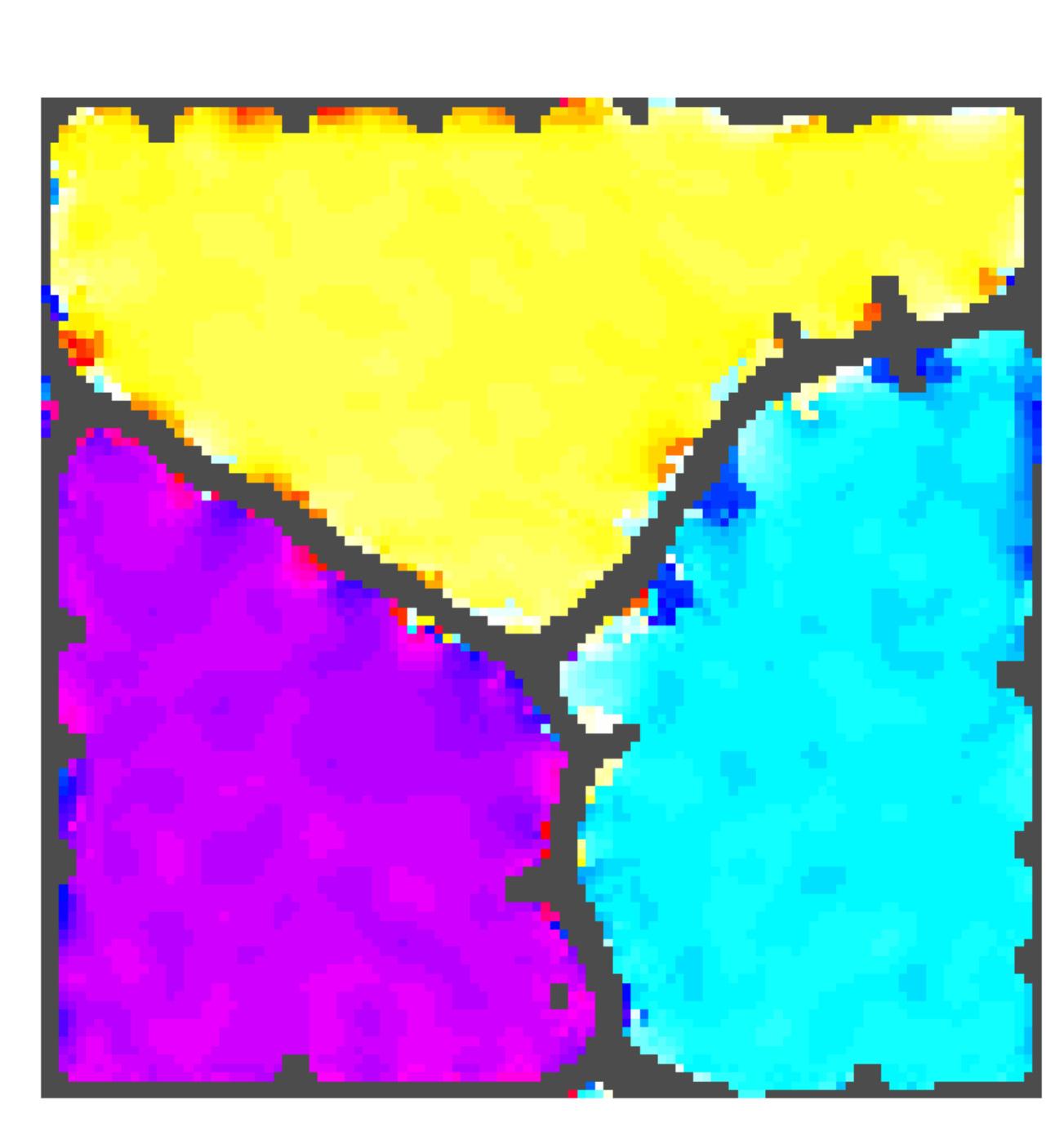} \includegraphics[height=1.2in]{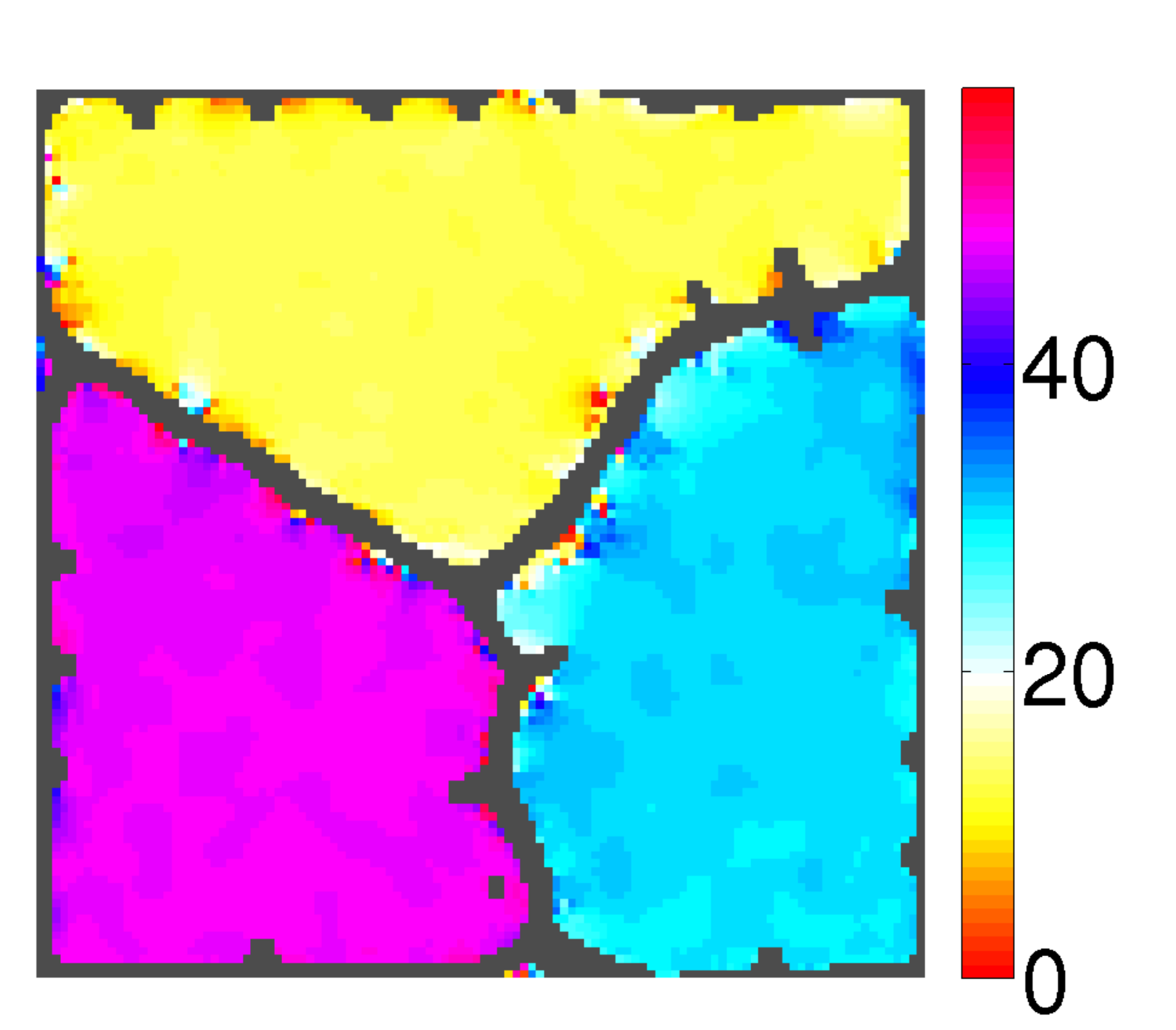} &
         \includegraphics[height=1.2in]{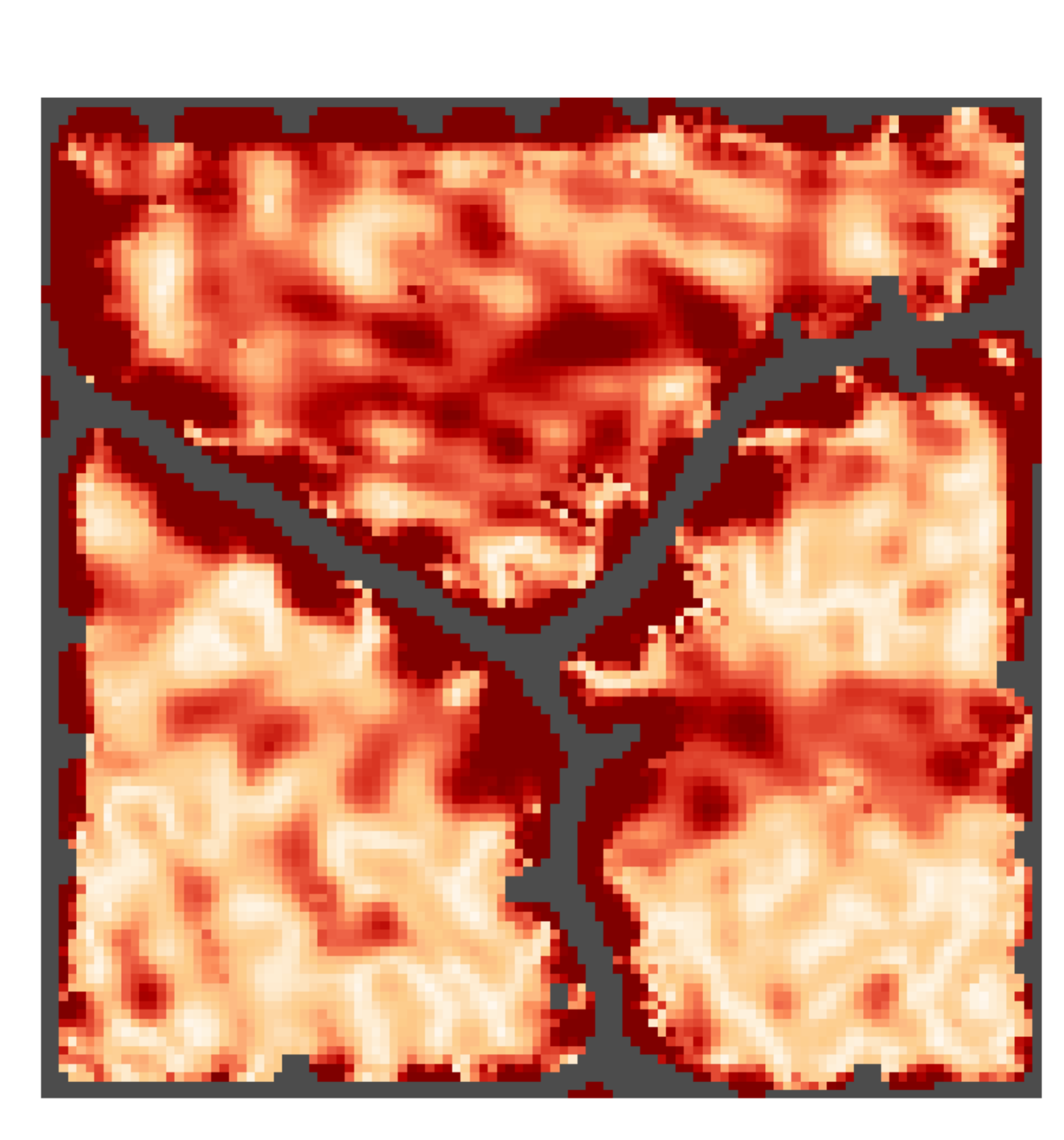} \includegraphics[height=1.2in]{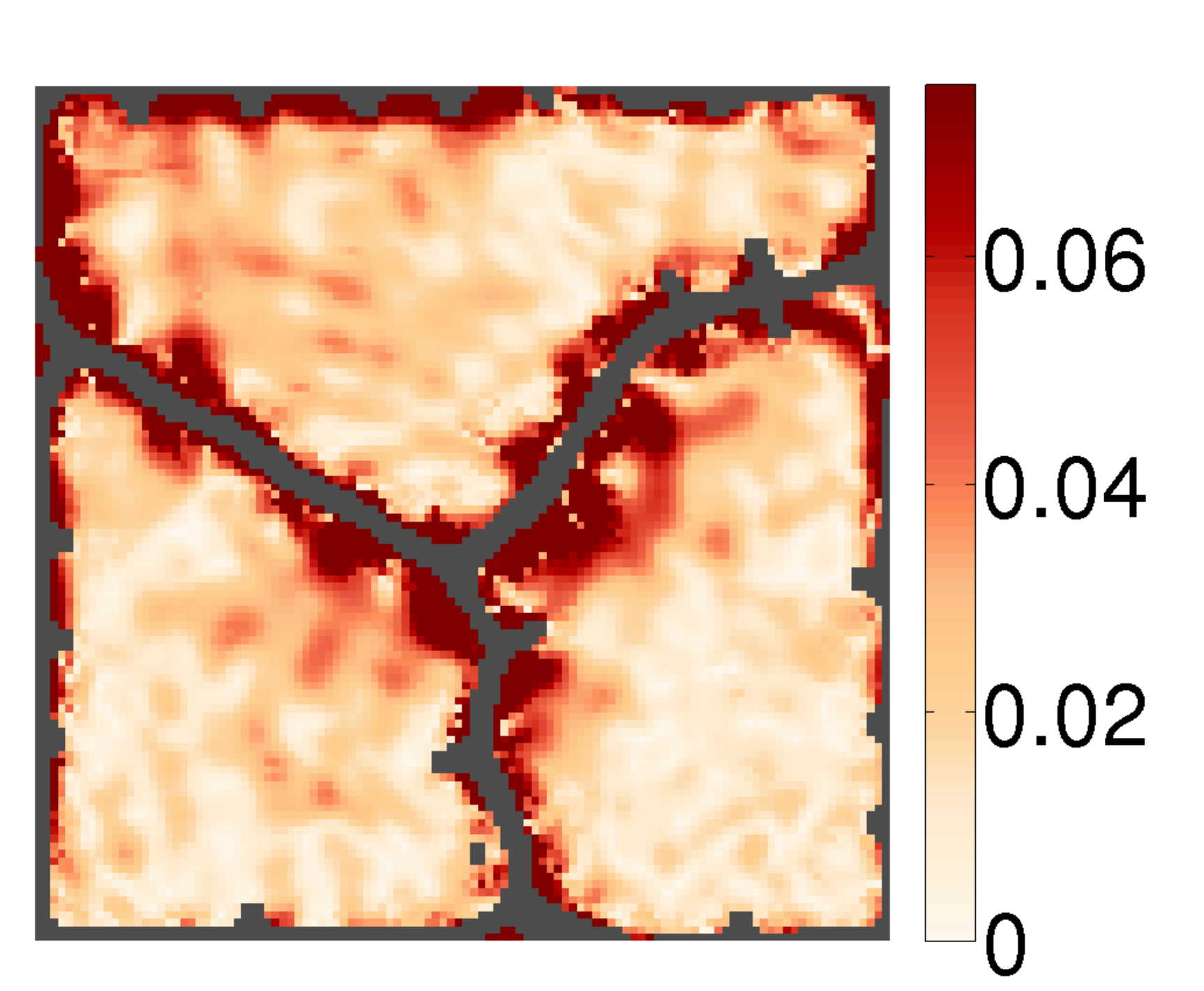} \\
         \includegraphics[height=1.2in]{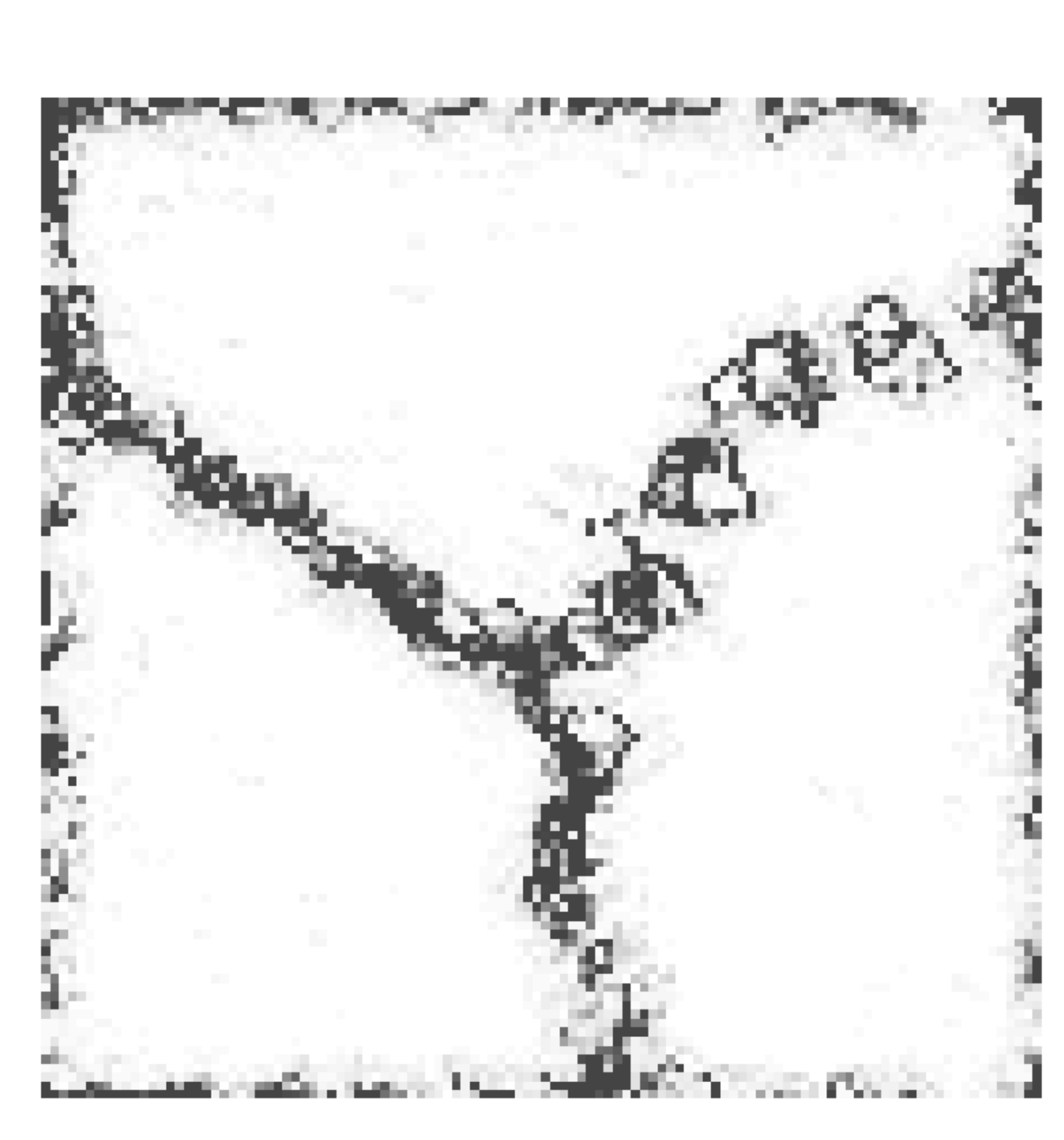}  \includegraphics[height=1.2in]{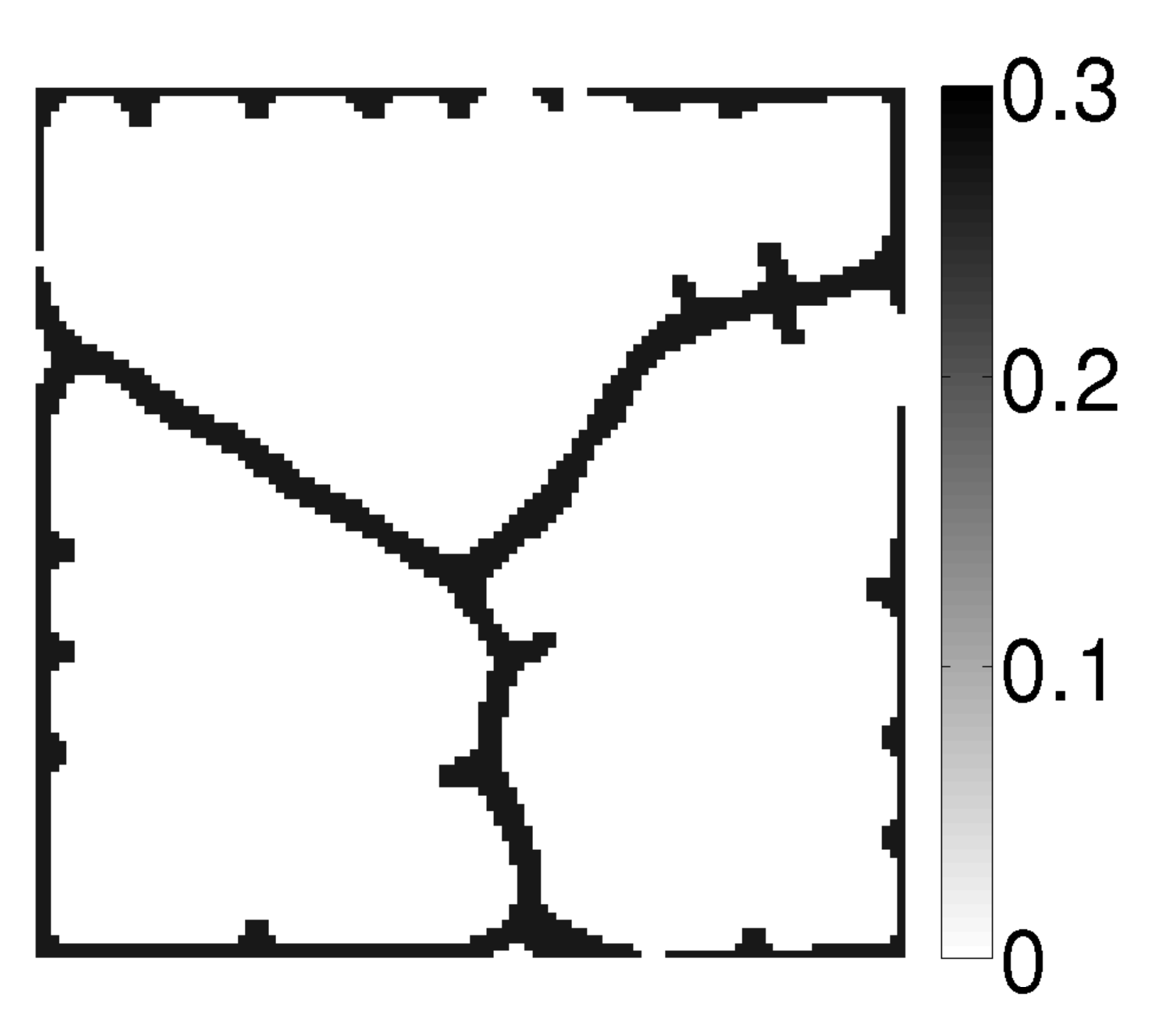} &
         \includegraphics[height=1.2in]{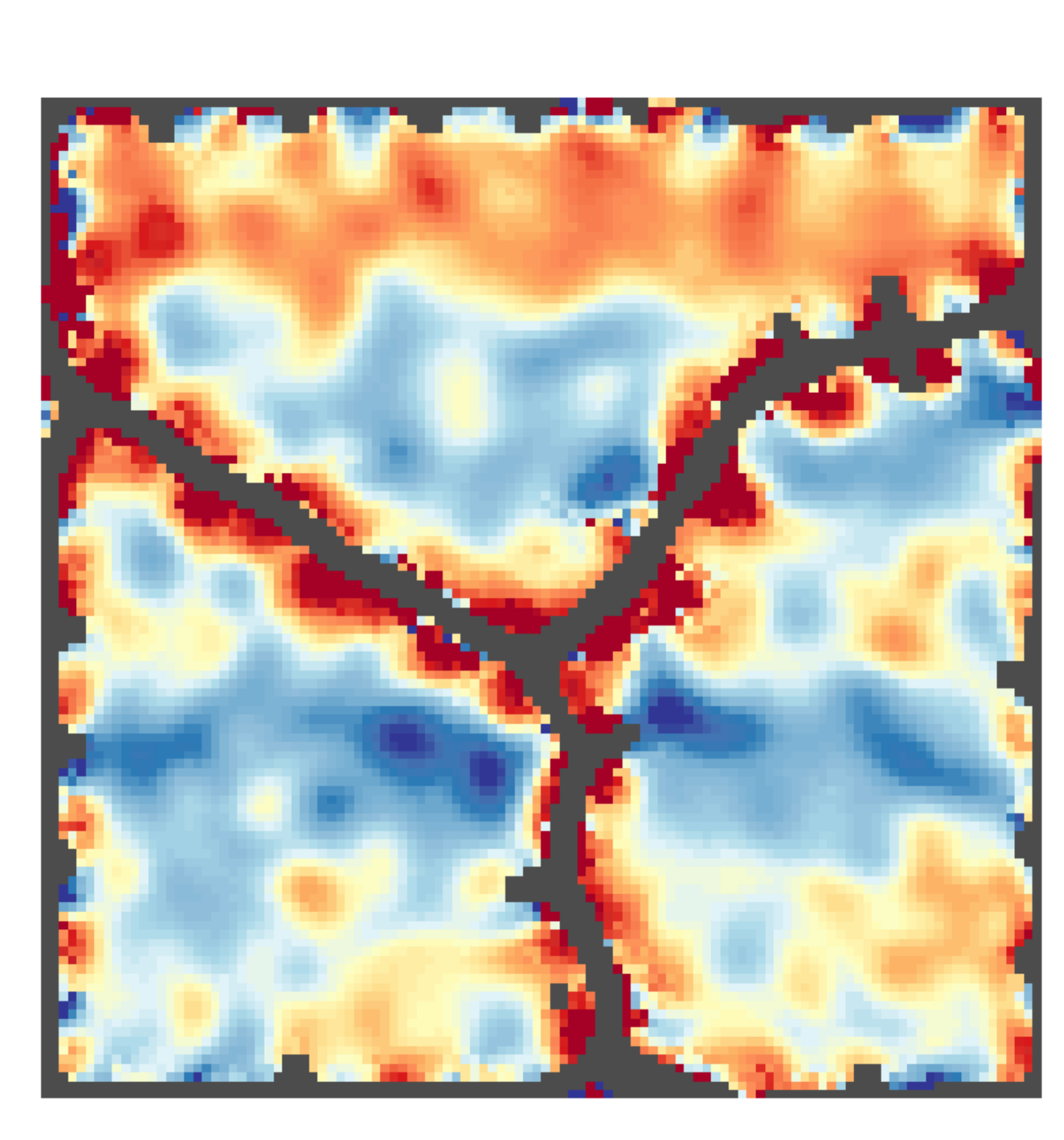}  \includegraphics[height=1.2in]{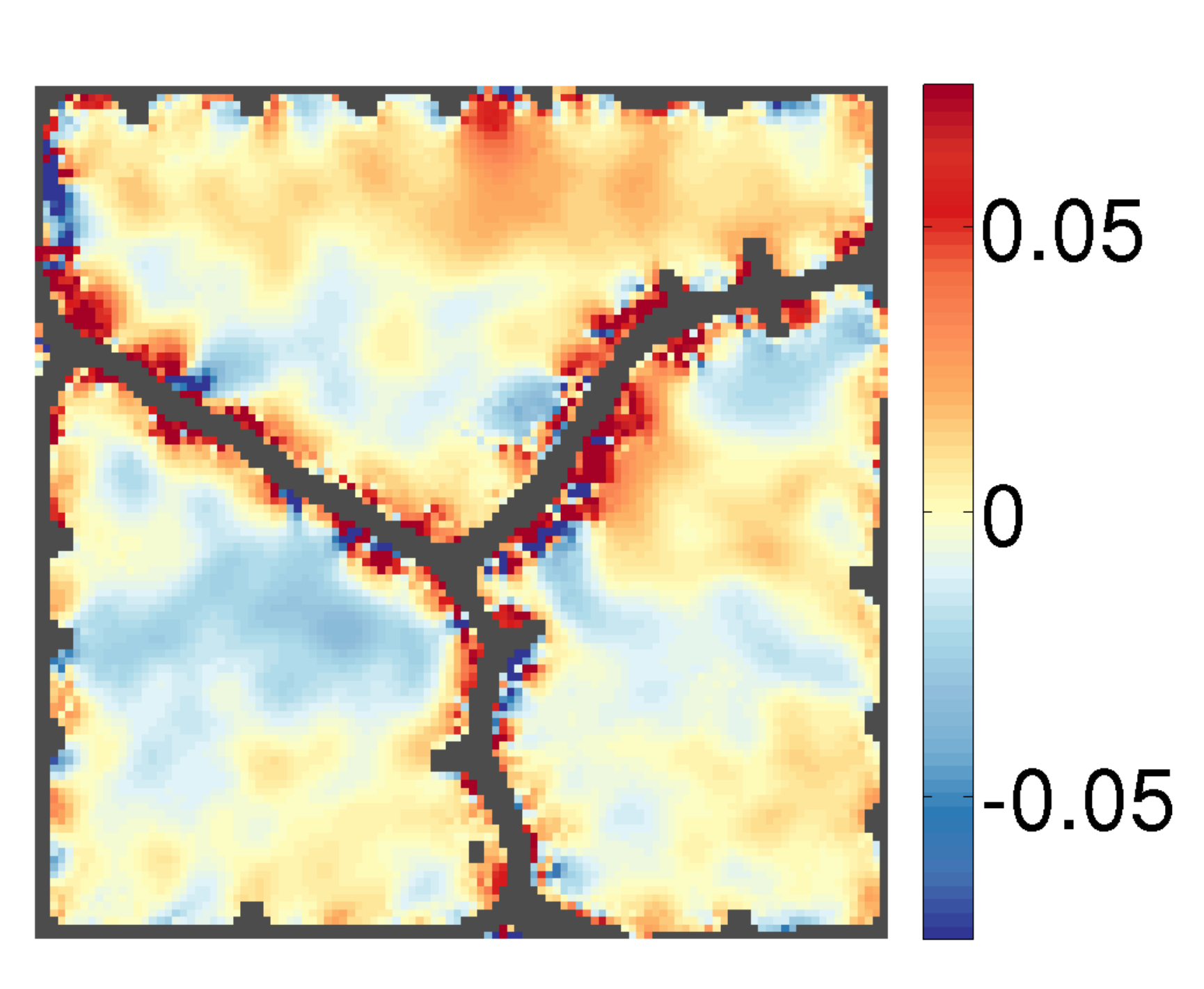} 
   \end{tabular}
\end{center}
\caption{Results of crystal analysis for Figure\,\ref{fig:Real1img} (right), using the same visualization as in Figure\,\ref{fig:PFC1}.}
\label{fig:BB}
\end{figure}

The second example is a photograph of a bubble raft with strongly disordered and blurry grain boundaries (Figure\,\ref{fig:Real1img} right). The image size is $223\times 223$ pixels. The SST result (see Figure\,\ref{fig:BB}) shows a spurious strong shear of the local crystal structure close to the grain boundaries, especially near the triple junction. One of the reasons for this behavior is that the SST, like any wavelet type transform, extracts directional and strain information from image patches, which here have to be larger than the unit cells. Thus, the grain boundaries are diffused, and information near the grain boundary is not trustworthy. The optimization can mostly remedies this effect.

The last experimental example is a TEM-image of a 
twin and a high angle grain boundary in Al (Figure~\ref{fig:Real2img}). The image size is $424\times 634$ pixels. The crystal structures in each grain are also slightly stretched. Although this example is very challenging, the SST-based analysis can still provide an accurate defect region estimate and a reasonable initial guess $G_0$ (see Figure\,\ref{fig:AI}). After optimization, we obtain a curl-free inverse deformation gradient $G$ in the grain interior. The difference of principal stretches becomes smaller and the volume distortion gets closer to zero outside the defect region.

\begin{figure}
  \begin{center}
     \begin{tabular}{c}
       \includegraphics[height=1.8in]{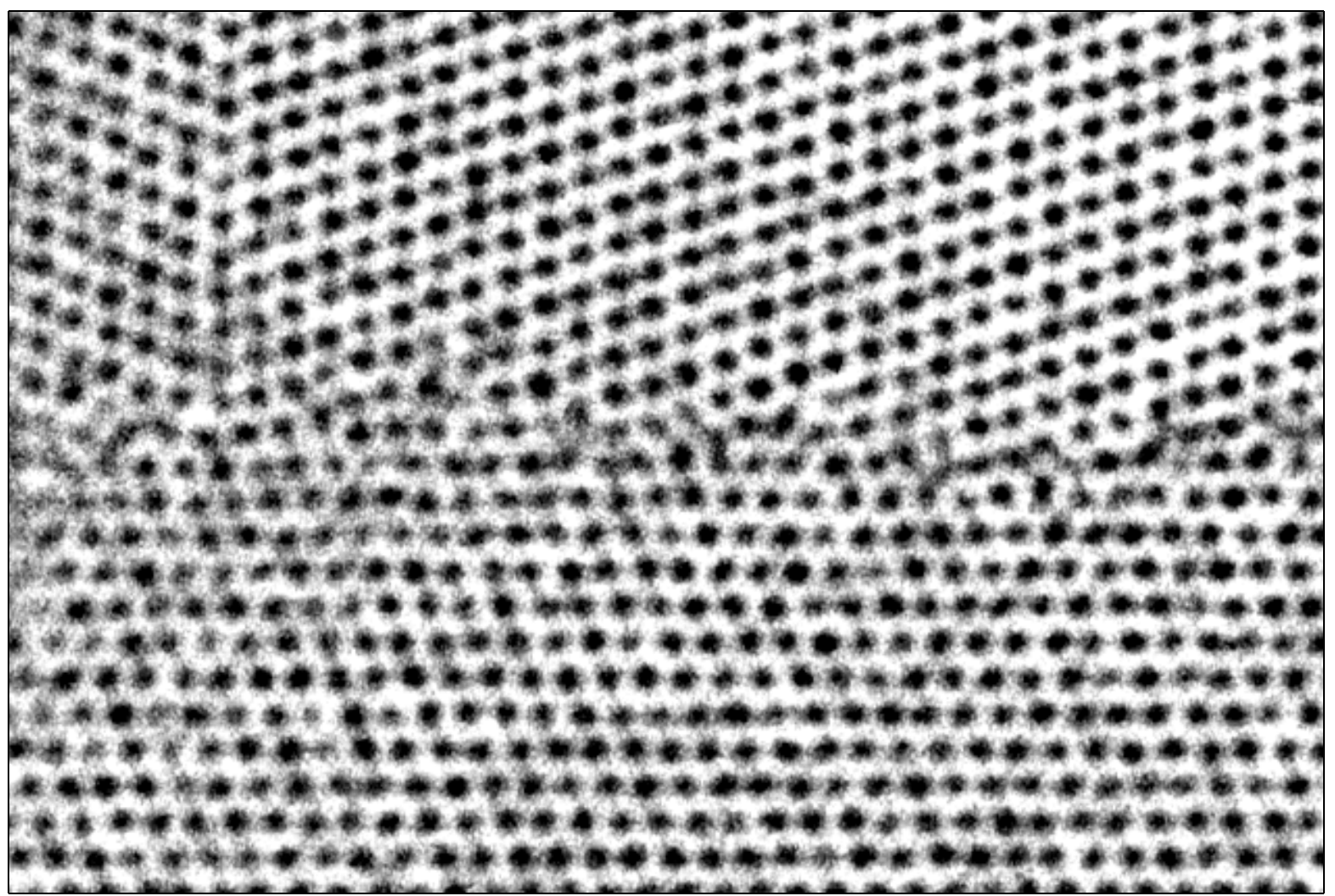}
     \end{tabular}
   \end{center}
 \caption{A TEM-image of Al (courtesy of the National Center for Electron Microscopy at the Lawrence Berkeley National Laboratory).}
 \label{fig:Real2img}
 \end{figure}
 
 \begin{figure}
  \begin{center}
  \begin{tabular}{cc}
        \includegraphics[height=0.9in]{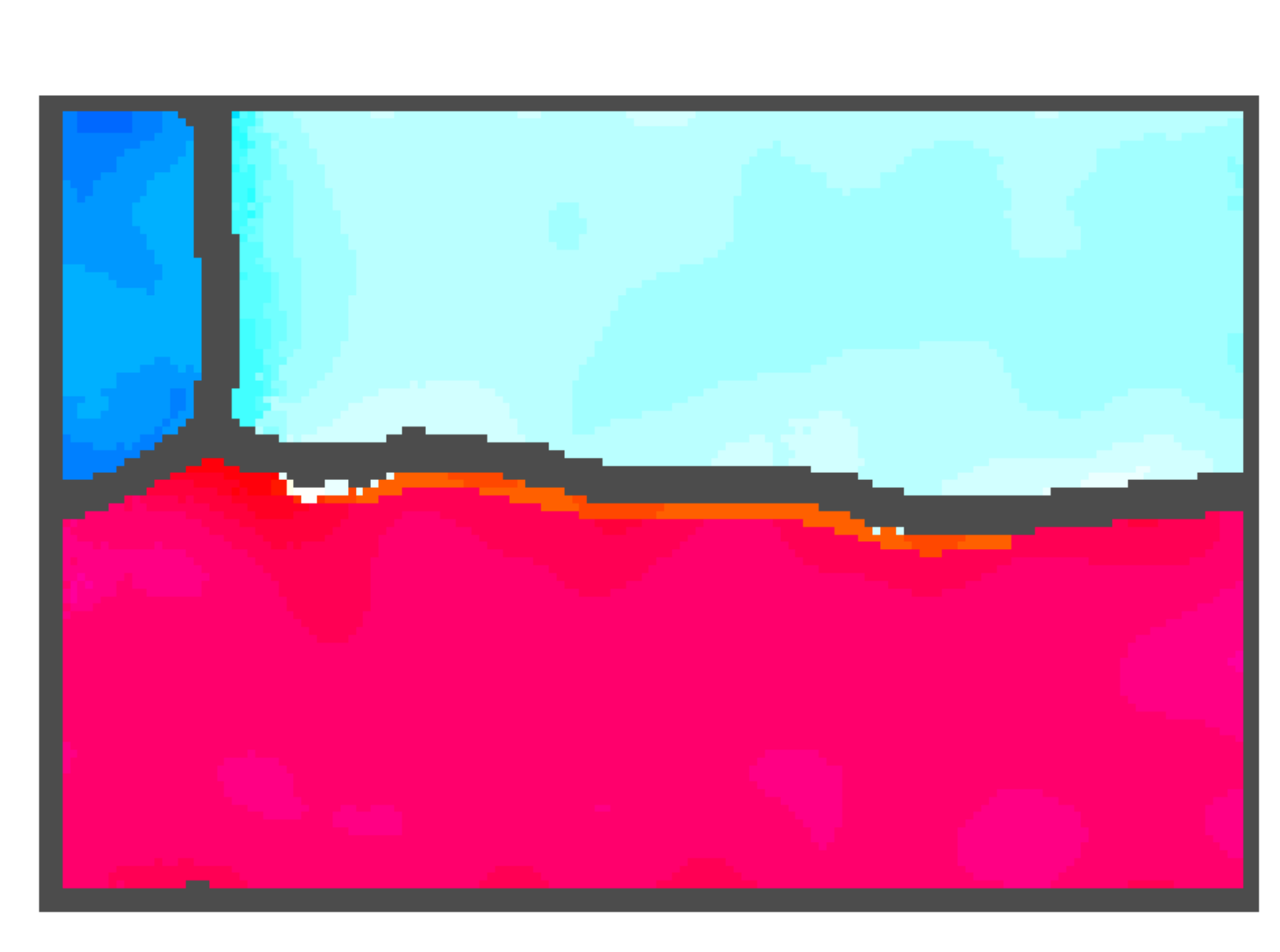} \includegraphics[height=0.9in]{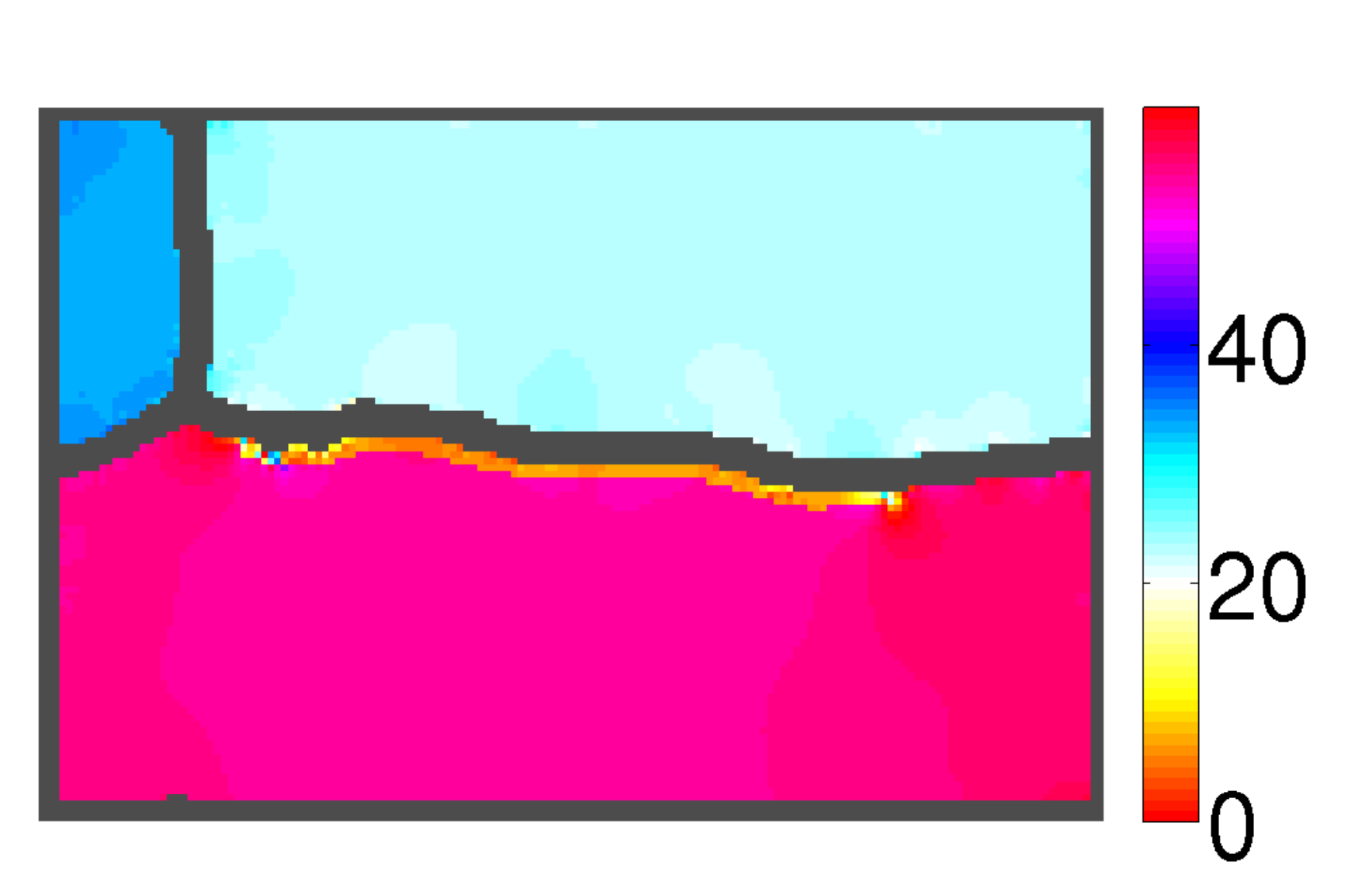} &
         \includegraphics[height=0.9in]{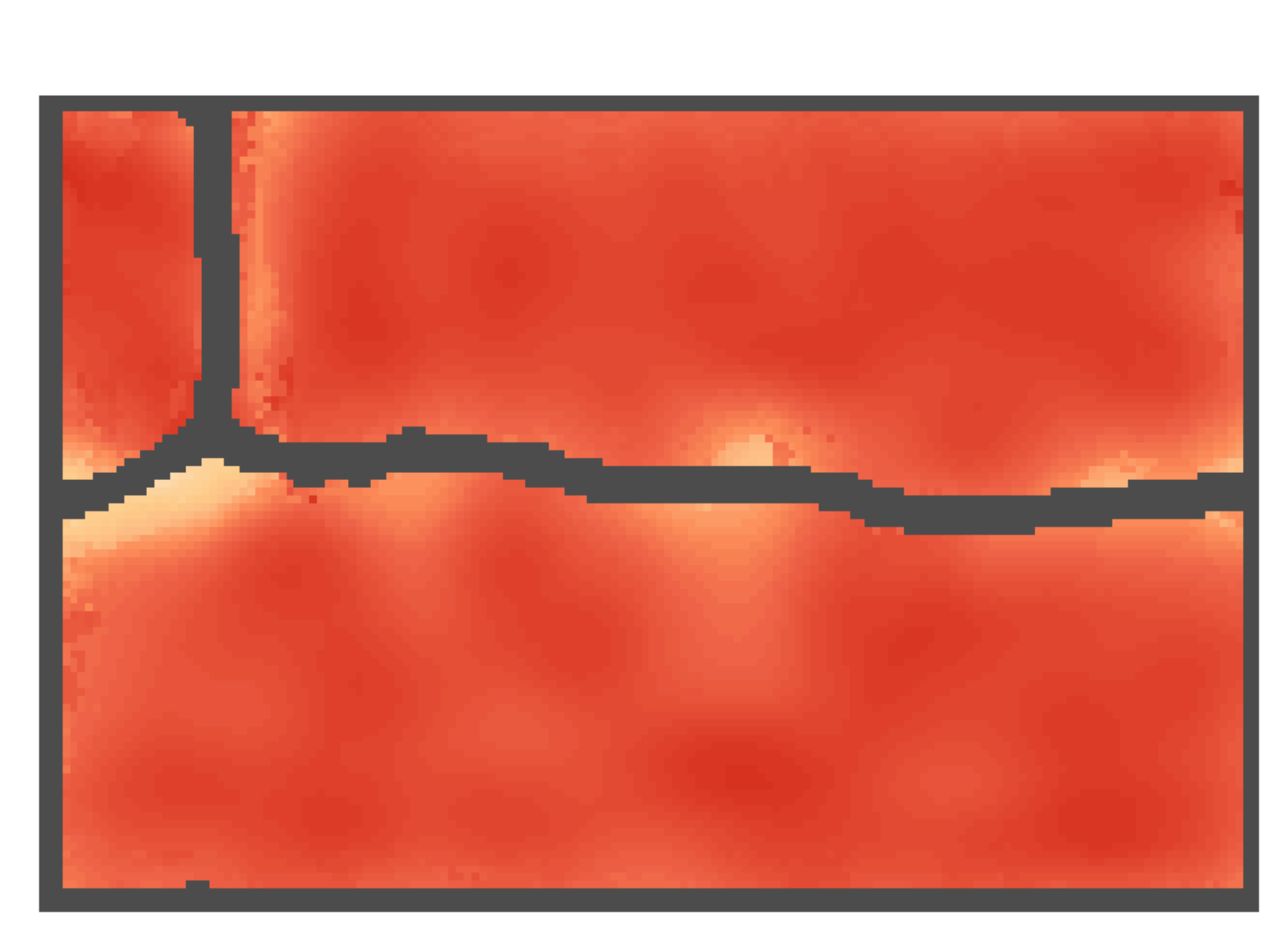} \includegraphics[height=0.9in]{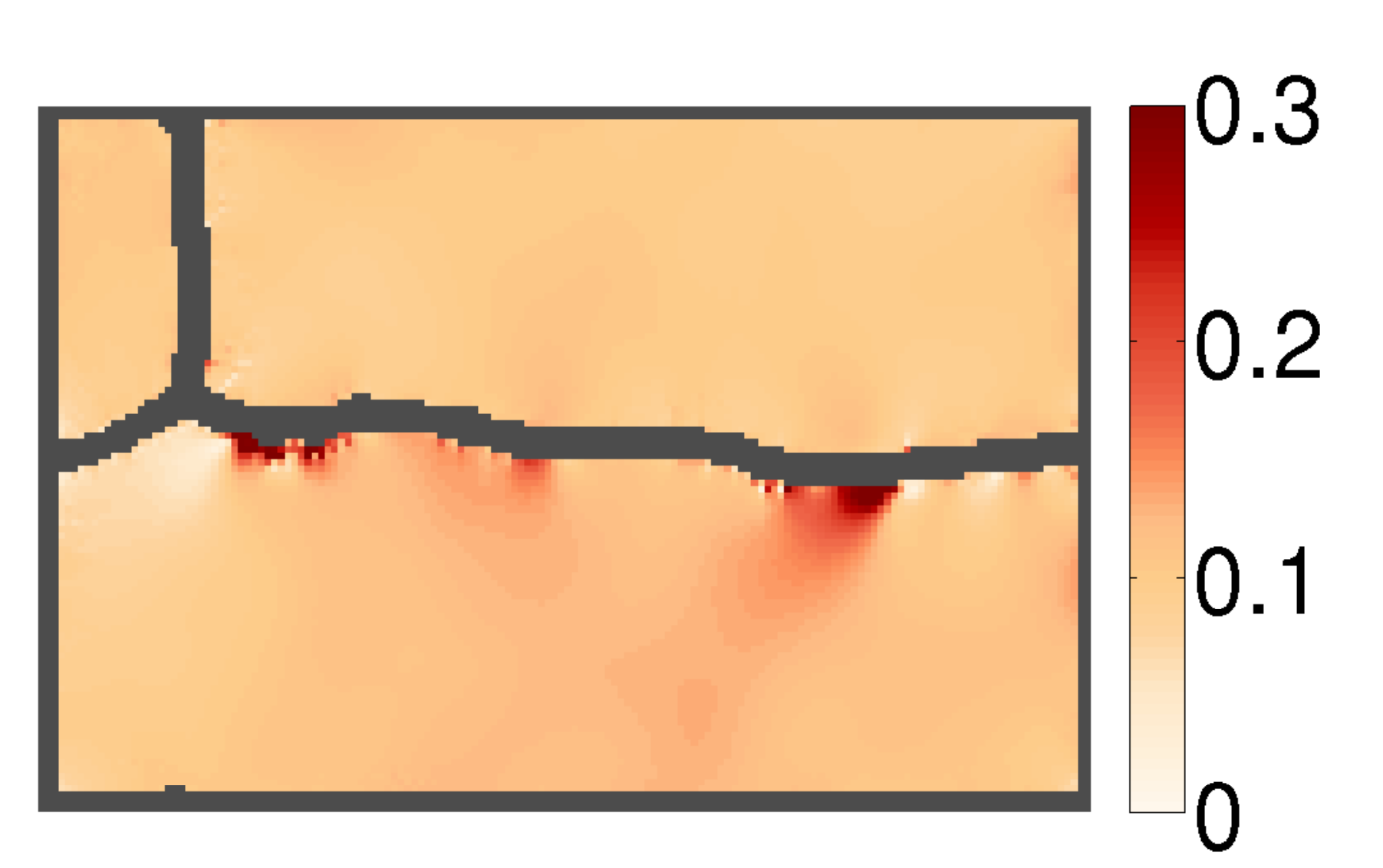} \\
         \includegraphics[height=0.9in]{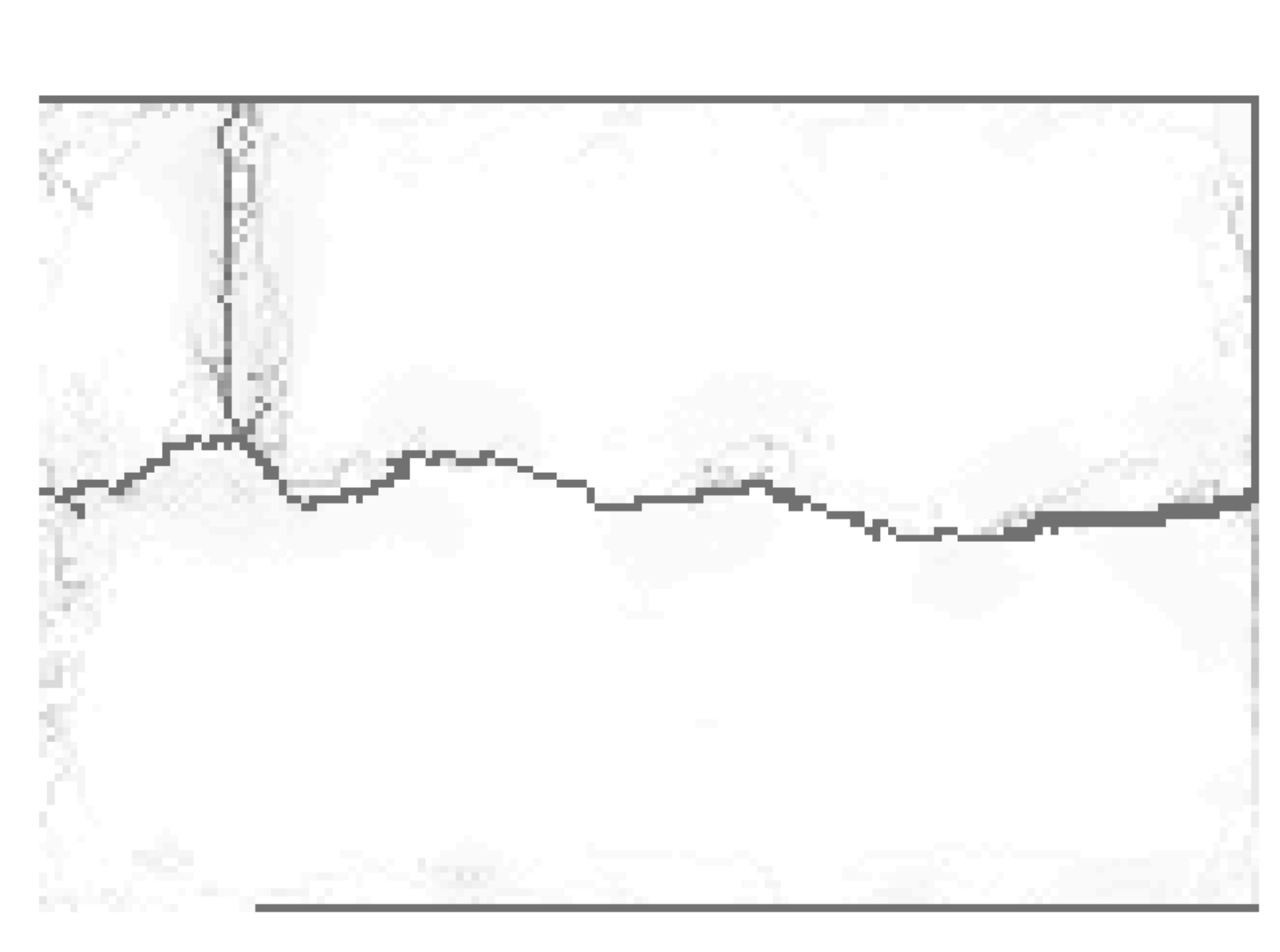}  \includegraphics[height=0.9in]{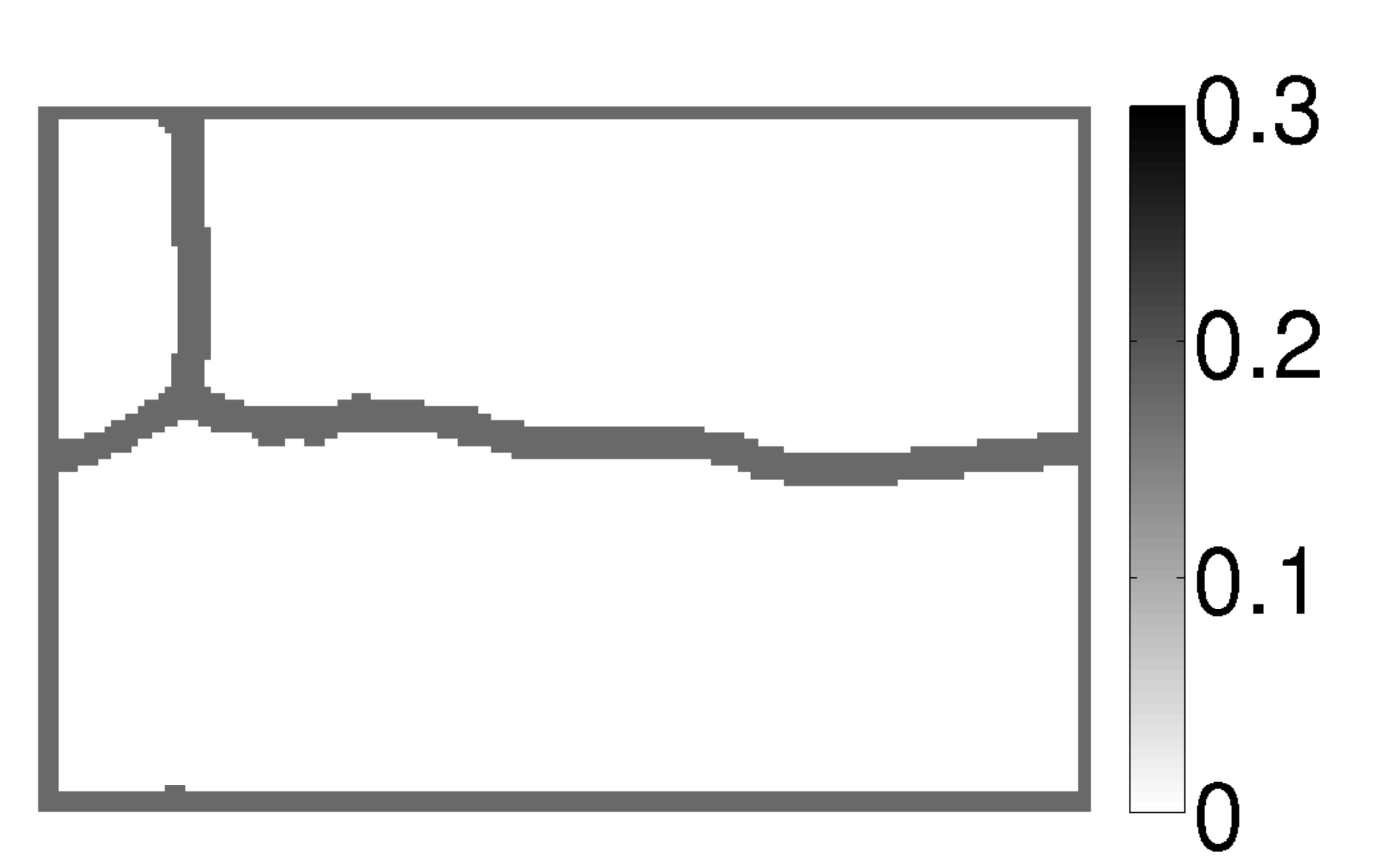} &
         \includegraphics[height=0.9in]{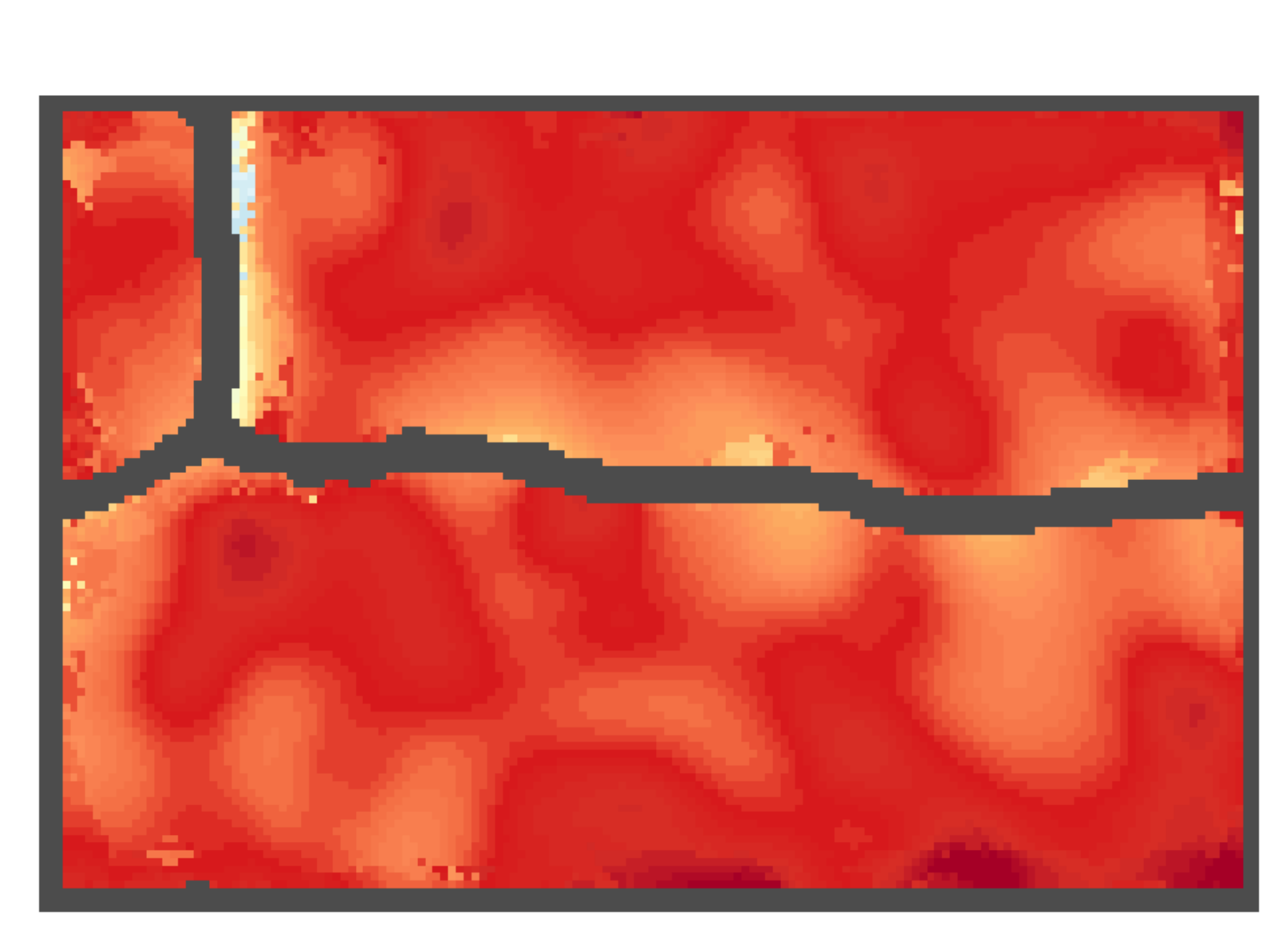}  \includegraphics[height=0.9in]{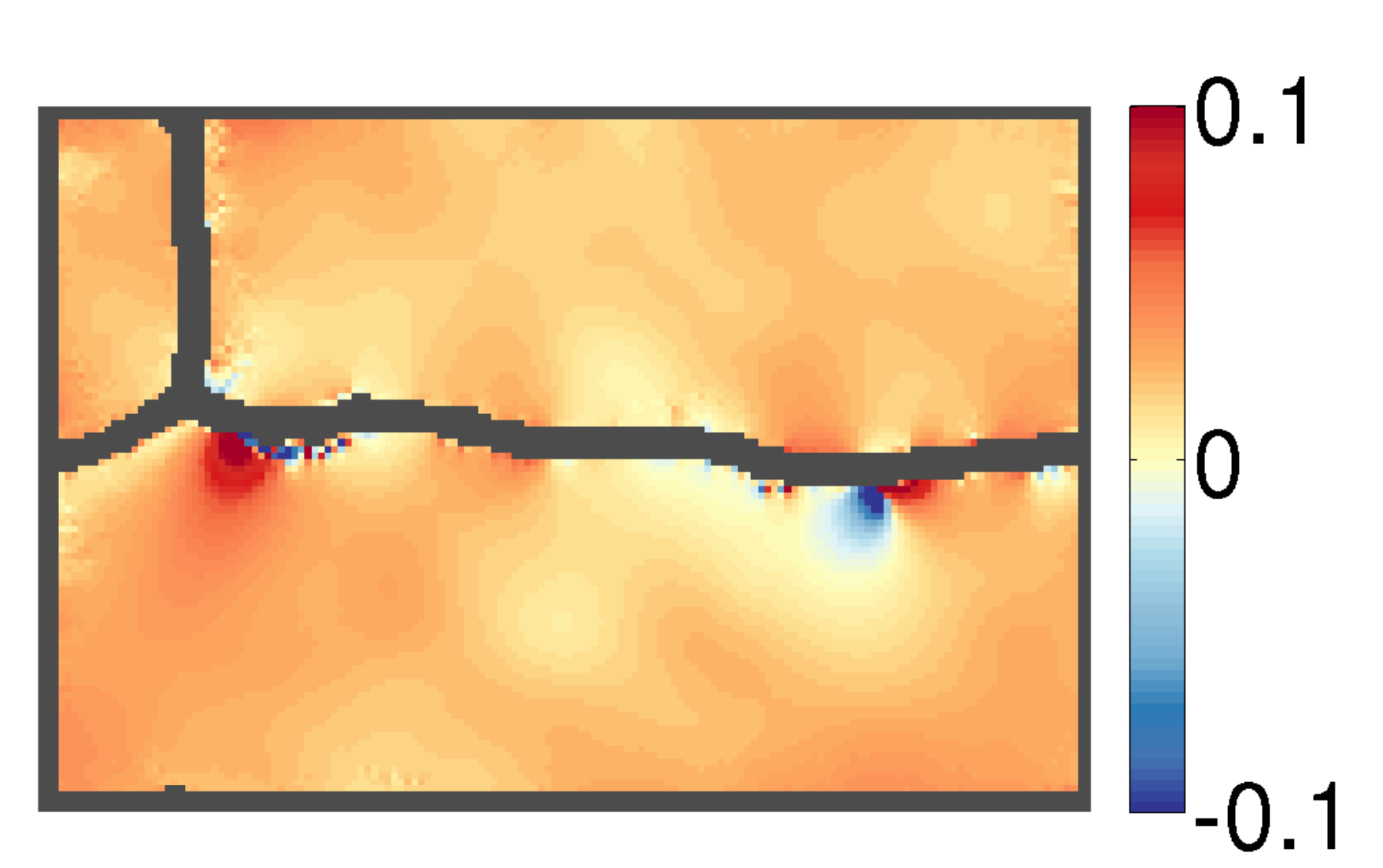} 
     \end{tabular}
  \end{center}
  \caption{Results of crystal analysis for Figure\,\ref{fig:Real2img}, using the same visualization as in Figure\,\ref{fig:PFC1}.}
  \label{fig:AI}
\end{figure}


\section{Discussion}

We have presented a variational method to improve an initial analysis of a crystal image provided by the SST-based method developed in \cite{SSCrystal}. By minimizing the elastic energy of the local inverse deformation gradient outside the SST-estimated defect region, we obtain an optimized deformation gradient $G$ that better agrees with physical understanding. The desired information about the local crystal state is immediately available from $G$: the polar decomposition of $G$ gives crystal orientations; the determinant of $G$ indicates volume dilation and compression; the curl of $G$ indicates defect regions, etc. 

The variational method in this paper presents a first step towards regularizing the local inverse deformation gradient provided by SST. In this step, we only focus on the optimization outside defect regions. In a next step, one has to investigate variational models to also obtain a meaningful estimate of the deformation gradient inside defect regions, which would provide a more accurate localization of defects. One possible solution could be a nested optimization model with an inner and an outer loop, where the inner loop performs the optimization outside the defect regions as proposed in this paper. The outer loop then tries to improve $\curl G$ inside the defect regions. 
We will leave this to future works. 

\section*{Acknowledgment}

The work of J.L. was supported in part by the Alfred P.~Sloan
foundation and National Science Foundation under award DMS-1312659. 
B.W.'s research was supported by the Alfried Krupp Prize for Young University Teachers awarded by the Alfried Krupp von Bohlen und Halbach-Stiftung.
H.Y. was partially supported by Lexing Ying's grants: the National Science Foundation under award DMS-1328230 and the U.S. Department of Energys Advanced Scientific Computing Research program under award DE-FC02-13ER26134/DE-SC0009409. 
The authors thank Lexing Ying for helpful suggestions and comments.

\bibliographystyle{abbrv}
\bibliography{synsquezvar}

\end{document}